\documentclass[preprint,1p,number,sort&compress]{elsarticle}

\PassOptionsToPackage{pagebackref=true,linktoc=page,colorlinks=true}{hyperref}
\PassOptionsToPackage{utf8}{inputenc}
\PassOptionsToPackage{shortlabels}{enumitem}

\usepackage{graphicx}
\usepackage{inputenc}
\usepackage{hyperref}
\usepackage{natbib}
\usepackage{todonotes}
\usepackage{xcolor}
\usepackage{mathptmx}      
\usepackage{paralist}
\usepackage{amsmath}
\usepackage{amssymb}
\usepackage{booktabs}
\usepackage{cleveref}
\usepackage{verbatim}
\usepackage{subcaption}
\usepackage{bbding}
\usepackage{numprint}
\usepackage{collectbox}
\npdecimalsign{.}

\usepackage{xspace}
\usepackage{paralist}
\usepackage{lastpage}
\usepackage{enumitem}

\usepackage{framed} 
\definecolor{shadecolor}{rgb}{.67,0.85,0.90}

\usepackage{tikz}
\usetikzlibrary{shapes,positioning,arrows,shadows}

\usepackage[binary-units = true]{siunitx} 

\usepackage{nicefrac}	

\setdefaultenum{(i)}{}{}{}

\renewcommand*{\backref}[1]{}
\renewcommand*{\backrefalt}[4]{%
    \ifcase #1 %
    \relax 
    \or
    (page #2).%
    \else
    (pages #2).%
    \fi%
}

\usepackage{fancyhdr}
\newcommand{\myTitle}{A Systematic Approach to Constructing Incremental Topology Control Algorithms Using Graph Transformation}

\newenvironment{relaxedbox}{\begin{shaded*}}{\end{shaded*}}

\newcommand{\enquote}[1]{``#1''\xspace}
\newcommand{\wrt}{w.r.t.\xspace}
\newcommand{\eg}{e.g.\xspace}
\newcommand{\idest}{i.e.\xspace}
\newcommand{\vs}{vs.\xspace}

\newcommand{\refToFig}[1]{\textsf{#1}}
\newcommand{\approximately}{approx.\xspace}

\let\noarg\relax

\DeclareMathOperator{\state}{s}
\DeclareMathOperator{\weight}{w}

\newcommand{\weightOf}[1]{\ensuremath{\weight(#1)}\xspace}
\newcommand{\ktc}[0]{kTC\xspace}
\newcommand{\iktc}[0]{i{-}kTC\xspace}
\newcommand{\bktc}[0]{b{-}kTC\xspace}
\newcommand{\TC}[0]{TC\xspace}
\newcommand{\WSNs}[0]{WSNs\xspace}
\newcommand{\WSN}[0]{WSN\xspace}
\newcommand{\CE}[0]{context event\xspace}
\newcommand{\CEs}[0]{context events\xspace}

\newcommand{\lvE}{\ensuremath{e}}

\newcommand{\nodeName}[1]{\ensuremath{n_{\text{#1}}\xspace}}
\newcommand{\nodeNameLong}[1]{node~\nodeName{#1}\xspace}

\newcommand{\linkName}[1]{\ensuremath{e_{#1}\xspace}}
\newcommand{\linkNameLong}[1]{link~\linkName{#1}}

\newcommand{\ACT}[0]{\texttt{Active}\xspace}
\newcommand{\INACT}[0]{\texttt{Inactive}\xspace}
\newcommand{\UNCL}[0]{\texttt{Unclassified}\xspace}

\newcommand{\constraint}[1]{\ensuremath{C_{#1}}\xspace}
\newcommand{\constraintLong}[1]{constraint~\constraint{y}\xspace}

\newcommand{\inactiveLinkConstraintKTC}[0]{\ensuremath{C_{\textrm{i}}}\xspace}
\newcommand{\inactiveLinkConstraintKTCLong}[0]{in\-ac\-tive-link constraint \inactiveLinkConstraintKTC}

\newcommand{\activeLinkConstraintKTC}[0]{\ensuremath{C_{\textrm{a}}}\xspace}
\newcommand{\activeLinkConstraintKTCLong}[0]{ac\-tive-link constraint \activeLinkConstraintKTC}

\newcommand{\unclassifiedLinkConstraint}[0]{\ensuremath{C_{\textrm{u}}}\xspace}
\newcommand{\unclassifiedLinkConstraintLong}[0]{un\-clas\-si\-fied-link constraint \unclassifiedLinkConstraint}

\newcommand{\noParallelLinksConstraint}[0]{\ensuremath{C_{\text{no-par-links}}}\xspace}
\newcommand{\noParallelLinksConstraintLong}[0]{no-par\-al\-lel-links constraint~\noParallelLinksConstraint}

\newcommand{\noLoopsConstraint}[0]{\ensuremath{C_{\text{no-loops}}}\xspace}
\newcommand{\noLoopsConstraintLong}[0]{no-loops constraint~\noLoopsConstraint}

\newcommand{\premiseNoParallelLinks}[0]{\ensuremath{p_{\text{no-par-links}}}\xspace}
\newcommand{\premiseNoLoops}[0]{\ensuremath{p_{\text{no-loops}}}\xspace}
\newcommand{\premiseUnclassifiedLink}[0]{\ensuremath{p_{\text{u}}}\xspace}

\newcommand{\conclusionDI}[0]{\ensuremath{c_{\textrm{i}}}\xspace}

\newcommand{\physicalConnConstraint}{\ensuremath{\hat{C}_\text{p-conn}\xspace}}
\newcommand{\weakConnConstraint}{\ensuremath{\hat{C}_\text{w-conn}\xspace}}
\newcommand{\strongConnConstraint}{\ensuremath{\hat{C}_\text{s-conn}\xspace}}

\newcommand{\GT}[0]{GT\xspace}
\newcommand{\ACs}[0]{ACs\xspace}
\newcommand{\AC}[1]{\ensuremath{\text{AC}_{#1}}\xspace}
\newcommand{\ACLong}[1]{application condition~\AC{#1}}
\newcommand{\NAC}[1]{\ensuremath{\text{NAC}_{#1}}\xspace}
\newcommand{\NACs}[0]{NACs\xspace}

\newcommand{\NACia}[1]{\ensuremath{\text{NAC}_{\textrm{i},\textrm{a},#1}}\xspace}
\newcommand{\NACaa}[1]{\ensuremath{\text{NAC}_{\textrm{a},\textrm{a},#1}}\xspace}
\newcommand{\NACui}[1]{\ensuremath{\text{NAC}_{\textrm{u},\textrm{i},#1}}\xspace}
\newcommand{\PACii}[1]{\ensuremath{\text{PAC}_{\textrm{i},\textrm{i},#1}}\xspace}
\newcommand{\PACai}[1]{\ensuremath{\text{PAC}_{\textrm{a},\textrm{i},#1}}\xspace}
\newcommand{\PACui}[1]{\ensuremath{\text{PAC}_{\textrm{u},\textrm{i},#1}}\xspace}
\newcommand{\PACmoddi}[1]{\ensuremath{\text{PAC}_{\textrm{mod-w},\textrm{i},#1}}}

\newcommand{\PAC}[1]{\ensuremath{\text{PAC}_{#1}}\xspace}
\newcommand{\PACs}[0]{PACs\xspace}
\newcommand{\LHS}[0]{LHS\xspace}
\newcommand{\RHS}[0]{RHS\xspace}

\newcommand{\guardSuccess}[0]{\refToFig{[Success]}\xspace}
\newcommand{\guardFailure}[0]{\refToFig{[Failure]}\xspace}

\newcommand{\GTrule}[1]{\ensuremath{R_{#1}}\xspace}
\newcommand{\GTruleLong}[1]{\GT rule~\GTrule{#1}}
\newcommand{\GTruleNN}[0]{\GT rule\xspace}
\newcommand{\GTrulesNN}[0]{\GT rules\xspace}
\newcommand{\GToperationNN}[0]{\GT operation\xspace}
\newcommand{\GToperationsNN}[0]{\GT operations\xspace}
\newcommand{\GToperation}[1]{\texttt{#1}\xspace}

\newcommand{\activationRule}[0]{\ensuremath{R_{\textrm{a}}}\xspace}
\newcommand{\activationRuleLong}[0]{activation rule~\activationRule}

\newcommand{\inactivationRule}[0]{\ensuremath{R_{\textrm{i}}}\xspace}
\newcommand{\inactivationRuleLong}[0]{inactivation rule~\inactivationRule}

\newcommand{\unclassificationRule}[0]{\ensuremath{R_{\textrm{u}}}\xspace}
\newcommand{\unclassificationRuleLong}[0]{un\-clas\-si\-fi\-ca\-tion rule~\unclassificationRule}

\newcommand{\nodeAdditionRule}[0]{\ensuremath{R_{\textrm{+n}}}\xspace}
\newcommand{\nodeAdditionRuleLong}[0]{node addition rule~\nodeAdditionRule}

\newcommand{\nodeRemovalRule}[0]{\ensuremath{R_{\textrm{-n}}}\xspace}
\newcommand{\nodeRemovalRuleLong}[0]{node removal rule~\nodeRemovalRule}

\newcommand{\linkAdditionRule}[0]{\ensuremath{R_{\textrm{+e}}}\xspace}
\newcommand{\linkAdditionRuleLong}[0]{link addition rule~\linkAdditionRule}

\newcommand{\linkRemovalRule}[0]{\ensuremath{R_{\textrm{-e}}}\xspace}
\newcommand{\linkRemovalRuleLong}[0]{link removal rule~\linkRemovalRule}

\newcommand{\weightModificationRule}[0]{\ensuremath{R_{\textrm{mod-w}}}\xspace}
\newcommand{\weightModificationRuleLong}[0]{weight modification rule~\weightModificationRule}

\newcommand{\varWnew}[0]{\ensuremath{w_{\text{new}}}\xspace}

\newcommand{\findUnclassifiedLinkRule}[0]{\ensuremath{R_{\textrm{find-u}}}\xspace}
\newcommand{\findUnclassifiedLinkRuleLong}[0]{un\-classi\-fied-link iden\-ti\-fi\-ca\-tion rule~\findUnclassifiedLinkRule}

\newcommand{\findClassifiedLinkRule}[0]{\ensuremath{R_{\textrm{find-ai}}}\xspace}
\newcommand{\findClassifiedLinkRuleLong}[0]{classi\-fied-link iden\-ti\-fi\-ca\-tion rule~\findClassifiedLinkRule}

\newcommand{\LSM}{link state modification\xspace}
\newcommand{\LSMs}{link state modifications\xspace}

\newcommand{\basicTCOperation}[0]{\GToperation{basic-tc}}
\newcommand{\tcOperation}[0]{\GToperation{ktc}}

\newcommand{\repairUOperation}[0]{\GToperation{repair-u}}
\newcommand{\handleUnclassification}[0]{\handlerOperation{\unclassificationRule}}

\newcommand{\leftPattern}[0]{\ensuremath{p_\ell}\xspace}
\newcommand{\rightPattern}[0]{\ensuremath{p_r}\xspace}
\newcommand{\leftMapping}[0]{\ensuremath{m_\ell}\xspace}
\newcommand{\rightMapping}[0]{\ensuremath{m_r}\xspace}
\newcommand{\gluing}[1]{\ensuremath{g_{#1}}\xspace}
\newcommand{\gluingLong}[1]{gluing~\gluing{#1}}
\newcommand{\gluingIA}[1]{\ensuremath{g_{\textrm{i},\textrm{a},#1}}\xspace}
\newcommand{\gluingII}[1]{\ensuremath{g_{\textrm{i},\textrm{i},#1}}\xspace}
\newcommand{\gluingAI}[1]{\ensuremath{g_{\textrm{a},\textrm{i},#1}}\xspace}
\newcommand{\gluingAA}[1]{\ensuremath{g_{\textrm{a},\textrm{a},#1}}\xspace}
\newcommand{\gluingUA}[1]{\ensuremath{g_{\textrm{u},\textrm{a},#1}}\xspace}
\newcommand{\gluingUILong}[1]{gluing~\gluingUI{#1}}
\newcommand{\gluingsUILong}[1]{gluings~\gluingUI{#1}}
\newcommand{\GluingUILong}[1]{Gluing~\gluingUI{#1}}
\newcommand{\gluingUI}[1]{\ensuremath{g_{\textrm{u},\textrm{i},#1}}\xspace}
\newcommand{\postcondition}[1]{\ensuremath{\text{PC}_{#1}}\xspace}
\newcommand{\postconditionAA}[1]{\ensuremath{\text{PC}_{\textrm{a},\textrm{a},#1}}\xspace}

\newcommand{\postconditionII}[1]{\ensuremath{\text{PC}_{\textrm{i},\textrm{i},#1}}\xspace}
\newcommand{\postconditionUI}[1]{\ensuremath{\text{PC}_{\textrm{u},\textrm{i},#1}}\xspace}

\newcommand{\premise}[2]{\ensuremath{p^{#2}_{#1}}\xspace}

\newcommand{\conclusion}[2]{\ensuremath{c{#2}_{#1}}\xspace}
\newcommand{\conclusionPC}[1]{\ensuremath{c^{\prime}_{#1}}\xspace}
\newcommand{\conclusionPCLong}[1]{conclusion~\conclusionPC{#1}}

\newcommand{\conclusionUI}[2]{\ensuremath{c{#2}_{\textrm{u},\textrm{i},#1}}\xspace}

\makeatletter
\newcommand{\refinementPairBox}{%
    \collectbox{%
        \setlength{\fboxsep}{1pt}%
        \fbox{\BOXCONTENT}%
    }%
}
\makeatother
\newcommand{\refinementPair}[2]{\refinementPairBox{#1 + #2}}
\newcommand{\refinementStep}[1]{step~(#1)}

\newcommand{\nodeVariableA}[0]{\ensuremath{\text{A}}\xspace}
\newcommand{\nodeVariableB}[0]{\ensuremath{\text{B}}\xspace}
\newcommand{\nodeVariableC}[0]{\ensuremath{\text{C}}\xspace}
\newcommand{\nodeVariableD}[0]{\ensuremath{\text{D}}\xspace}
\newcommand{\nodeVariablea}[0]{\ensuremath{\text{a}}\xspace}

\newcommand{\nodeVariablec}[0]{\ensuremath{\text{c}}\xspace}
\newcommand{\nodeVariableOne}[0]{\ensuremath{1}\xspace}

\newcommand{\linkVariableAB}[0]{\linkName{\text{AB}}\xspace}
\newcommand{\linkVariableab}[0]{\linkName{\text{ab}}\xspace}
\newcommand{\linkVariableac}[0]{\linkName{\text{ac}}\xspace}
\newcommand{\linkVariablecb}[0]{\linkName{\text{cb}}\xspace}
\newcommand{\linkVariableOneTwo}[0]{\linkName{\text{12}}\xspace}
\newcommand{\linkVariableOneThree}[0]{\linkName{\text{13}}\xspace}

\newcommand{\handlerOperation}[1]{\GToperation{handle-#1}}
\newcommand{\handlerOperationLong}[1]{handler operation \handlerOperation{#1}}

\newcommand{\restorationOperation}[1]{\GToperation{restore-#1}}

\newcommand{\violationIdentificationOperation}[1]{\GToperation{find-vi\-o\-la\-tion-#1}}
\newcommand{\violationIdentificationRule}[1]{\GTrule{\text{find-violation-}#1}}

\newcommand{\violationResolutionOperation}[1]{\GToperation{resolve-violation-#1}}

\newcommand{\eMoflon}[0]{\textsc{eMoflon}\xspace}
\newcommand{\Simonstrator}[0]{\textsc{Simonstrator}\xspace}
\newcommand{\PFS}[0]{\PFSLong}
\newcommand{\PFSLong}[0]{\textsc{PeerfactSim.KOM}\xspace}

\newcommand{\RQCorrectness}[0]{RQ1\xspace}
\newcommand{\RQCorrectnessLong}[0]{\RQCorrectness--Correctness\xspace}

\newcommand{\RQIncrementality}[0]{RQ2\xspace}
\newcommand{\RQIncrementalityLong}[0]{\RQIncrementality--Incrementality\xspace}

\newcommand{\RQPerformance}[0]{RQ3\xspace}
\newcommand{\RQPerformanceLong}[0]{\RQPerformance--Performance\xspace}

\newcommand{\RQGeneralizability}[0]{RQ4\xspace}
\newcommand{\RQGeneralizabilityLong}[0]{\RQGeneralizability--Generalizability\xspace}

\newcommand{\OK}[0]{\checkmark} 
\newcommand{\notOK}{\XSolidBrush}
\newtheorem{theorem}{Theorem}
\crefname{theorem}{Theorem}{Theorems}

\crefname{lemma}{Lemma}{Lemmata}
\newtheorem{corollary}[theorem]{Corollary}
\crefname{corollary}{Corollary}{Corollaries}
\newenvironment{proof-sketch}{\noindent{\emph{Sketch of proof.}}}{\qed}
\newenvironment{proof}{\noindent{\emph{Proof.}}}{\qed}
\newenvironment{myinparaenum}{\begin{inparaenum}[(i)]}{\end{inparaenum}}

\newcommand{\rootOfEvaluationDataset}[0]{figures/evaluation2/batchrun_2016-05-07T192839}
\newcommand{\rootOfConfiguration}[2]{\rootOfEvaluationDataset/run_b130-pexp2.0__k1.41__n#1__ws#2__ID_KTC/output/}
\newcommand{\evalParSeedCount}{\numprint{15}\xspace}
\newcommand{\evalParK}{\numprint{1.41}\xspace}
\newcommand{\evalParTCInterval}{\SI{10}{\minute}\xspace}
\newcommand{\evalParSimulationDuration}{\SI{20}{\hour}\xspace}
\newcommand{\evalParBatterySource}{\SI{130}{\joule}\xspace}
\newcommand{\evalParBatteryTarget}{\SI{100}{\kilo\joule}\xspace}

\newcommand{\evalParMeasurementsPerRun}{\numprint{119}\xspace}
\newcommand{\evalParHesitation}{\SI{99}{\percent}\xspace}

\def\cfgSmallDense{\texttt{n100w250}\xspace}
\def\cfgSmallMedium{\texttt{n100w500}\xspace}
\def\cfgSmallSparse{\texttt{n100w750}\xspace}
\def\cfgLargeDense{\texttt{n1000w1000}\xspace}
\def\cfgLargeMedium{\texttt{n1000w1500}\xspace}
\def\cfgLargeSparse{\texttt{n1000w2000}\xspace}
\newcommand{\metricIkTCVsBkTCLong}[0]{\iktc-to-\bktc \LSM ratio\xspace}
\newcommand{\metricScopeLong}[0]{scope\xspace}
\newcommand{\metricDegreeNormalizedScopeLong}{degree-normalized scope\xspace}
\newcommand{\metricAliveNodeCountLong}[0]{number of alive nodes\xspace}
\usetikzlibrary{positioning, shadows, shapes, arrows}

\tikzset{main node/.style={circle, draw, fill=black, text=white,font=\sffamily\scriptsize, minimum size = 0.4cm, inner sep= 1pt}}
\tikzset{active/.style={->, line width=0.9pt, font=\scriptsize}}
\tikzset{inactive/.style={->, densely dotted, line width=0.9pt, font=\scriptsize}}
\tikzset{unclassified/.style={->, densely dashed,line width=0.9pt, font=\scriptsize}}
\tikzset{activeOrInactive/.style={->, dashdotted,line width=0.9pt, font=\scriptsize}}
\tikzset{any/.style={->, draw={rgb:black,5;white,6}, line width=0.9pt, font=\scriptsize}}
\tikzset{induction/.style={->, >=latex, line width=0.9pt, font=\scriptsize}}

\tikzset{rcBig/.style={rectangle,draw, fill=white, minimum height=5.25cm,text width=4.4cm, align = center}}
\tikzset{rcGraphSmall/.style={rectangle,draw, fill=white, minimum height=1.1cm,text width=4cm, align = center}}
\tikzset{rcGraph/.style={rectangle,draw, fill=white, minimum height=3.25cm,text width=4cm, align = center}}
\tikzset{rcGraphInner/.style={rectangle,draw, fill=white, minimum height=0.705cm,text width=3.55cm, align = left}}
\tikzset{rcMini/.style={rectangle,draw, fill=white, minimum height=0.1cm,text width=0.55cm, align = left}}
\tikzset{rcMiniMedium/.style={rectangle,draw, fill=white, minimum height=0.1cm,text width=0.95cm, align = left}}
\tikzset{rcMiniLong/.style={rectangle,draw, fill=white, minimum height=0.1cm,text width=1.25cm, align = left}}

\tikzset{rc/.style={rectangle,draw, fill=white, minimum height=2.65cm,text width=3.75cm, align = center}}
\tikzset{rcSmall/.style={rectangle,draw, fill=white, minimum height=0.705cm,text width=3.35cm, align = center}}
\tikzset{bigArrow/.style={fill={rgb:black,5;white,6},shape=single arrow,text width=0.75cm,text height=2.5ex}}

\tikzset{rcMedium/.style={rectangle,draw, fill=white, minimum height=3.5cm,text width=3cm, align = center}}
\tikzset{rcGraphMedium/.style={rectangle,draw, fill=white, minimum height=1.9cm,text width=2.6cm, align = center}}

\tikzset{rcShort/.style={rectangle,draw, fill=white, minimum height=1.7cm,text width=2.25cm, align = center}}
\tikzset{rcShortDouble/.style={rectangle,draw, fill=white, minimum height=1.7cm,text width=4.35cm, align = center}}
\tikzset{rcShortGraph/.style={rectangle,draw, fill=white, minimum height=1.1cm,text width=1.9cm, align = center}}
\tikzset{rcLongGraph/.style={rectangle,draw, fill=white, minimum height=2.3cm,text width=1.9cm, align = center}}

\newdefinition{my-definition}[theorem]{Definition}

\begin{document}

\title{\myTitle}

\author[es]{Roland Kluge\corref{cor1}}
\ead{roland.kluge@es.tu-darmstadt.de}
\author[tk]{Michael Stein}
\ead{michael.stein@tk.informatik.tu-darmstadt.de}
\author[es]{Gergely Varró}
\ead{gergely.varro@es.tu-darmstadt.de}
\author[es]{Andy Schürr\corref{cor2}}
\ead{andy.schuerr@es.tu-darmstadt.de}
\author[seemoo]{Matthias Hollick}
\ead{matthias.hollick@seemoo.tu-darmstadt.de}
\author[tk]{Max Mühlhäuser}
\ead{max@informatik.tu-darmstadt.de}

\address[es]{Real-Time Systems Lab, Merckstr. 25, 64283 Darmstadt, Germany}
\address[tk]{Telecooperation Group, Hochschulstr. 10, 64289 Darmstadt, Germany}
\address[seemoo]{Secure Mobile Networking Lab, Mornewegstr. 32, 64293 Darmstadt, Germany}

\cortext[cor1]{Corresponding author}
\cortext[cor2]{Principal corresponding author}

\begin{abstract}
\begin{relaxedbox}
\noindent This document corresponds to the accepted manuscript of the article
\emph{Kluge, R., Stein, M., Varró, G., Schürr, A., Hollick, M., Mühlhäuser, M.: "A Systematic Approach to Constructing Incremental Topology Control Algorithms using Graph Transformation," in: JVLC 2016.}
The URL of the formal version is \url{https://dx.doi.org/10.1016/j.jvlc.2016.10.003}.
This document is made available under the CC-BY-NC-ND 4.0 license \url{http://creativecommons.org/licenses/by-nc-nd/4.0/}.
\end{relaxedbox}    
Communication networks form the backbone of our society.
Topology control algorithms optimize the topology of such communication networks.
Due to the importance of communication networks, a topology control algorithm should guarantee certain required consistency properties (\eg, connectivity of the topology), while achieving desired optimization properties (\eg, a bounded number of neighbors).
Real-world topologies are dynamic (\eg, because nodes join, leave, or move within the network), which requires topology control algorithms to operate in an incremental way, \idest, based on the recently introduced modifications of a topology.
Visual programming and specification languages are a proven means for specifying the structure as well as consistency and optimization properties of topologies.
In this paper, we present a novel methodology, based on a visual graph transformation and graph constraint language, for developing incremental topology control algorithms that are guaranteed to fulfill a set of specified consistency and optimization constraints.
More specifically, we model the possible modifications of a topology control algorithm and the environment using graph transformation rules, and we describe consistency and optimization properties using graph constraints.
On this basis, we apply and extend a well-known constructive approach to derive refined graph transformation rules that preserve these graph constraints.
We apply our methodology to re-engineer an established topology control algorithm, \ktc, and evaluate it in a network simulation study to show the practical applicability of our approach.

\end{abstract}


\maketitle 

\pagebreak

\section{Introduction}
\label{sec:introduction}

Topology control (\TC) is an important research area in the wireless network communication domain.
TC aims at adapting the topology of wireless networks to optimize, for instance, the total power consumption, while maintaining crucial constraints of the topology (\eg, connectivity)~\cite{Santi2005,YMG08,Wang08,Sun2010,Li2013}.
A \TC algorithm typically works by 
\begin{inparaenum}
\item first selecting a subset of the links of the original topology so that all required constraints are fulfilled, and 
\item then adjusting the transmission power of each node to reach its farthest neighbor across one of these selected links.
\end{inparaenum}
In realistic settings, context events such as movement of sensor nodes continuously modify the structure of a topology.
Therefore, a \TC algorithm should operate in an incremental way by efficiently updating only the affected subset of links based on the occurred context events.

The development of a \TC algorithm is performed by highly skilled and experienced professionals.
The development process usually starts with an informal specification of the basic properties of a \TC algorithm.
This informal description is then supplemented by a formal specification using a theoretically well-founded framework such as graph theory~\cite{SWBM12,WZ04,LHL05} or game theory~\cite{MS08,SWKC12}, which allows to prove that the algorithm preserves all required constraints.
The first evaluation of a \TC algorithm is typically carried out using a network simulator, which may be succeeded by a second evaluation in a testbed environment, \idest, on real wireless devices.
Both types of evaluation require an implementation of the \TC algorithm in one or---most often---two programming languages, such as Java or MATLAB for the simulation and C or C++ for the testbed evaluation~\cite{Papadopoulos2016}.
This means that, in the end, the (hopefully) same \TC algorithm is represented in two or three more or less completely different representations.
In research publications, often only the formalization (along with the proofs of correctness and optimality) and a pseudo-code representation of the \TC algorithm is given, while the implementations are typically omitted.
Still, even for skilled researchers, it may be difficult to verify that the pseudo code is a valid implementation of the formal specification.

In the following, we illustrate this by means of the Cooperative Topology Control Algorithm (CTCA), proposed by Chu and Sethu in an IEEE INFOCOM paper in 2012~\cite{XH12}.
The authors first give a graph-based intuition of the proposed \TC algorithm, which shall lead to an improved distribution of the nodes' lifetimes in a wireless sensor network~\cite[Sec.~III]{XH12}.
The resulting goals are then formalized as so-called ordinary potential game~\cite{Monderer1996}, which allows to prove that the proposed algorithm eventually leads to a stable and optimal state of the network.
The authors present an implementation of their algorithm in pseudo code, which is 83 lines long, distributed across four listings, and enriched with network-specific aspects such as communication message exchange;
all these aspects make it highly non-trivial to understand the correspondences to the game-theoretic formalization~\cite[Sec.~V]{XH12}.
In a simulation-based evaluation, the authors compare CTCA with other state-of-the-art \TC algorithms.
Unfortunately, no details about the simulation platform are given~\cite[Sec.~VI]{XH12}.
To the best of our knowledge, no testbed evaluation of CTCA has been performed yet.

While the previous example considers only one of the many existing \TC algorithms, experience shows that it is (at least partly) representative.
The example illustrates an obvious and prevalent gap between the formal specification, which serves for proving important properties of the algorithm, and the implementation, which serves for assessing the \TC algorithm.
Due to this gap, it inherently remains unclear whether the evaluated implementation of a \TC algorithm fulfills the properties that have been proved based on the specification. 

This is especially true for the case where an incrementally working TC algorithm is required. 
The transition from a batch \TC algorithm, which takes a complete topology as input and produces a modified (optimized) complete topology as output, to an incremental version, which takes an arbitrary set of topology (context) modifications as input and produces a (minimal) set of topology adaptations as output, is an error-prone process. 
Experience shows that it is extremely challenging to cover all possible combinations of topology modifications in such a way that the computed topology adaptations never violate the given set of topology constraints and optimization goals.
Contributions such as those by Zave show that even formalizations of well-known network algorithms often reveal special cases where these algorithms do not work properly~\cite{Zave2008,Zave2012}.
For a more comprehensive survey of the application of formal methods to networking algorithms, we refer the reader to~\cite{QH15}. 

\paragraph{Towards a seamless construction process for \TC algorithms}
It is the vision of our research activities as part of the collaborative research center MAKI (Multi-Mechanism Adaptation for the Future Internet%
\footnote{In German: \emph{\textbf{M}ulti-Mechanismen-\textbf{A}daption für das \textbf{k}ünftige \textbf{I}nternet}, \url{http://www.maki.tu-darmstadt.de}}) to close the gap between a carefully crafted formal specification and its derived implementation as follows.
We propose a methodology for constructing \TC algorithms starting with a concise formal specification and refining this specification step-by-step to an efficiently working implementation.
The resulting implementation is correct by construction if we can show that all refinement steps preserve the properties of the initial specification.
For this purpose, we adhere to a model-driven engineering (MDE)~\cite{VSBHH13} approach, which works as follows: 
\begin{itemize}
\item  Topologies are formalized as models that are proper instances of a common meta-model that represents all relevant properties of the considered class of topologies, \eg, link-weighted topologies.
\item The meta-model of a studied topology class is extended with a set of consistency constraints and optimization goals.
\item Model transformation rules describe all relevant (context) modifications of a topology class and the expected constraint-preserving topology adaptations of the constructed TC algorithm.
\item Code generators translate the rule-based description of a TC algorithm into an efficiently working implementation that can be used in a software simulator or a hardware testbed for evaluation purposes.
\end{itemize}

Directed or undirected graphs are commonly used to formalize the structure of communication system topologies (\eg,~\cite{SWBM12,WZ04,XH12,BRVH15}) and  TC algorithms are often sketched visually as sequences of topology graph modifications (\eg, \cite{JRSST09,KKS14}).
Therefore, graph transformation (\GT) constitutes a natural basis for developing our MDE methodology.
\GT offers a set of rule-based and declarative techniques for the high-level specification of model- or graph-manipulating algorithms with a well-defined semantics~\cite{GT:HandbookI,EEPT06}.
A variety of GT-based tools are available for formal specification and rapid prototyping purposes of specified algorithms~\cite{LAS14,Ren04,ABJKT10,GBGHS06,Taentzer00,NNZ00}. 

GT languages and tools are established representatives of the whole class of visual languages (VL).
As a consequence, our selected approach adheres to the tradition of both the VL and the MDE community to adopt visual modeling and programming languages for the high-level description of the structure and behavior of communication systems and, more generally, of distributed information systems.
Today, UML~\cite{OMG06}---an assembly of a number of previously popular VLs (\eg, state charts, message sequence charts, ROOM structure diagrams)---is a well-established visual modeling language used for MDE activities in the area of communication and distributed systems (\eg, \cite{Gomaa2001,DeWet2005, WBSO07}).
Languages like eMoflon~\cite{LAS14} or MechatronicUML~\cite{SW10} even integrate GT concepts for dynamic communication topology manipulation purposes with UML-like activity, class, and composite structure diagrams as well as state charts. 
Apart from these languages, the VL community has already been developing visual programming languages with well-defined syntax and semantics for distributed communication and information system construction purposes for several decades (\eg, G-Net~\cite{CDC93,NKMD96}).
A comprehensive survey of related research activities can be found in~\cite{Zhang07, MM98, CGBGB90}.

To summarize, the TC algorithm approach presented in this paper relies on the subclass of GT-based visual languages and follows the tradition of the formal program-construction-by-transformation approaches (see, \eg, \cite{BMPP89}), which have their roots in research activities of the 1970s like the Munich project CIP (computer-aided intuition-guided programming) and are today part of the vision of model-driven engineering activities~\cite{FR07}.

Model constraints may be used for specifying required or forbidden properties of models.
Visual graph constraints have been introduced by the \GT{} community as a means to characterize classes of graphs in a formal and declarative way~\cite{EEPT06,HP05,HW95}.
For particular classes of graph constraints, formal refinement algorithms have been proposed that take a set of \GT{} rules and graph constraints as input and produce a refined set of \GT{} rules that  \emph{preserve} the given constraints~\cite{HW95,HP05,DV14}.
More precisely, this means that applying the refined \GT{} rules will not cause violations of the specified graph constraints.

A number of model-driven approaches for developing \TC algorithms have been proposed (\eg,~\cite{FRL09,RDBPP15,WBSO07,TSFH15,BKLLXZ04}) that target two major objectives:
The first major objective is to reduce the complexity of developing \TC algorithms from the point of view of domain experts by providing suitable (visual) abstractions, \eg, by using activity diagram-like syntax to specify the control flow of a \TC algorithm.
The second major objective is to simplify the testing and debugging by implementing the \TC algorithm against a middleware layer, which enables that the exact same algorithm may be exercised inside a software simulation and a hardware testbed environment.
However, to the best of our knowledge, none of these approaches focuses on integrating consistency properties \emph{constructively} into the development process of \TC algorithms.
Instead, the proposed approaches focus on facilitating formal analysis, debugging, automated code generation, or the deployment of \TC algorithms based on models.

\paragraph{Contribution}
In this paper, we present a model-based methodology for constructing incremental \TC algorithms that are guaranteed to preserve specified formal properties.
More specifically, we eliminate the previously described gap between specification and implementation of \TC algorithms as follows:
\begin{enumerate}
\item
We characterize the required and desired properties of output topologies of a \TC algorithm using graph constraints (as described in \cite{HW95}).
\item
We use \GT rules to describe \TC operations, which specify the possible modifications during an execution of a \TC algorithm, and context events, which specify the possible modifications by the environment.
\item 
We apply and enhance the constructive refinement approach for \GT{} rules presented in~\cite{HW95,DV14} to \TC.
Our approach is able to cope with temporarily constraint-violating rules, which contrary to \cite{HW95,DV14} may not be restricted.
\item
We exemplify our approach by constructing a representative \TC algorithm, \ktc~\cite{SWBM12}, and assess it quantitatively in a network simulation environment.
\end{enumerate}
This work is a considerable extension of~\cite{KVS15}:
First, we give a detailed explanation of our constructive approach;
second, we show how the approach may be extended to cope with context events as well;
and third, we evaluate our approach \wrt correctness, incrementality, performance, and general applicability.

\paragraph{Structure}
The structure of this paper is illustrated in \Cref{fig:architecture}:
\Cref{sec:metamodeling} introduces the basic concepts of network topologies and topology control.
\Cref{sec:constraints} introduces graph constraints as a means to specify consistency properties of topologies as well as optimization goals of \TC algorithms.
\Cref{sec:gratra} presents \GT rules as a means to specify \TC operations and context events.
\Cref{sec:rule-refinement} describes the rule refinement procedure, which combines the \GT rules of \TC and context events with the graph constraints to produce refined \GT rules that preserve the specified consistency properties of topologies and achieve the specified optimization goals.
\Cref{sec:evaluation} presents the results of the evaluation, and \Cref{sec:related-work} surveys related work.
\Cref{sec:conclusion} concludes this paper with a summary and an outlook.

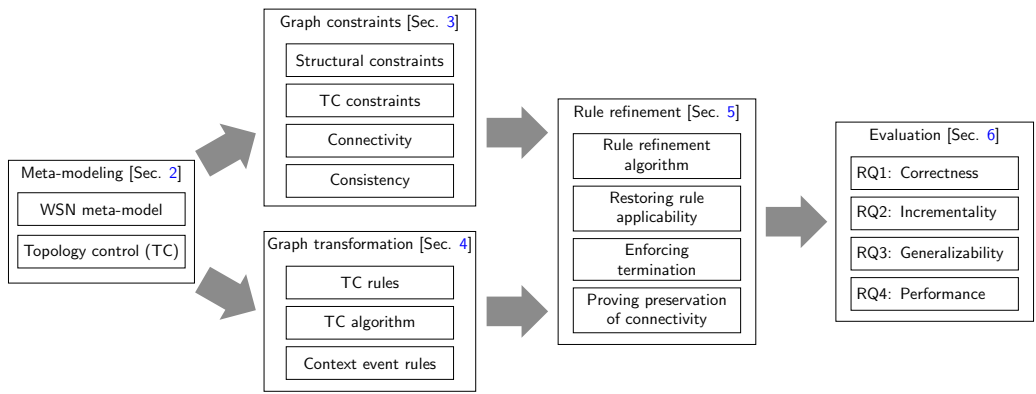
\begin{figure}
    \begin{center}
    	\resizebox{1\textwidth}{!}{%
        \begin{tikzpicture}[node distance = 3cm]
        
        \sffamily
       
        \node [rc] (MetaModelingFrameNode) at (0,0) {};
        \node [below] (MetaModelingTitleNode) at (MetaModelingFrameNode.north) {Meta-modeling \lbrack Sec. \ref{sec:metamodeling}\rbrack};
        \node [rcSmall, below=0.1cm of MetaModelingTitleNode] (MetaModelingSubNode) {WSN meta-model};
        \node [rcSmall, below=0.2cm of MetaModelingSubNode] {Topology control (\TC)};
        
        \node[bigArrow, rotate=30,right of=MetaModelingFrameNode] {};
        \node[bigArrow, rotate=-30,right of=MetaModelingFrameNode] {};
        
        \node[rc, right of=MetaModelingFrameNode, above=-1.cm of MetaModelingFrameNode, minimum height=4.25cm, text width=4.3cm, node distance =5.75cm](GraphConstraintsFrameNode) {};
        \node [below] (GraphConstraintsTitleNode) at (GraphConstraintsFrameNode.north) {Graph constraints \lbrack Sec. \ref{sec:constraints}\rbrack};
        \node [rcSmall, below=0.15cm of GraphConstraintsTitleNode] (GraphConstraintsSubNode1) {Structural constraints};
        \node [rcSmall, below=0.15cm of GraphConstraintsSubNode1] (GraphConstraintsSubNode2) {\TC constraints};
        \node [rcSmall, below=0.15cm of GraphConstraintsSubNode2] (GraphConstraintsSubNode3) {Connectivity};
        \node [rcSmall, below=0.15cm of GraphConstraintsSubNode3] {Consistency};
        
        \node[rc, right of=1, above=-5cm of MetaModelingFrameNode, minimum height=3.45cm, text width=4.3cm, node distance =5.75cm](GraphTransformationFrameNode) {};
        \node [below] (GraphTransformationTitleNode) at (GraphTransformationFrameNode.north) {Graph transformation \lbrack Sec. \ref{sec:gratra}\rbrack};
        \node [rcSmall, below=0.15cm of GraphTransformationTitleNode] (GraphTransformationSubNode1) {\TC rules};
        \node [rcSmall, below=0.15cm of GraphTransformationSubNode1] (GraphTransformationSubNode2) {\TC algorithm};
        \node [rcSmall, below=0.15cm of GraphTransformationSubNode2]  {Context event rules};
        
        \node[bigArrow, right of=GraphConstraintsFrameNode, above=-3cm of GraphConstraintsFrameNode] {};
        \node[bigArrow, right of=GraphTransformationFrameNode] {};
        
        \node[rc, right of=MetaModelingFrameNode, minimum height=5.3cm, text width=4cm, node distance =11.9cm](RuleRefinementFrameNode) {};
        \node [below] (RuleRefinementTitleNode) at (RuleRefinementFrameNode.north) {Rule refinement \lbrack Sec. \ref{sec:rule-refinement}\rbrack};
        \node [rcSmall, below=0.15cm of RuleRefinementTitleNode] (RuleRefinementSubNode1) {Rule refinement\\algorithm};
        \node [rcSmall, below=0.15cm of RuleRefinementSubNode1] (RuleRefinementSubNode2) {Restoring rule applicability};
        \node [rcSmall, below=0.15cm of RuleRefinementSubNode2] (RuleRefinementSubNode3) {Enforcing\\termination};
        \node [rcSmall, below=0.15cm of RuleRefinementSubNode3]  {Proving preservation\\of connectivity};
        
        \node[bigArrow, right=0.2cm of RuleRefinementFrameNode] {};
        
        \node[rc, right of=MetaModelingFrameNode, minimum height=4.3cm, text width=4cm, node distance =17.85cm](EvaluationFrameNode) {};
        \node [below] (EvaluationTitleNode) at (EvaluationFrameNode.north) {Evaluation \lbrack Sec. \ref{sec:evaluation}\rbrack};
        \node [rcSmall, below=0.15cm of EvaluationTitleNode,align=left] (EvaluationSubNode1) {RQ1: Correctness};
        \node [rcSmall, below=0.15cm of EvaluationSubNode1,align=left] (EvaluationSubNode2) {RQ2: Incrementality};
        \node [rcSmall, below=0.15cm of EvaluationSubNode2,align=left] (EvaluationSubNode3) {RQ3: Generalizability};
        \node [rcSmall, below=0.15cm of EvaluationSubNode3,align=left]  {RQ4: Performance};
        
        \end{tikzpicture}}%
    \end{center}
    \caption{Overview of the proposed constructive methodology and structure of this paper}
    \label{fig:architecture}
\end{figure}

\section{Network Topologies and Topology Control}
\label{sec:metamodeling}

In this section, we introduce basic concepts of meta-modeling, wireless sensor networks and topology control, including the algorithm \ktc~\cite{SWBM12}, which serves as our running example.

\subsection{Meta-Modeling}

A \emph{meta-model} is a graph that describes the basic elements of a domain. 
A meta-model consists of \emph{classes}, which describe the entities in the domain, and \emph{associations} between classes, which describe possible connections between entities.
A class has zero or more \emph{attributes} (shown in the lower part of the class), which may have primitive or enumeration types.
Each end of an association has a descriptive \emph{role} and a \emph{multiplicity}, which restricts the number of associations in an instance of the meta-model.
A \emph{model} is a graph that represents a concrete instance of a meta-model.

\subsection{Topology Control and Context Events}
\label{sec:model-tc}

\emph{Topology control} (\TC{}) is the discipline of manipulating the topology of a network to achieve optimize goals while preserving a set of consistency properties.
In this paper, we focus on \TC{} for wireless sensor networks (\WSNs{}).
A \WSN{} network \emph{topology} consists of a large number of battery-powered sensor \emph{node}s that collectively perform a dedicated task, \eg, data collection, environmental monitoring, or movement tracking~\cite{YMG08}.
A node communicates with other nodes that are within its maximum transmission range via communication \emph{link}s.
Taking into account that a \WSN{} may consist of thousands of nodes, it is often infeasible to recharge the batteries of all sensor nodes.
For this reason, reducing the total energy consumption of a \WSN{} is one of the most important optimization goals of \TC{}.

A \emph{\TC{} algorithm} typically works in two steps:
In the \emph{marking step}, the \TC{} algorithm selects all links that are crucial for fulfilling the consistency and optimization requirements.
In the \emph{adaptation step}, each node may reduce its transmission power in a way that it is still able to reach its farthest neighbor within the set of selected links.
In practice, the adaptation step may also enforce additional properties if they are not ensured during the marking step, \eg, that for each selected link, its reverse link is also treated as selected.
The focus of constructing a \TC{} algorithm lies on the marking step because the adaptation step is typically performed by the concrete application that uses the output topology of the \TC algorithm.

To represent the selection state of the edges $E$ in a topology, we introduce the state function $s: E \to \{\ACT, \INACT, \UNCL\}$:
A link is in state \emph{active} (\idest, $\state(e)= \ACT$) if it is selected by the \TC{} algorithm and \emph{inactive} (\idest $\state(e) = \INACT$) if not.
A link is \emph{unclassified} (\idest $\state(e) = \UNCL$) if the \TC{} algorithm has not made a decision about it, yet.

\paragraph{Context Events}

The topology of realistic networks is \emph{dynamic}, \idest, it is continuously modified by the environment, which we model using the following five types of  \emph{context events} (abbreviated as CEs).
\begin{itemize}
    \item \emph{Node addition}: A new node may appear, \eg, because it replaces a deceased node or because it has been recharged.
    
    \item \emph{Node removal}: A node may disappear, \eg, because it runs out of energy.
        
    \item \emph{Link removal}: A link may disappear completely, \eg, if the weight of its incident nodes exceeds the maximum transmission range.
    
    \item \emph{Link addition}: A link may (re-)appear, \eg, if two nodes converge so that the weight between them drops below the maximum transmission range.
    
    \item \emph{Link-weight modification}: The weight of a link (representing, \eg, the distance of its incident nodes) may change, \eg, if its incident nodes move.
\end{itemize}

\paragraph{Incrementality}
In a typical application scenario, a \TC{} algorithm optimizes the entire topology initially.
Afterwards, a number of \CEs{} modify the topology.
The modified topology may then be repaired by executing the \TC{} algorithm in one of two ways:
A \emph{batch \TC{} algorithm} neglects the concrete \CEs{} that led to the current situation and optimizes the topology from scratch.
This is equivalent to unclassifying all links whenever a \CE{} occurs.
An \emph{incremental \TC{} algorithm} takes the occurred \CEs{} into account and repairs the topology based on this knowledge.
This is equivalent to storing the link state for each link and only unclassifying those links that are affected by a \CE{}.
For this reason, an incremental \TC algorithm consists of two main parts: the actual \emph{\TC algorithm} that implements the optimization logic of the \TC algorithm and \emph{\CE handler}s that react appropriately to \CEs.
Typically, each type of \CE (\eg, link removal) is handled by a dedicated \CE handler.
Note that it is desirable to develop incremental \TC{} algorithms because the extent of \CEs{} is typically small compared to the network size.

\subsection{Network Topology Meta-Model}

Like the common practice in the \WSN community, we model \WSN{} topologies as attributed, directed graphs.
To simplify the presentation, we apply the meta-model of link-weighted graphs as shown in \Cref{fig:metamodel-topology}.
A \emph{topology} consists of \emph{node}s, which are connected via directed \emph{link}s.
A \emph{link} \lvE{} has a real-valued \emph{weight} \weightOf{\lvE}, which represents its cost (\eg, the physical distance of its incident nodes or its latency).
To simplify the presentation in this paper, the meta-model elements contain only essential attributes that are required to model link-weighted topologies.
In fact, several state-of-the-art \TC algorithms in the \WSN community (implicitly or explicitly) apply this model (\eg, XTC~\cite{WZ04}, RNG~\cite{Karp2000}, GG~\cite{Wang08}, LMST~\cite{LHL05}).
The \emph{state} $\state(e)$ of a link \lvE{} may be either \ACT{}, \INACT{}, or \UNCL.
In concrete syntax, the state $\state(e)$ of a link $e$ is represented by %
\begin{inparaenum}
    \item a solid line if $\state(e) = \ACT$,
    \item a dashed line if $\state(e) = \UNCL$,
    \item a dotted line if $\state(e) = \INACT$,
    \item a mixed solid-dotted line if $e$ is \emph{classified}, \idest, $\state(e) \in \{\ACT,\INACT\}$, and
    \item a gray solid line if $e$ is in an arbitrary state, \idest, $\state(e) \in \{\ACT,\INACT,\UNCL\}$.
\end{inparaenum}
In concrete topology models, only the first three syntax elements (\idest, solid, dotted, and dashed lines) may be used, while the latter two syntax elements (\idest, mixed solid-dotted and gray lines) may be applied, for instance, to characterize or represent sets of topology models.
A \emph{link state modification} (abbreviated as \LSM) is the action of changing the state of a particular link.
The \emph{size of a topology} is the sum of its node count and link count.
The \emph{active-link subtopology} of a topology consists of all active links in this topology.
Analogously, we define inactive-link, unclassified-link, and classified-link subtopologies.
For conciseness, we refrained from explicitly representing instances of the class \refToFig{Topology} in concrete syntax;
in the following, we assume that only one instance of \refToFig{Topology} exists and that each node and link is connected to this topology via a \refToFig{nodes-topology} and \refToFig{links-topology} association, respectively.

\Cref{fig:small-sample-topology} shows a sample topology of size \numprint{9} consisting of the three nodes \nodeName{1}, \nodeName{2}, and \nodeName{3}, which are connected by the six directed links \linkName{12}, \linkName{21}, \linkName{13}, \linkName{31}, \linkName{23}, and \linkName{32}.
The name of a link encodes its source and target node, \idest, link \linkName{12} starts at source \nodeNameLong{1} and ends at target \nodeNameLong{2}.
If a topology is undirected, we omit arrow heads, and where possible, we omit link names for improved readability.
The label of a link refers to its weight.

We say that a link is \emph{weight-maximal} (\emph{weight-minimal}) within a set of links, if its weight value is greater (smaller) than or equal to the weight values of all other links in this set;
more formally:
\begin{align*}
    e \text{ is } \left\{\begin{array}{c}\text{\emph{weight-maximal}}\\\text{\emph{weight-minimal}}\end{array}\right\} \text{\wrt\ } E_1 \subseteq E 
    &\Leftrightarrow 
    \forall e' \in E_1: \weight(e) \left\{\begin{array}{c}\geq\\\leq\end{array}\right\} \weight(e')
\end{align*}
For instance, in \Cref{fig:small-sample-topology}, the \linkNameLong{12} and \linkNameLong{21} are weight-maximal, and \linkName{13}, \linkName{31}, \linkName{23}, and \linkName{32} are weight-minimal \wrt $E$.
\begin{figure}
    \begin{center}
        \includegraphics[width=.7\textwidth]{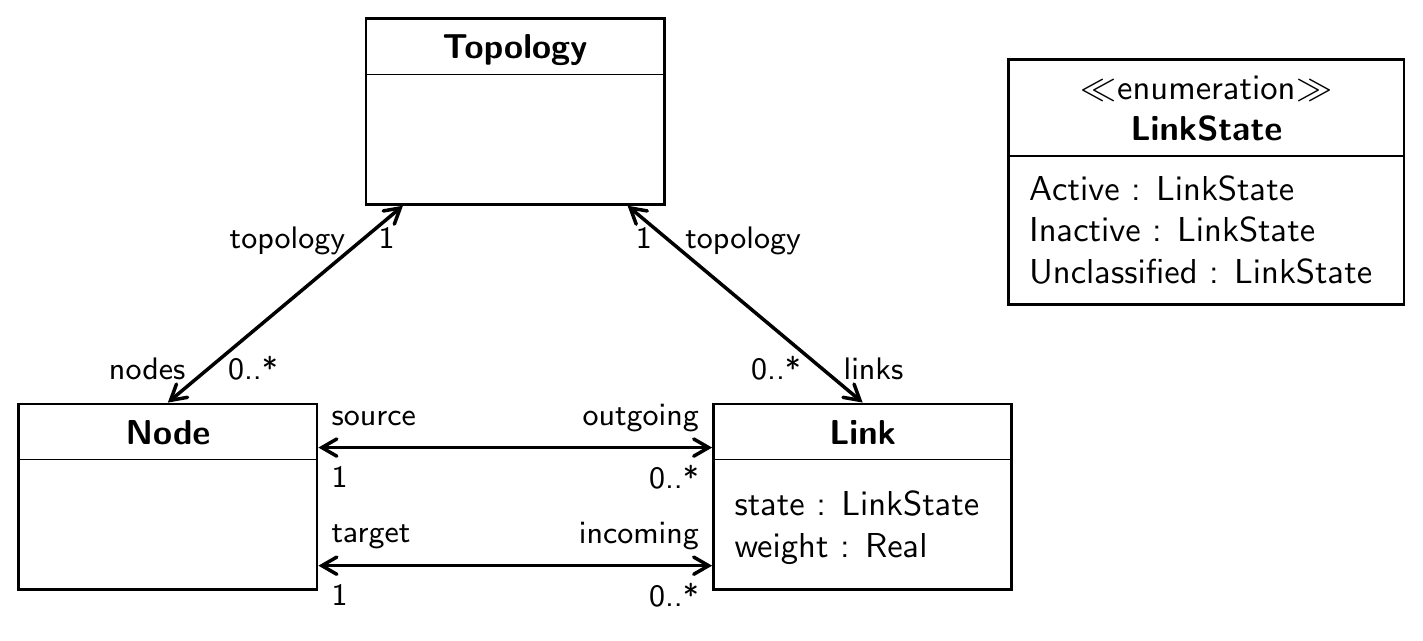}
    \end{center}
    \caption{Meta-model of link-weighted \WSN topologies}
    \label{fig:metamodel-topology}
\end{figure}
\begin{figure*}
    \begin{center}
        \subcaptionbox%
        {Abstract syntax}[\textwidth]%
        {\includegraphics[width=\textwidth]{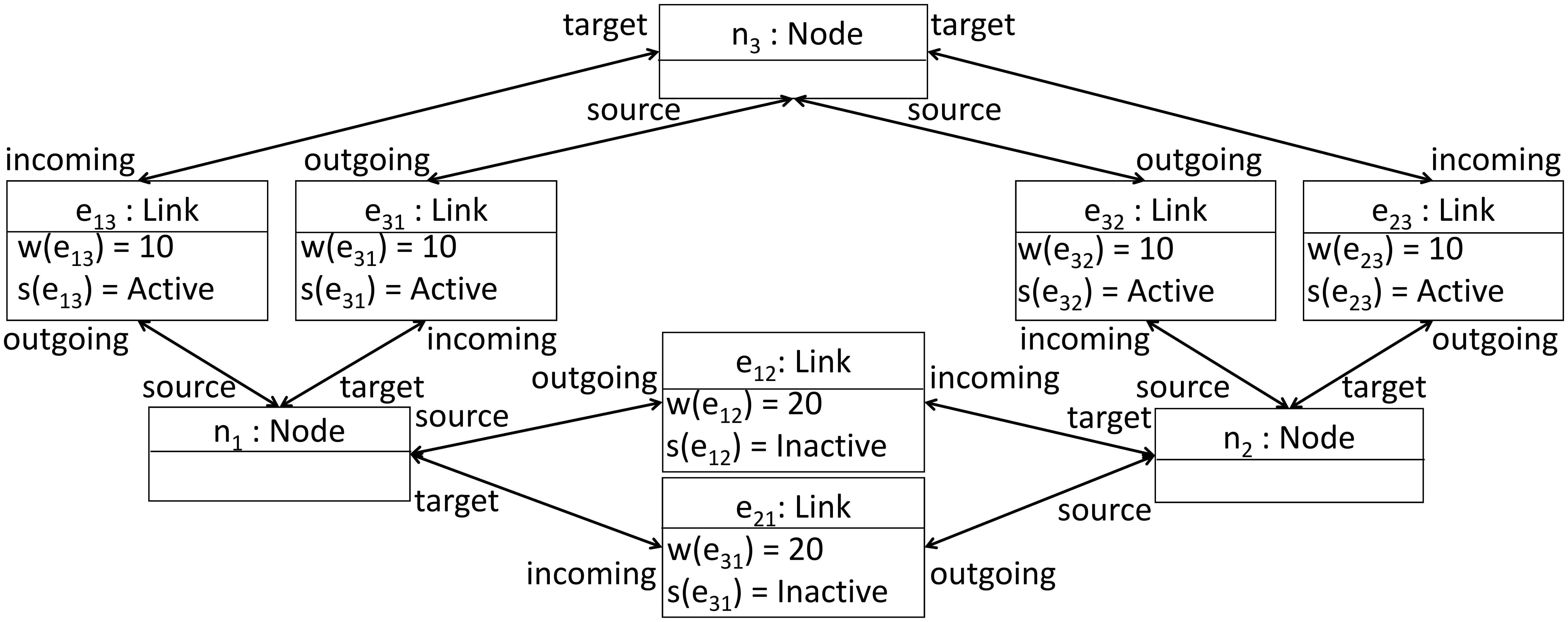}}
        
        \subcaptionbox%
        {Concrete syntax (directed)}[.3\textwidth]%
        {\includegraphics[width=.3\textwidth]{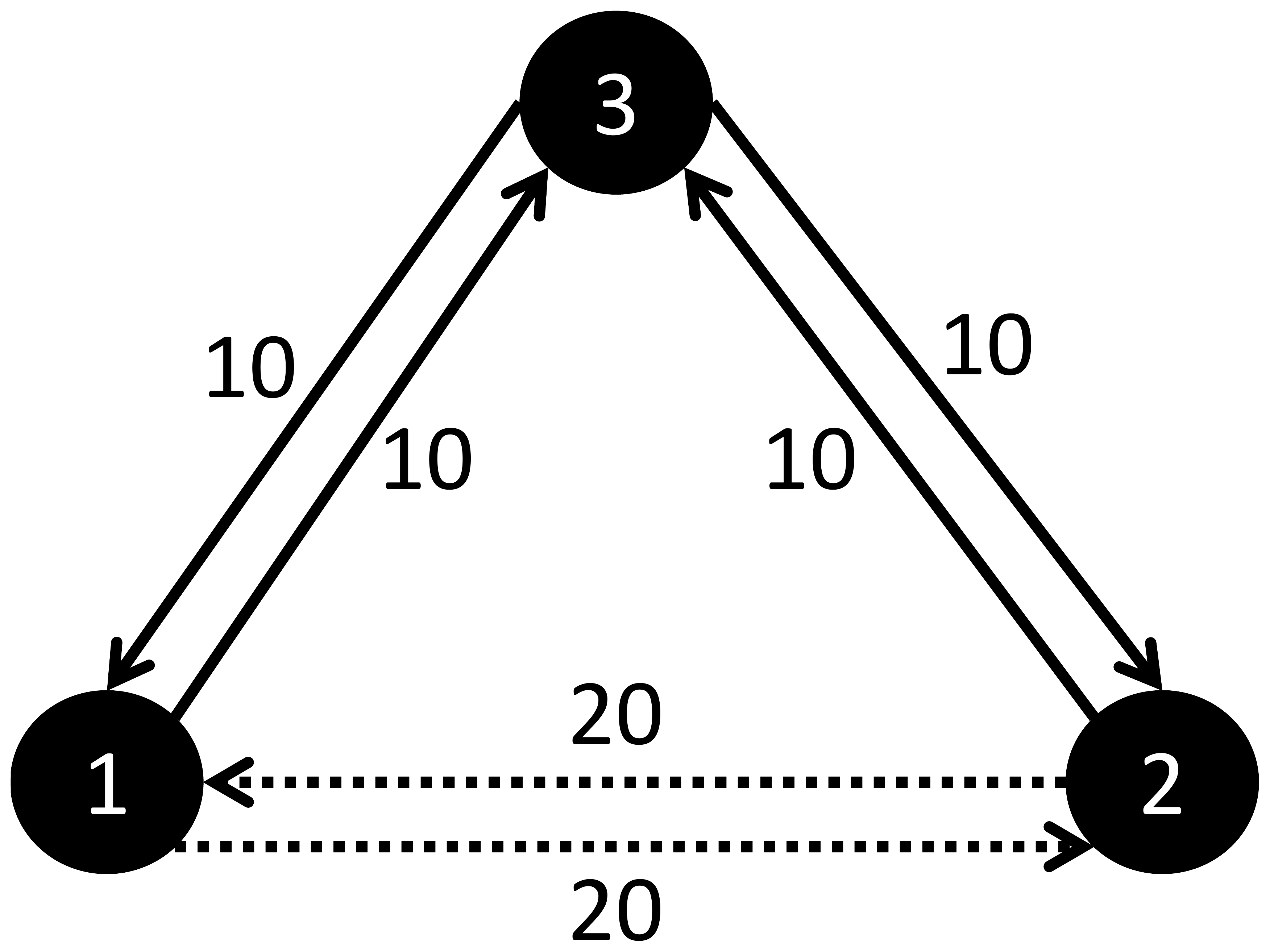}}
        \hspace{1em}
        \subcaptionbox%
        {Concrete syntax (undirected)}[.3\textwidth]%
        {\includegraphics[width=.3\textwidth]{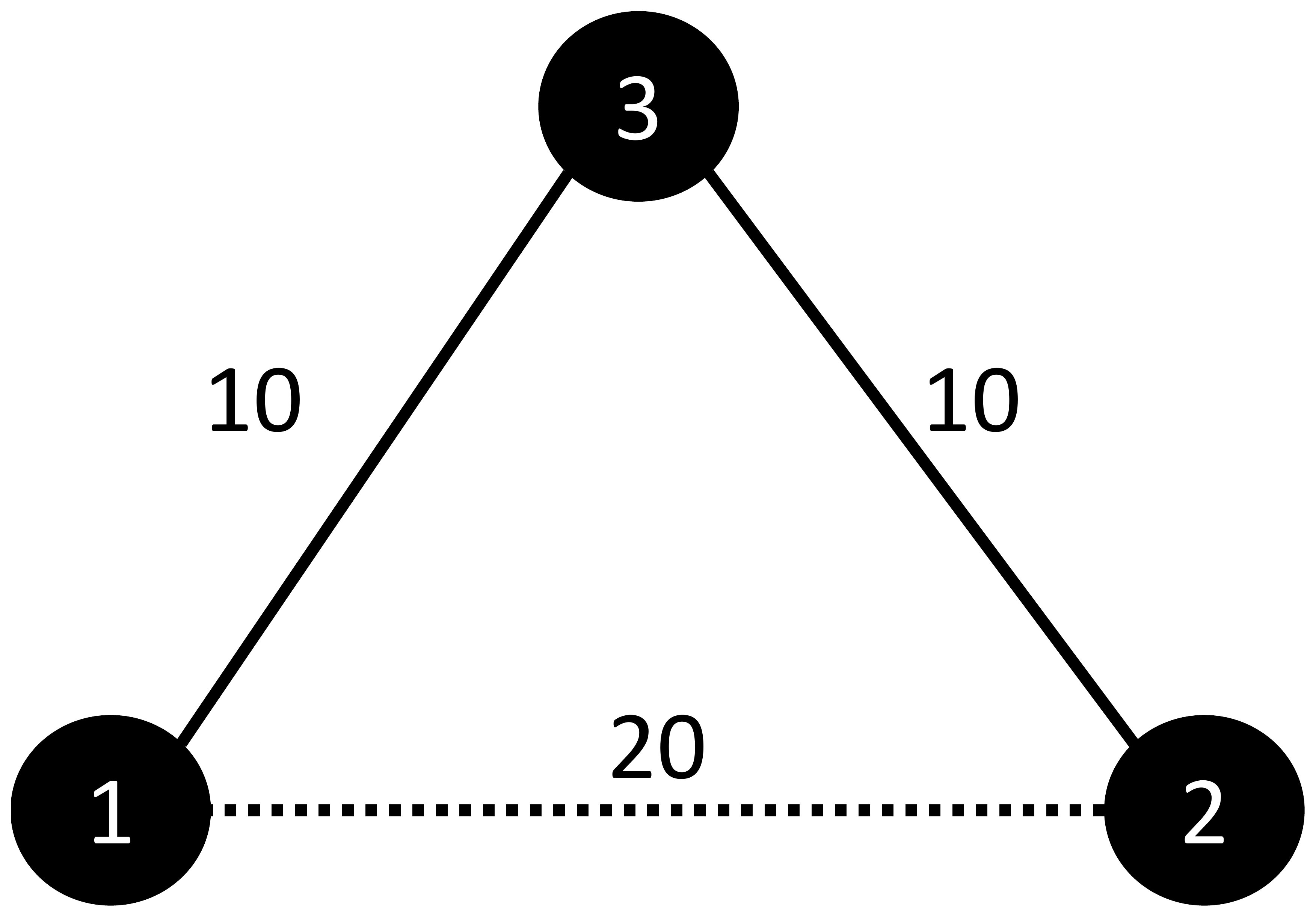}}
        \hspace{1em}
        \subcaptionbox*%
        {}[.3\textwidth]%
        {\includegraphics[width=.3\textwidth]{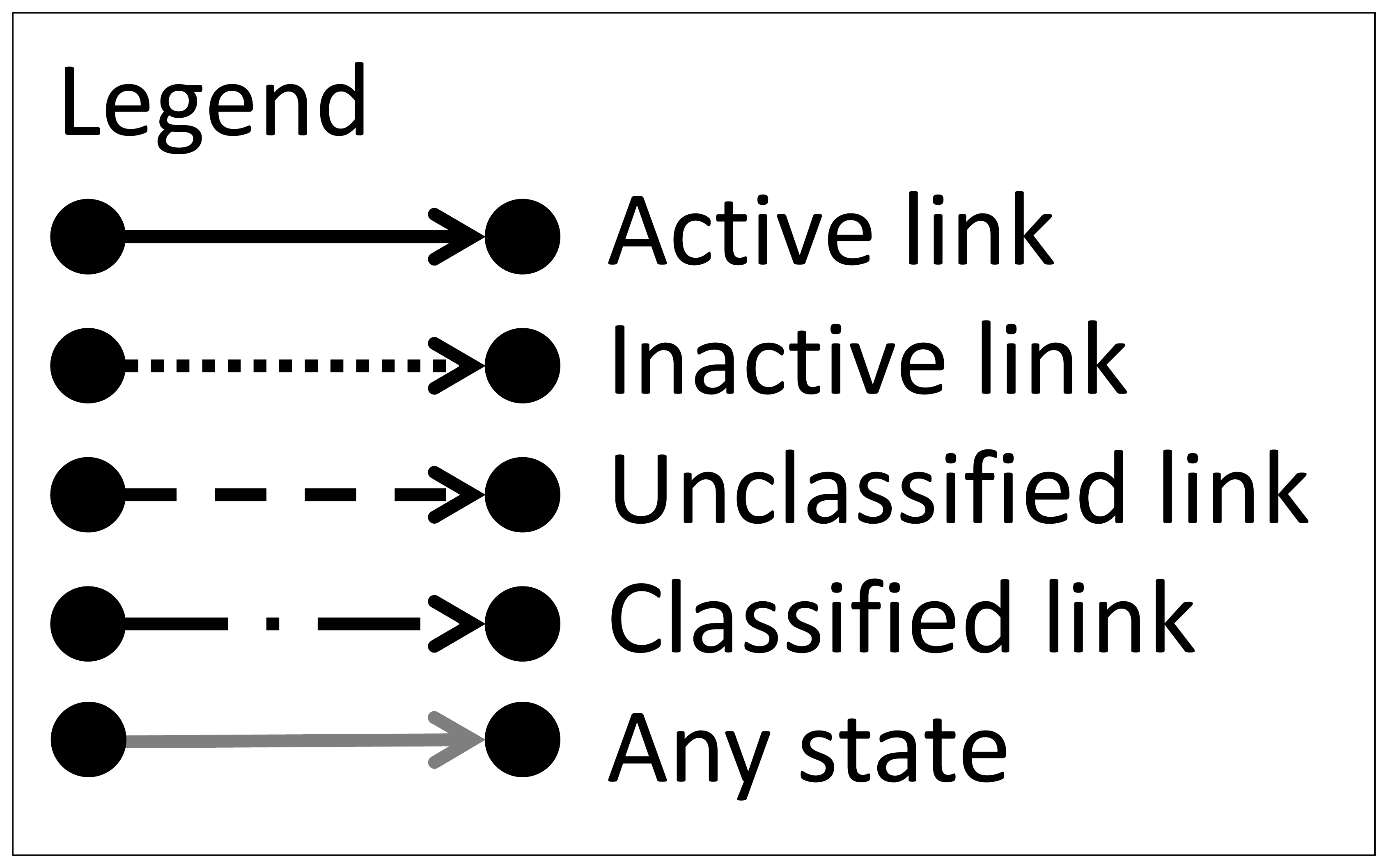}}
    \end{center}
    \caption{Triangle topology in abstract (top) and concrete (bottom) syntax}
    \label{fig:small-sample-topology}
\end{figure*}

\paragraph{Controllability and restrictability}
\label{sec:restrictability}
The \TC developer may restrict the applicability of \TC operations arbitrarily by introducing additional application restrictions.
In contrast, \CEs \emph{may not} be restricted because they represent the influence of the environment on the system.
Instead, the \TC developer has to ensure that applications of \CE rules are \emph{handled} appropriately.
For similar reasons, it makes sense not to restrict the unclassification of links because it should always be possible to \enquote{revert} the classification of a link, \eg, if it turns out to be suboptimal.
Therefore, we say that \CEs and link unclassification are \emph{unrestrictable}.

\subsection{Running Example: The Topology Control Algorithm \ktc}

As the running example of the following explanations, we chose the \TC algorithm \ktc.
\emph{\ktc}~\cite{SWBM12} is a \TC{} algorithm that, in batch mode, inactivates a link if this link is the unique weight-maximal link in some triangle and if the weight of this link is additionally at least $k$-times greater than the weight of the weight-minimal link in the same triangle.
All links that do not fulfill this constraint are activated by \ktc.

\ktc it is a typical representative of a larger class of \emph{local \TC algorithms}, which operate based on local knowledge only~\cite{Stein2016,Naor1995,Linial1990}.
This means that a node only knows about its immediate neighborhood, which is often characterized in terms of the maximum number of \emph{hops}, \idest, the path length, between the node and its known neighbors.
In the \WSN community, local knowledge is often restricted to at most 2~hops because acquiring a larger local knowledge is typically infeasible due to, \eg, memory limitations of the sensor nodes and the message overhead to collect topology information~\cite{SPSBM16}.
Examples of other WSN algorithms that operate with 2-hop local knowledge are XTC~\cite{Wattenhofer2004}, RNG~\cite{Toussaint1980}, Yao Graph~\cite{XiangYang2002}, and GG~\cite{Rodoplu1999}.

\begin{figure*}
    \begin{center}
        \subcaptionbox{Initial topology\label{fig:ktc-example-A}}[.4\textwidth]{\includegraphics[width=.4\textwidth]{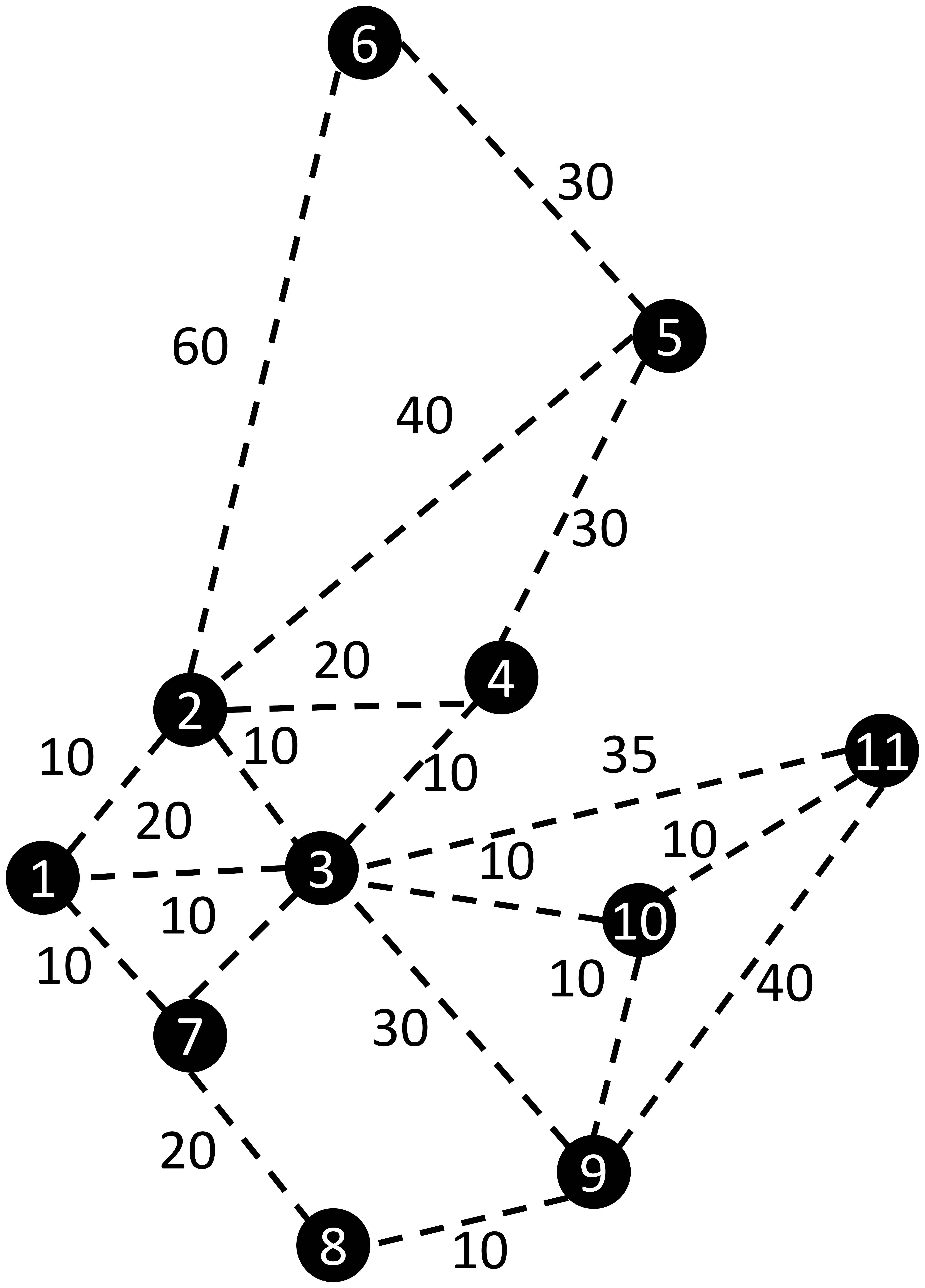}}
        \hspace{1em}
        \subcaptionbox{Topology after executing \ktc ($k{=}2$)\label{fig:ktc-example-B}}[.4\textwidth]{\includegraphics[width=.4\textwidth]{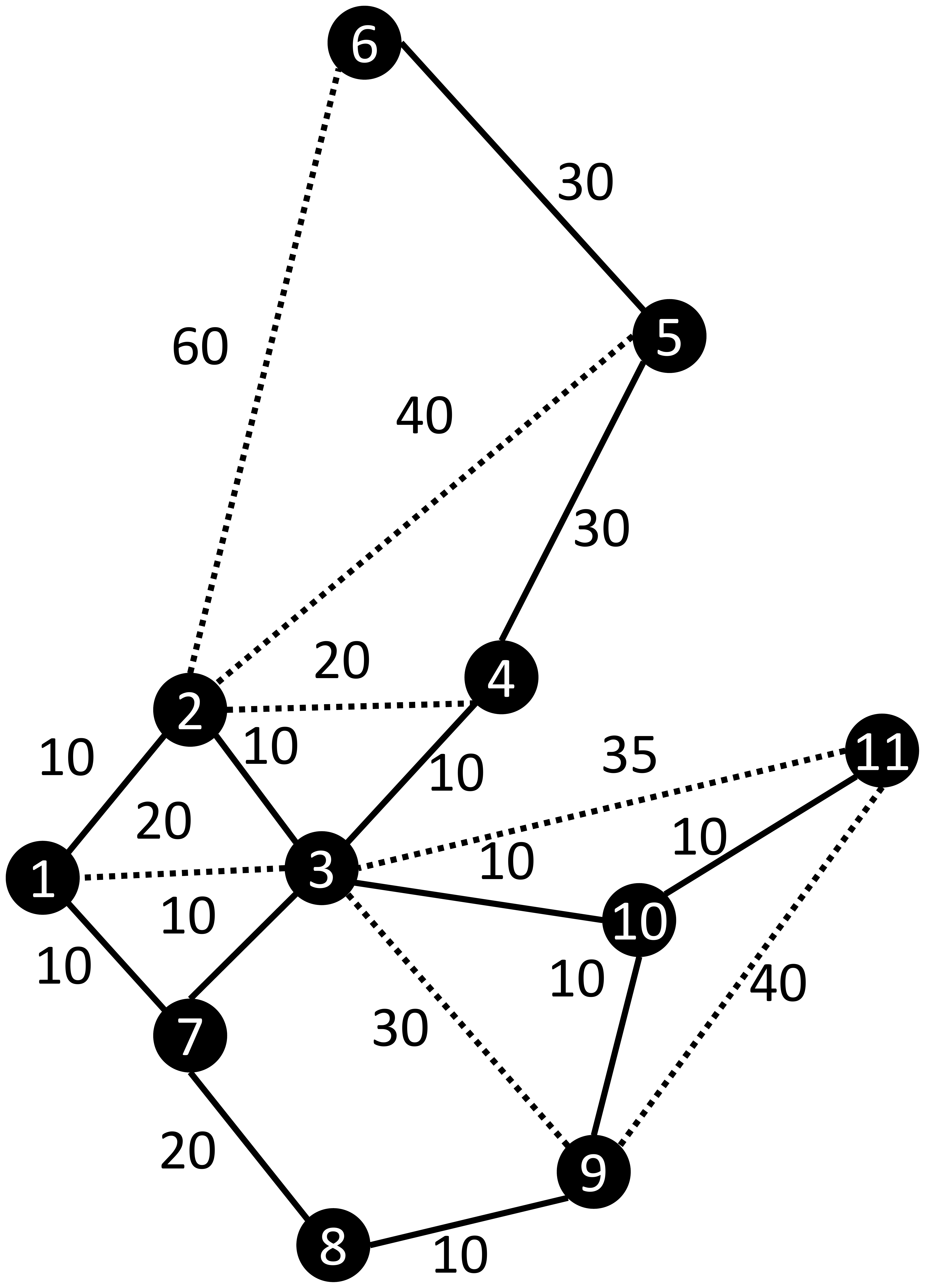}}
        
        \subcaptionbox{Topology after adding links \linkName{79} and \linkName{97} and removing node \nodeName{10}\label{fig:ktc-example-C}}[.4\textwidth]{\includegraphics[width=.4\textwidth]{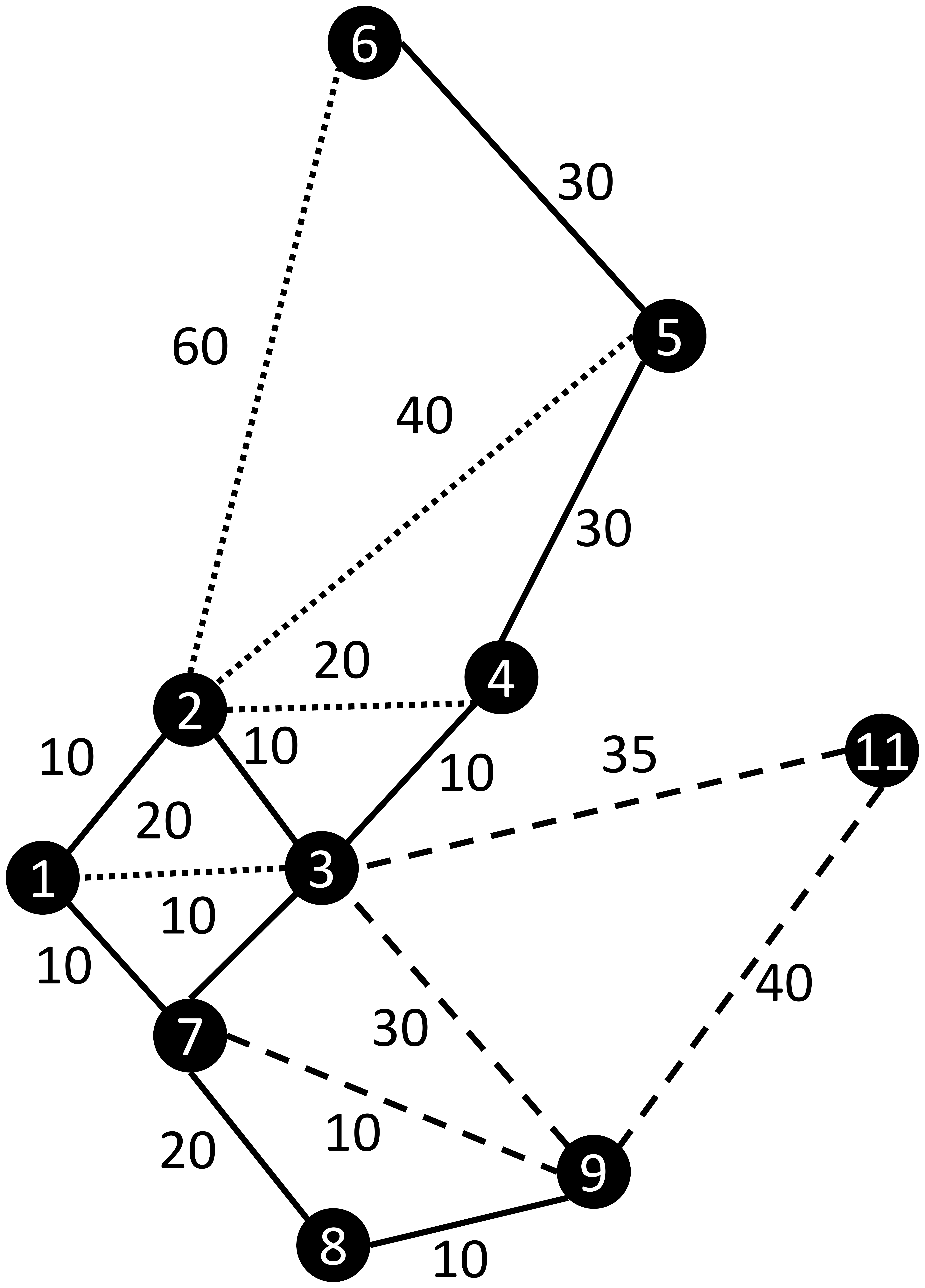}}
        \hspace{1em}
        \subcaptionbox{Topology after re-executing \ktc ($k{=}2$)\label{fig:ktc-example-D}}[.4\textwidth] {\includegraphics[width=.4\textwidth]{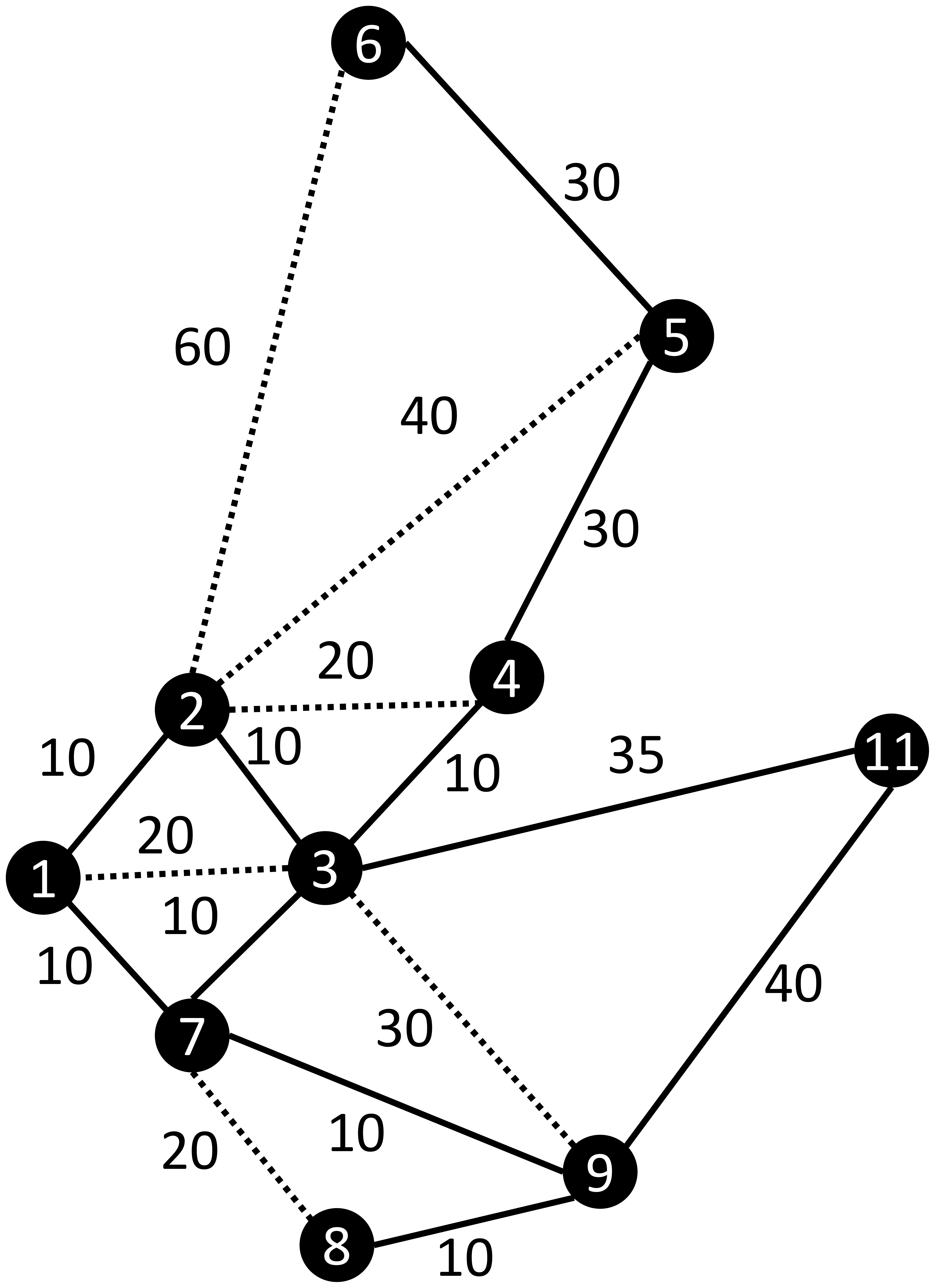}}
    \end{center}
    \caption{Incremental execution of \ktc on a sample topology (link directions omitted)}
    \label{fig:ktc-example}
\end{figure*}
\Cref{fig:ktc-example} illustrates the operation of \ktc on a sample topology consisting of 11~nodes and 38~(directed) links, which are shown as 19~undirected links for presentation purposes.
While $k$ should be in the interval $[1, \sqrt{2}]$ in realistic settings (see \cite{SWBM12} for details), we chose $k{=}2$ for presentation purposes.
Throughout this paper, we assume that \CEs{} do not interfere with the execution of \TC{}, which is realistic because \TC is typically carried out on a snapshot of the local neighborhood.
In the initial topology of our example, all links are unclassified (\Cref{fig:ktc-example-A}).
After executing \ktc, 14~links are inactive and 24~links are active (\Cref{fig:ktc-example-B}).
For instance, \linkNameLong{25} is inactivated because its weight is at least $k$-times greater than the weight of \linkNameLong{24}, which is the weight of the shortest link in the triangle consisting of the links \linkName{25}, \linkName{24}, and \linkName{45}.
Afterwards, the links \linkName{79} and \linkName{97} are added to the topology, \eg, because an obstacle between \nodeName{7} and \nodeName{9} disappeared, and node 10 is removed, \eg, because it was switched off.
These \CEs{} lead to the unclassification of the directly neighbored links \linkName{39}, \linkName{3\,11}, \linkName{79}, and \linkName{9\,11} (\Cref{fig:ktc-example-C}).
In the last step, \ktc repairs the topology incrementally by activating the links \linkName{3\,11}, \linkName{9\,11}, \linkName{79} and by inactivating the links \linkName{78} and \linkName{39} (\Cref{fig:ktc-example-D}).

\section{Specifying Consistent Topologies Using Graph Constraints}
\label{sec:constraints}

In this section, we characterize consistent and optimal output topologies of a \TC algorithm by graph constraints.
We begin with an introduction to the general concepts of graph patterns and graph constraints.
Afterwards, we describe the concrete graph constraints that specify general structural consistency properties of topologies, general requirements of output topologies of \TC algorithms in general, and the optimization goals of \ktc in particular.

\subsection{Graph Patterns and Graph Constraints}

A \emph{graph pattern} $p$ is a graph consisting of \emph{node variables} and \emph{link variables}, which are placeholders for nodes and links of a topology, plus a set of attribute constraints.
An \emph{attribute constraint} relates an attribute value of a node or link variable with attributes of the same or other node or link variables, \eg, using standard comparison operators such as \texttt{=} (equality), \texttt{!=} (inequality).

A \emph{match $m$ of a graph pattern $p$ in a topology $G$} maps node (link) variables of $p$ to nodes (links) in $G$ such that 
\begin{inparaenum}
\item all attribute constraints are fulfilled, and 
\item node variables that are incident to a link variable are mapped to the source and target nodes of the link that the link variable is mapped to.
\end{inparaenum}
Additionally, a match has to be \emph{injective}, \idest, no two node (link) variables may be mapped to the same node (link).
A pattern is \emph{satisfiable} (\emph{unsatisfiable}) if, among all possible topologies, at least one topology (no topology) contains a match of this pattern.

A \emph{graph constraint}~\cite{HW95} consists of a \emph{premise} pattern and a \emph{conclusion}, which is a (potentially empty) set of \emph{conclusion pattern}s.
The premise of a \emph{positive graph constraint} is isomorphic to a subgraph of each conclusion pattern;
the attribute constraints of the conclusion imply the attribute constraints of the premise.
A \emph{negative graph constraint} has no conclusion (denoted by $\emptyset$).
A topology \emph{fulfills} a graph constraint if any match of the premise of the constraint can be extended to a match of one of its conclusion patterns.
This definition implies that a topology fulfills a negative graph constraint if (and only if) the topology does not contain any match of the premise.
A topology \emph{violates} a constraint if it does not fulfill the constraint.

\subsection{Examples of Patterns and Graph Constraints}

In the following, we use graph constraints to characterize structurally consistent topologies in general and valid output topologies of \ktc in particular.

\paragraph{Structural constraints}

We model topologies as simple, directed graphs, \idest, a topology may neither contain parallel links nor loops.
Two edges are \emph{parallel} if their source nodes and their target nodes are identical, respectively.
A \emph{loop} is a link whose source and target node are identical.
The negative \noParallelLinksConstraintLong (shown in \Cref{fig:no-par-links}) and the negative \noLoopsConstraintLong (shown in \Cref{fig:no-loops}) describe these structural consistency properties.
The premise~\premiseNoParallelLinks of the \noParallelLinksConstraintLong matches two distinct links \linkName{ab,1} and \linkName{ab,2} that both connect the same source \nodeNameLong{a} to the same target \nodeNameLong{b}.
The premise~\premiseNoLoops of the no-loops constraint~\noLoopsConstraint matches a single link \linkName{aa} that connects a \nodeNameLong{a} to itself.
\begin{figure}
    \begin{center}
        \newcommand{\subfigwidth}{.32\textwidth}
        \subcaptionbox{\noParallelLinksConstraint\label{fig:no-par-links}}[\subfigwidth]{            \includegraphics[width=\subfigwidth]{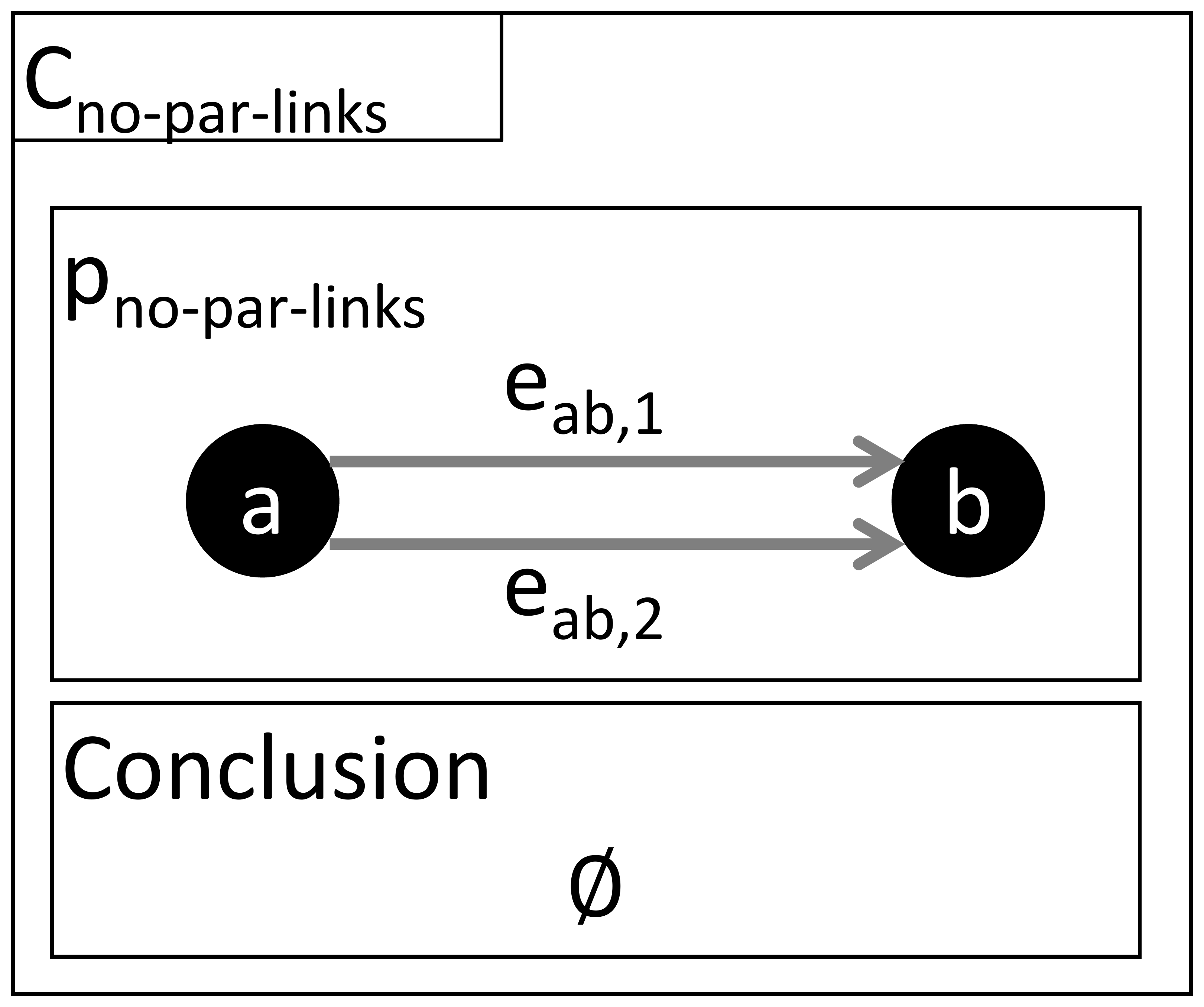}}
        \subcaptionbox{\noLoopsConstraint\label{fig:no-loops}}[\subfigwidth]{
             \includegraphics[width=\subfigwidth]{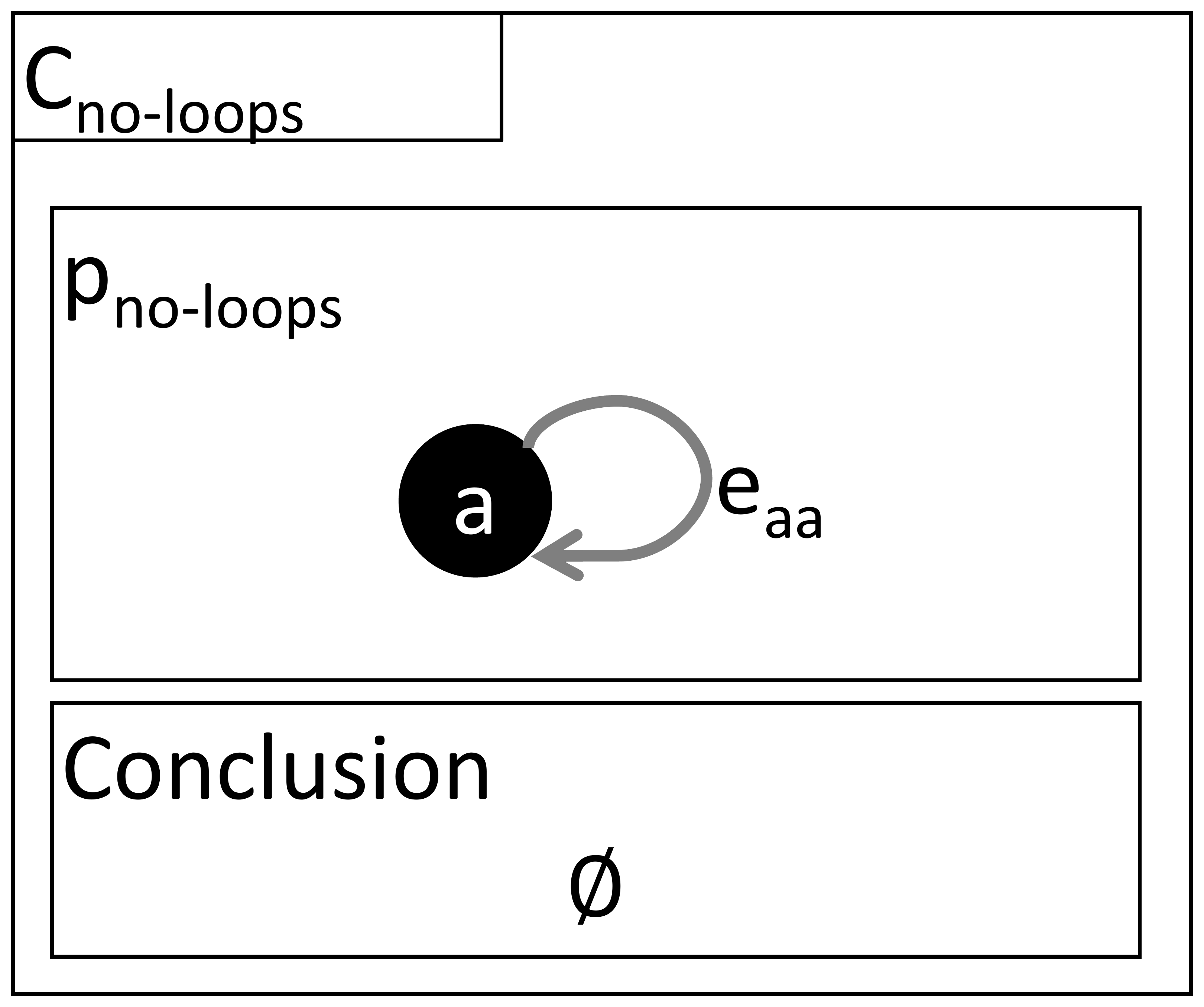}}
        \subcaptionbox{\unclassifiedLinkConstraint\label{fig:unclassified-link-constraint}}[\subfigwidth]{
            \includegraphics[width=\subfigwidth]{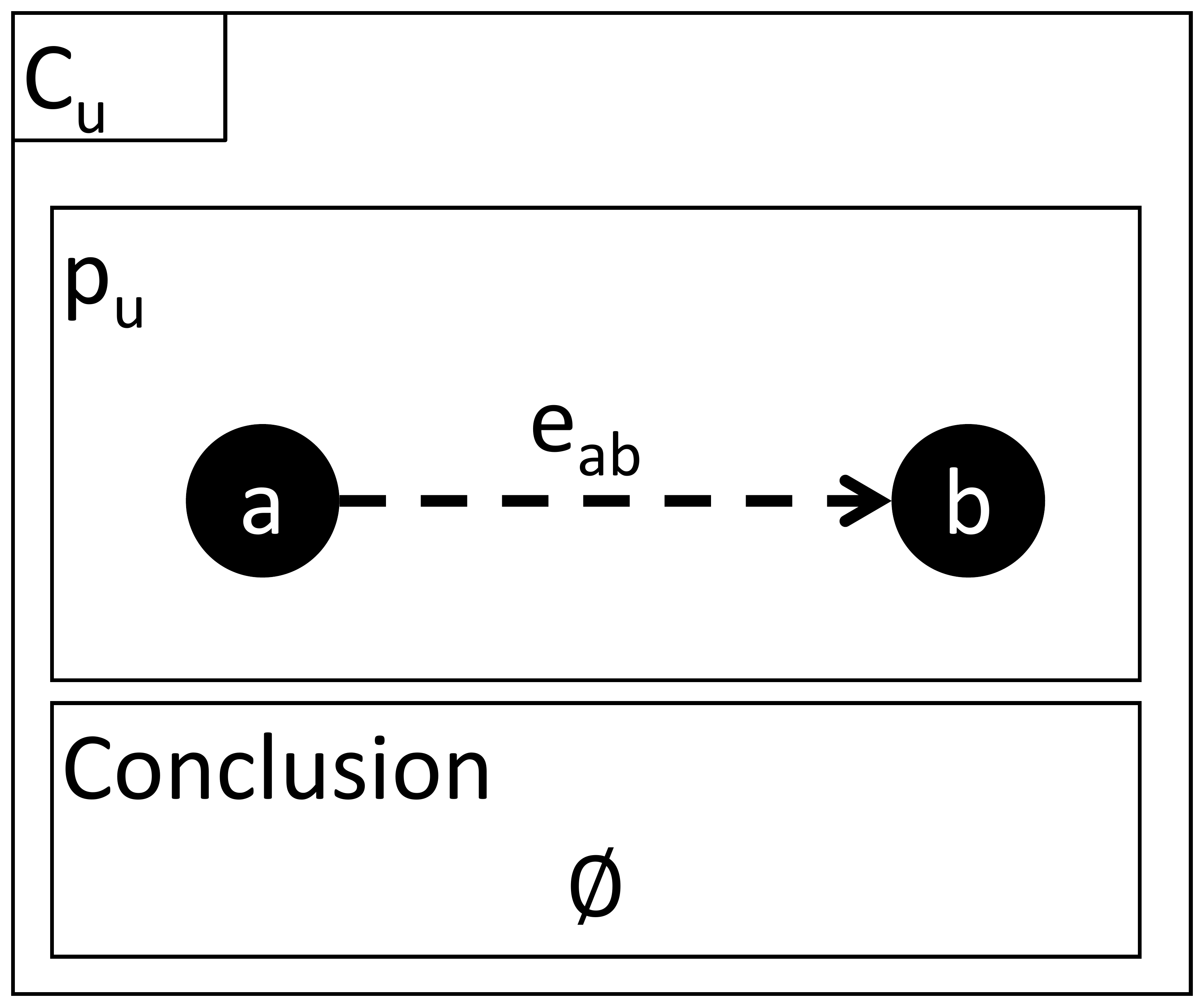}}
    \caption{Structural constraints \noParallelLinksConstraint and \noLoopsConstraint and \unclassifiedLinkConstraintLong}
\end{center}
\end{figure}
From now on, we may assume that all topologies fulfill the \noParallelLinksConstraintLong and the \noLoopsConstraintLong.
This implies that neither \premiseNoParallelLinks nor \premiseNoLoops are unsatisfiable.

\paragraph{Unclassified-link constraint}

After executing a \TC algorithm, each link must be either active or inactive, \idest, the output topology does \emph{not} contain unclassified links.
This requirement avoids situations where a \TC algorithm may immediately return without classifying all links.
The negative \unclassifiedLinkConstraintLong (shown in \Cref{fig:unclassified-link-constraint}) describes this optimization goal.
Its premise \premiseUnclassifiedLink matches each unclassified link \linkVariableab.

\paragraph{Algorithm-specific constraints}

The aforementioned graph constraints are generic in the sense that either \emph{each topology} fulfills them (structural constraints) or that \emph{each \TC algorithm} must ensure that its output topology fulfills them (unclassified-link constraint).
Additionally, each \TC algorithm has \emph{algorithm-specific graph constraints}.
In the following, we derive the algorithm-specific graph constraints of \ktc.

The specification of \ktc in~\cite{SWBM12} states that a link is inactive in the output topology of \ktc if and only if it is the unique weight-maximal link in a triangle and if its weight is at least $k$-times greater than the weight of the weight-minimal link in the triangle;
more formally:
\begin{align}\label{eqn:ktc-constraint}
\begin{split}
 \state(\linkVariableab) = \INACT 
 & \;\Leftrightarrow\; \linkVariableab \text{ is in a triangle with classified links } \linkVariableac, \linkVariablecb \\
 & \wedge\; \weight(\linkVariableab) > \max\left(\weight(\linkVariableac), \weight(\linkVariablecb)\right) \\
 & \wedge\; \weight(\linkVariableab) \geq k \cdot \min\left(\weight(\linkVariableac), \weight(\linkVariablecb)\right). 
\end{split}
\end{align}
The output topology consists of active and inactive links only because it fulfills the \unclassifiedLinkConstraintLong.
Therefore, \Cref{eqn:ktc-constraint} can be reformulated as follows:
A link is active in the output topology if and only if it is not part of such a triangle where it is the unique weight-maximal link in a triangle and if its weight is at least $k$-times greater than the weight of the weight-minimal link in the triangle;
more formally:
\begin{align}\label{eqn:ktc-constraint-active}
\begin{split}
\state(\linkVariableab) = \ACT 
& \;\Leftrightarrow\; \linkVariableab \text{ is in \emph{no} triangle with classified links } \linkVariableac,  \linkVariablecb \\
& \text{ with } \weight(\linkVariableab) > \max\left(\weight(\linkVariableac), \weight(\linkVariablecb)\right) \text{ and }\\
& \text{ with } \weight(\linkVariableab) \geq k \cdot \min\left(\weight(\linkVariableac), \weight(\linkVariablecb)\right).
\end{split}
\end{align}
The left-to-right implications of \Cref{eqn:ktc-constraint} and \Cref{eqn:ktc-constraint-active} correspond to the two \emph{\ktc-specific constraints}, shown in \Cref{fig:original-ktc-constraint} and described in the following.

The premise of the positive \emph{\inactiveLinkConstraintKTCLong} matches an inactive link \linkVariableab, and its single conclusion pattern \conclusionDI matches an inactive link \linkVariableab if it is part of a triangle (together with the classified links \linkVariableac and \linkVariablecb) where %
\begin{inparaenum}
    \item \linkVariableab has a weight greater than $\max(\weight(\linkVariableac), \weight(\linkVariablecb))$, and
    \item \linkVariableab has a weight greater than or equal to $k \cdot \min(\weight(\linkVariableac), \weight(\linkVariablecb))$.
\end{inparaenum}

The premise of the negative \emph{\activeLinkConstraintKTCLong} matches any triangle consisting of an active link \linkVariableab, and classified links \linkVariableac and \linkVariablecb where  
\begin{inparaenum}
    \item \linkVariableab has a weight greater than $\max(\weight(\linkVariableac), \linebreak \weight(\linkVariablecb))$, and
    \item \linkVariableab has a weight greater than or equal to $k \cdot \min(\weight(\linkVariableac), \weight(\linkVariablecb))$.
\end{inparaenum}
The \activeLinkConstraintKTCLong is a negative constraint because its conclusion is empty;
it is only fulfilled for topologies that contain no such triangles.
\begin{figure}
    \begin{center}
        \newcommand{\subfigwidth}{.32\textwidth}
         \subcaptionbox{\inactiveLinkConstraintKTC}[\subfigwidth]{
             \includegraphics[width=\subfigwidth]{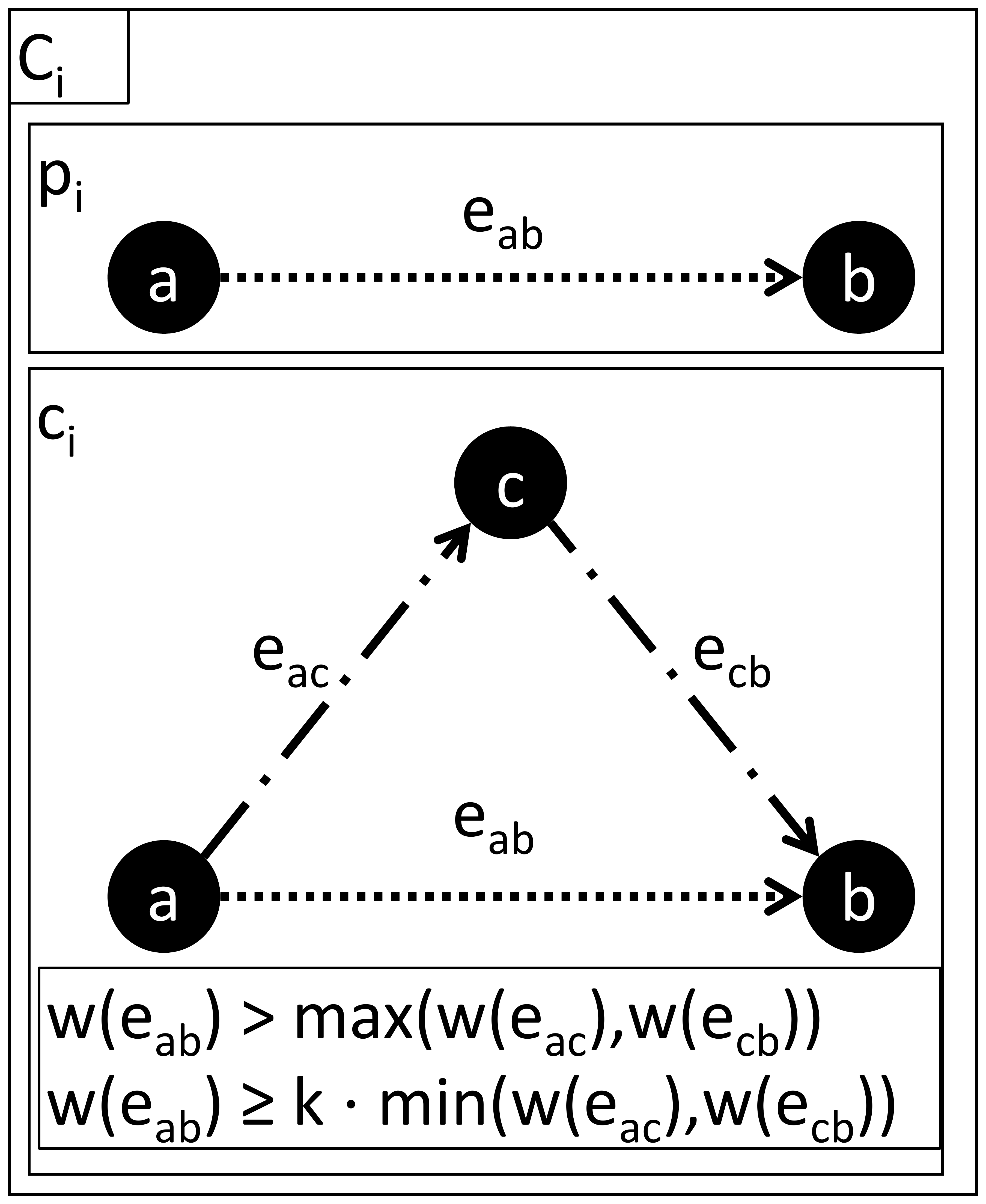}}
         \hspace{2em}
         \subcaptionbox{\activeLinkConstraintKTC}[\subfigwidth]{
             \includegraphics[width=\subfigwidth]{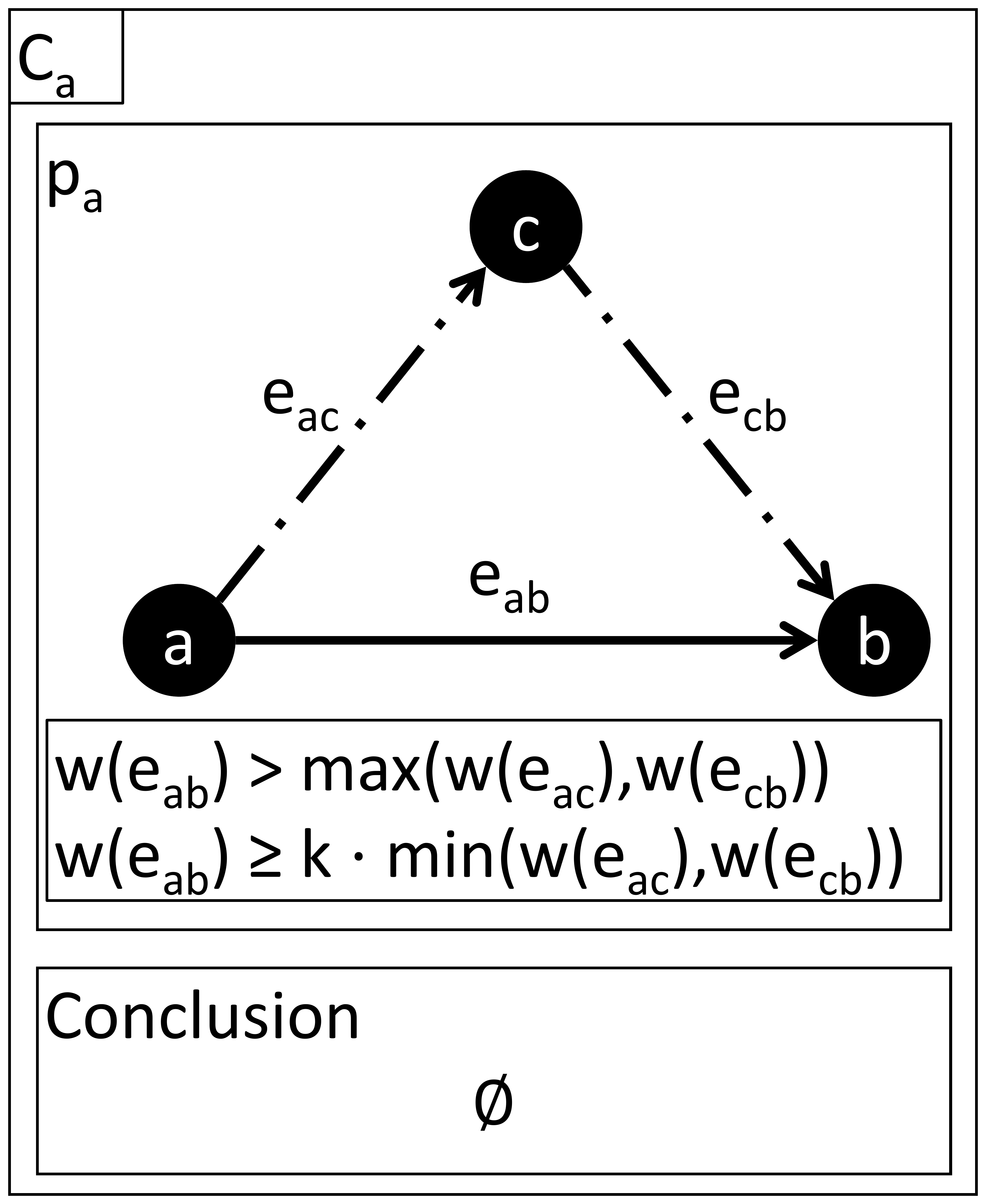}}
     \end{center}
    \caption{Positive \inactiveLinkConstraintKTCLong and negative \activeLinkConstraintKTCLong}
    \label{fig:original-ktc-constraint}
\end{figure}

To sum up, we have specified five graph constraints.
The no-parallel-links constraint \linebreak\noParallelLinksConstraint and the \noLoopsConstraintLong specify basic structural properties of topologies.
The \unclassifiedLinkConstraintLong specifies the general requirement that the output topology of a \TC algorithm should consist only of classified links,
Finally, the \ktc-specific \inactiveLinkConstraintKTCLong and  \activeLinkConstraintKTCLong jointly specify the optimization goal of the output topologies of the \TC algorithm \ktc.

\subsection{Connectivity}
\label{sec:connectivity-definition}
Connectivity is a crucial consistency constraint of \TC algorithms.
Still, connectivity cannot be described in the framework of graph constraints as defined in~\cite{HW95} because it is a non-local constraint, \idest, connectivity cannot be expressed by a finite graph constraint consisting of  premise and conclusion.
A number of theoretical frameworks for expressing non-local graph constraints exists (\eg, \cite{Fli15}), but the constructive approach will still need to be extended to these frameworks.
In this paper, we define different levels of topology connectivity, and we prove that the \TC algorithm establishes the strongest type of connectivity for weakly consistent input topologies.
\begin{my-definition}[Connectivity]
    A \emph{path} is an ordered sequence of links, where the target node of one link in the path is the source node of the next link.
    Based on this definition, we define the following three types of connectivity.
    
    A topology fulfills \emph{physical connectivity \physicalConnConstraint} if a directed path exists between any two nodes.
    The `hat' notation make connectivity constraints easily distinguishable from other graph constraints.
    
    A topology fulfills \emph{weak connectivity \weakConnConstraint} if a directed path consisting of active and unclassified links only exists between any two nodes.
        
    A topology fulfills \emph{strongly connectivity \strongConnConstraint} if a directed path consisting of active links only exists between any two nodes.
\end{my-definition}
The subsequent results immediately follow from the preceding definitions:
\begin{corollary}
    Each strongly connected topology is also weakly connected, and each weakly connected topology is also physically connected.
\end{corollary}
\begin{corollary}\label{thm:weak-and-strong-connectivity}
    A weakly connected topology that also fulfills the \unclassifiedLinkConstraintLong is strongly connected.
\end{corollary}
Physical connectivity is a structural constraint:
If the underlying topology is not physically connected, a \TC algorithm has cannot produce a connected output topology.
Weak and strong connectivity are general constraints of \TC algorithms:
If the topology is at least weakly connected, messages can still be exchanged between all nodes by treating all unclassified links as if they were active.
From Corollary~\ref{thm:weak-and-strong-connectivity} follows that, after the execution of a \TC algorithm that enforces the \unclassifiedLinkConstraintLong, the output topology is strongly connected.

\subsection{Consistency of Topologies}

A topology is \emph{consistent} \wrt a set of graph constraints if it fulfills all of these graph constraints.
In our scenario, we distinguish three types of consistency:
A topology is \emph{structurally consistent} if it fulfills the topology constraints \noLoopsConstraint and \noParallelLinksConstraint.
We require that any topology is structurally consistent.
A topology is \emph{weakly consistent} if it is structurally consistent and additionally fulfills the algorithm-specific graph constraints.
A topology is \emph{strongly consistent} if it is weakly consistent and additionally fulfills the \unclassifiedLinkConstraintLong, which means that each link is known either to be active or inactive.
Strong consistency implies weak consistency, which in turn implies structural consistency.

In case of \ktc, a topology is weakly consistent if it fulfills the topology constraints and the \ktc-specific constraints \activeLinkConstraintKTC and \inactiveLinkConstraintKTC, and it is strongly consistent if it additionally fulfills the \unclassifiedLinkConstraintLong.
\Cref{tab:overview-consistency} gives an overview of the three different notions of consistency \wrt the \ktc case study.
A check mark (\OK) indicates that the fulfillment of a particular type of consistency implies the fulfillment of a particular constraint.
\begin{table} 
\begin{center}
    \caption{Overview of graph constraints constituting structural, weak, and strong consistency for the \ktc example}
    \label{tab:overview-consistency}
    \begin{tabular}{l|ccc|ccc|cc}
        \toprule
        \textbf{Consistency}&\multicolumn{3}{c|}{\textbf{Structural Constraints}}&
        \multicolumn{3}{c|}{\textbf{\TC Constraints}}&
        \multicolumn{2}{c}{\textbf{\ktc Constr.}}\\
        &\noParallelLinksConstraint&\noLoopsConstraint&\physicalConnConstraint&
            \unclassifiedLinkConstraint&\weakConnConstraint&\strongConnConstraint&
            \activeLinkConstraintKTC&\inactiveLinkConstraintKTC\\
        \midrule
        Structural &\OK&\OK&\OK&& & && \\
        Weak  &\OK&\OK&\OK&& \OK & &\OK&\OK\\
        Strong&\OK&\OK&\OK&\OK&\OK&\OK&\OK&\OK\\
        \bottomrule
    \end{tabular}
\end{center}
\end{table}

The following theorem states a useful observation:
If we assume that the initial topology is physically connected and consists exclusively of unclassified links, then this topology is weakly consistent.
If we manage to treat \CEs appropriately, we may keep the topology at least weakly consistent.
\begin{theorem}\label{thm:unclassified-topology-fulfills-weak-consistency}
    A physically connected topology consisting entirely of unclassified links is weakly consistent.
\end{theorem}
\begin{proof-sketch}
    Assume to the contrary that the topology violates at least one of the constraints of weak consistency.
    The structural constraints are guaranteed to hold.
    Additionally, the topology is also weakly connected because its unclassified-link subtopology equals the physical topology, which is connected.
    If the topology were to violate the negative \activeLinkConstraintKTCLong, the topology would contain a match of the premise of this constraint, which implies that there would be three classified links.
    If the topology were to violate the positive \inactiveLinkConstraintKTCLong, the topology would have to contain a match of the premise of these constraints, which implies the presence of at least one inactive link.
\end{proof-sketch}

While structural consistency is guaranteed to hold for any topology, \eg, because a node will never use a link to communicate with itself, weak consistency may be violated by \CEs, which needs to be taken into account by the \TC developer.
Finally, invoking a \TC algorithm for a weakly consistent topology should always produce a strongly consistent topology.

\subsection{Proving Preservation of Connectivity for \ktc}
\label{sec:connectivity}

In the original paper that proposed \ktc, the authors proved that \ktc topologies are connected by showing that topologies produced by \ktc contain, \eg, a minimum spanning tree~\cite{SWBM12}.
In this paper, we directly prove that the output topology of \ktc is strongly connected if the input topology is physically connected.
We carry out the proof based on the \ktc-specific graph constraints \activeLinkConstraintKTC and \inactiveLinkConstraintKTC, only.
\begin{theorem}[Strong consistency and strong connectivity]\label{thm:strong-consistency-strong-connectivity}
    A strongly consistent and physically connected topology is strongly connected.
\end{theorem}
\begin{proof}
    Due to the assumption of physical connectivity, we know that a path of links (of arbitrary state) exists between any two nodes.
    Therefore, it suffices to show the following \emph{claim}: 
    The source and target node of each link are connected by a path of active links in the output topology.
    This trivially holds for the end nodes of active links.
    As the topology is strongly consistent, the \activeLinkConstraintKTCLong, the \inactiveLinkConstraintKTCLong, and the \unclassifiedLinkConstraintLong are fulfilled, which implies that the topology only contains active and inactive links and that each inactive link \lvE{} is part of a triangle of active or inactive links having a weight that is smaller than the weight of \lvE.
    
    By induction, we show that the claim also holds for all inactive links $e_{\textrm{i},z}, 1 \leq z \leq N_{\textrm{i}}$ with $N_{\textrm{i}} = |\{e \in E \mid \state(e) = \INACT\}|$:
    We consider the $N_{\textrm{i}}$ inactive links $e_{\textrm{i},1},\dots,e_{\textrm{i},N_{\textrm{i}}}$ of the topology ordered by weight, \idest, $\weightOf{e_{\textrm{i},x}} < \weightOf{e_{\textrm{i},y}} \Rightarrow x < y$.
            
    \emph{Induction start ($z = 1$):} The weight-minimal inactive link, \linkName{\textrm{i},1}, is part of a triangle with two links, \linkName{2} and \linkName{3}, which have a smaller weight and which are active because there is no inactive link with a less weight than \linkName{\textrm{i},1} (\Cref{fig:proof-connectivity-start}).
    Therefore, the path $[\linkName{2},\linkName{3}]$ connects the start node of \linkName{\textrm{i},1} with its target node, and the claim holds for link \linkName{\textrm{i},1}.
    
    \emph{Induction step ($z \to z + 1$):} We now consider an inactive link \linkName{\textrm{i},z+1} with $1 \leq z \leq N_\textrm{i}{-}1$, which is part of a triangle with the two links \linkName{2} and \linkName{3} because the \inactiveLinkConstraintKTCLong holds  (\Cref{fig:proof-connectivity-step}).
    Without loss of generality, we assume that \linkName{2} and \linkName{3} are inactive.
    There exist integers $s$ and $t$ smaller than $z{+}1$ such that $\linkName{2} = \linkName{\textrm{i},s}$ and $\linkName{3} = \linkName{\textrm{i},t}$.
    As the claim has been proved for all inactive links of less weight than \linkName{\textrm{i},z+1}, a path $p_2$ of active links connects the source node with the target node of \linkName{2}, and a path $p_3$ of active links connects the source node with the target node of \linkName{3}.
    The joined paths $p_2$ and $p_3$ connect the source node of \linkName{\textrm{i},z+1} to its target node, and the claim holds for link~\linkName{\textrm{i},z+1}
    \begin{figure}
        \begin{center}
            \newcommand{\subfigwidth}{.27\textwidth}
            \subcaptionbox{Induction start ($z=1$)
                \label{fig:proof-connectivity-start}}[.3\textwidth]{\includegraphics[width=\subfigwidth]{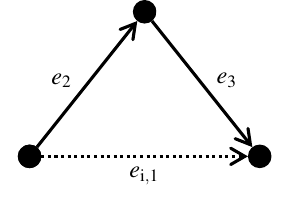}}
            \subcaptionbox{Induction step ($z \to z + 1$
                \label{fig:proof-connectivity-step})}[.3\textwidth]{\includegraphics[width=\subfigwidth]{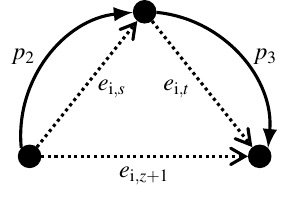}}
        \end{center}
        \caption{Sketches for the proof of \Cref{thm:strong-consistency-strong-connectivity}}
    \end{figure}
\end{proof}
\begin{theorem}[Weak consistency and weak connectivity]
    \label{thm:weak-consistency-weak-connectivity}
    A weakly consistent, physically connected topology is weakly connected.
\end{theorem}
\begin{proof-sketch}
    The proof is analogous to the proof of \Cref{thm:strong-consistency-strong-connectivity}.
    In this case, the constructed paths may contain active and unclassified links.
\end{proof-sketch}
\begin{corollary}
    The output topology of \ktc is strongly connected if its input topology is physically connected.
\end{corollary}
\begin{proof-sketch}
    The output topology of \ktc is strongly consistent.
    Strong connectivity follows from \Cref{thm:strong-consistency-strong-connectivity} and from the assumption that the input topology of the \TC algorithm is physically connected.
\end{proof-sketch}

\section{Specifying Topology Control and Context Event Rules Using Graph Transformation}
\label{sec:gratra}

In this section, we first present the basic concepts of graph transformation (\GT) rules and programmed \GT operations.
Afterwards, we illustrate these concepts by specifying concrete \TC and \CE operations as \GT rules and by specifying complete \TC algorithms using programmed \GT operations.

\subsection{Graph Transformation}

The following definitions are in concordance with standard literature in the GT community~\cite{EEPT06}.
A \emph{\GT rule} consists of a \emph{left-hand side} (\LHS) pattern, a \emph{right-hand side} (\RHS) pattern and a number of application conditions.
An application condition~(\AC{\noarg}) is a graph constraint whose premise contains the \LHS.
If the application condition is a positive constraint, it is called a \emph{positive application condition} (\PAC{\noarg});
otherwise it is called a \emph{negative application condition} (\NAC{\noarg}).
The attribute constraints of the \RHS pattern only consist of equality constraints.
A \GT rule has zero or more \emph{rule parameters}, which may be node and link variables of the \LHS or any variable that appears on the right side of an attribute constraint of the \LHS or the \RHS.

An \emph{application condition \AC{x} is fulfilled for a match $m$} of the \LHS of a \GT rule if any possible extension of $m$ to a match of the premise of \AC{x} may be extended to a match of at least one conclusion pattern of \AC{i}.
This implies that a negative application condition is fulfilled for a match $m$ if the match cannot be extended to a match of the premise of the application condition.

A \GT rule is \emph{applicable on a topology $G$} if a match $m$ of the \LHS in $G$ exists that fulfills all application conditions of the rule.
An \emph{application of a \GT rule at a match $m$ in a topology $G$} is performed as follows:
\begin{inparaenum}
\item
All nodes (links) of $G$ that have a corresponding node (link) variable in the \LHS but not in the \RHS are removed.
\item 
For each node (link) variable in the \RHS that is not in the \LHS, a new node (link) is added to $G$.
\item
Finally, the attribute constraints of the \RHS are applied:
In each attribute constraint, the expression to the right of the assignment operator is evaluated and the result is assigned to the variable to the left of the assignment operator.
The co-match $m'$ of the resulting graph $G'$ maps each node (link) variable of \RHS to a node (link) in $G'$.
\end{inparaenum}
After the successful application of a \GT rule, the node and link variables of the \RHS are \emph{bound}, \idest, the resulting co-match maps each node (link) variable in the \RHS to a fixed node (link) in $G$.
In our context, rule parameters are also bound, which means that it has a fixed value that may not be re-assigned during a rule application.
This permits us to pass the bound node (link) variables of a successful rule application as parameters to a second rule application.

In this paper, we use Story Diagrams~\cite{FNTZ98}, a programmed \GT language~\cite{GT:HandbookI}, to structure \GT rules into a control flow.
A \emph{(programmed) \GT operation} is a directed graph that consists of activity nodes and activity edges.
Additionally, it has a signature consisting of an \emph{operation name} and a set of \emph{operation parameter}s.
An \emph{activity node} may either be a start node, stop node, story node, statement node, or foreach node.
An \emph{activity edge} interconnects two activity nodes and it may be labeled either with \guardSuccess or \guardFailure.
If an activity edge is unlabeled, \guardSuccess is assumed.
\begin{itemize}
\item 
A \emph{start node} (depicted as solid circle) specifies the entry point of the control flow.
Each operation has exactly one start node, which has no incoming activity edges and one unlabeled outgoing activity edge.
\item 
A \emph{stop node} (depicted as circle with a solid center) specifies an exit point of the control flow.
Each operation has one or more stop nodes, and each stop node has at least one incoming activity edge and no outgoing activity edges.
\item 
A \emph{story node} (depicted as rounded rectangle) contains a single \GT rule.
It has at least one incoming activity edge and either one unlabeled outgoing activity edge or two outgoing activity edges labeled with \guardSuccess and \guardFailure.
\item 
A \emph{foreach node} (denoted as a stacked rounded rectangle) is a special type of story node and contains a single \GT rule.
\item 
A \emph{statement node} (depicted as rounded rectangle) contains an invocation of a \GT operation.
It has at least one incoming activity edge and one unlabeled outgoing activity edge.
\end{itemize}

When \emph{invoking} a programmed \GToperationNN, the execution begins at the start node and continues along the activity edges.
What happens if the execution arrives at a particular story node, depends on the type of activity node:
\begin{itemize}
\item 
By definition, the execution may never return to the start node of a \GToperationNN.
\item 
When the execution arrives at a stop node, the operation returns.
\item 
When the execution arrives at a story node, its \GTruleNN is applied (if possible), and if the rule application was successful (unsuccessful), the execution continues along the \guardSuccess (\guardFailure) activity edge.
If the rule application was successful, the link and node variables of the \RHS of the rule are bound and may be re-used during subsequent rule applications or operation invocations.
\item 
When the execution arrives at a foreach node, the following happens:%
\begin{myinparaenum}
    \item \emph{All} matches of the \LHS of the contained \GTruleNN are determined;
    \item The contained \GTruleNN is applied to each match, and for each successful rule application, the control flow continues along the activity edge labeled with \guardSuccess;
    \item If all matches have been processed, the execution continues along the activity edge labeled with \guardFailure and the set of collected matches is cleared.
\end{myinparaenum}
\item 
When the execution arrives at a statement node, the contained \GToperationNN is invoked.
\end{itemize}

\subsection{Examples of GT Rules and GT Operations}
In the following, we present 
\begin{inparaenum}
\item the \GTrulesNN that describe basic building blocks of a \TC algorihm,
\item the \GToperationNN that describes a basic \TC algorithm, and
\item the \GTrulesNN that describe \CEs
\end{inparaenum}

\subsubsection{\TC Rules}
\TC rules make up the basic building blocks of eacg \TC algorithm.
More specifically, according to our notion of \TC algorithms, the marking step of \emph{any} \TC algorithm can be implemented as one or more \GToperationsNN that consist of the three \TC operations link activation, link inactivation, and link unclassification.
The latter is required, \eg, to revert unfavorable decisions or to propagate the unclassified flag to links that are affected by a \CE in a non-local way.
The corresponding \GT rules \activationRule, \inactivationRule, and \unclassificationRule are shown in \Cref{fig:topology-control-rules}, along with the two auxiliary rules \findUnclassifiedLinkRule and \findClassifiedLinkRule.
\begin{figure}
    \begin{center}
        \includegraphics[width=\textwidth]{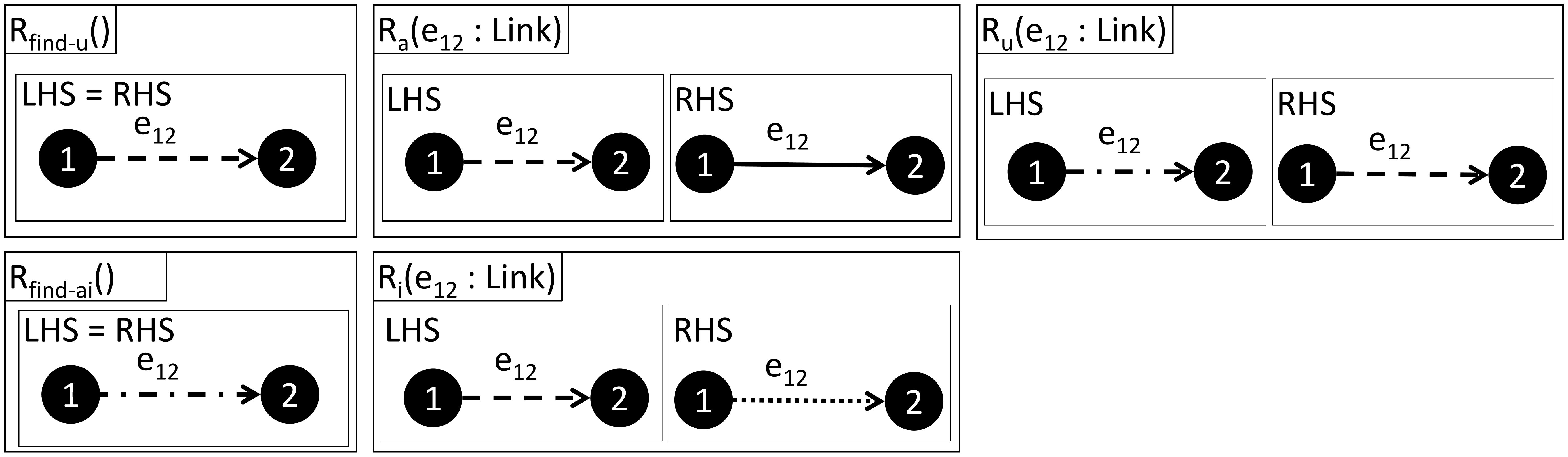}
        \caption{Topology control rules}
        \label{fig:topology-control-rules}
    \end{center}
\end{figure}
\begin{itemize}
    \item The \emph{\findUnclassifiedLinkRuleLong} identifies some unclassified link~\linkVariableOneTwo in the topology.
    As in this case, we use a compact notation (\textsf{LHS=RHS}) that only depicts one pattern if \LHS and the \RHS are identical.
    
    \item The \emph{\findClassifiedLinkRuleLong} identifies some classified link~\linkVariableOneTwo in the topology, \idest, a link that is either active or inactive.
    This rule solely serves for illustrating a shorthand notation, applied in the following:
    In a strict sense, the \RHS pattern of \findClassifiedLinkRule is invalid because it contains the attribute constraint $\linkVariableOneTwo \in \{\ACT,\INACT\}$.
    In the following, we follow the convention that this attribute constraint only applies to the \LHS of the \GTruleNN.
        
    \item The \emph{\activationRuleLong} activates a given unclassified link~\linkVariableOneTwo.
    
    \item The \emph{\inactivationRuleLong} inactivates a given unclassified link~\linkVariableOneTwo.
    
    \item The \emph{\unclassificationRuleLong} unclassifies a given classified (\idest, active or inactive) link~\linkVariableOneTwo.
\end{itemize}

\subsubsection{\GT Operation \basicTCOperation: A Basic \TC Algorithm}
\Cref{fig:topology-control-algorithm-basic} shows the \GToperationNN \basicTCOperation, which serves as starting point for implementing the full specification of \ktc in the following.
\basicTCOperation processes all unclassified links in a topology.
For each unclassified link \linkVariableOneTwo, the \activationRuleLong is applied first if possible.
If the application of \activationRule was successful, the next unclassified link is processed,
otherwise the \inactivationRuleLong is applied, if possible, and the execution returns to the story node containing the \findUnclassifiedLinkRuleLong.
We emphasize that for most of the following explanations, the order of applications of \activationRule and \inactivationRule is arbitrary.
\begin{figure}
    \begin{center}
        \includegraphics[width=\textwidth]{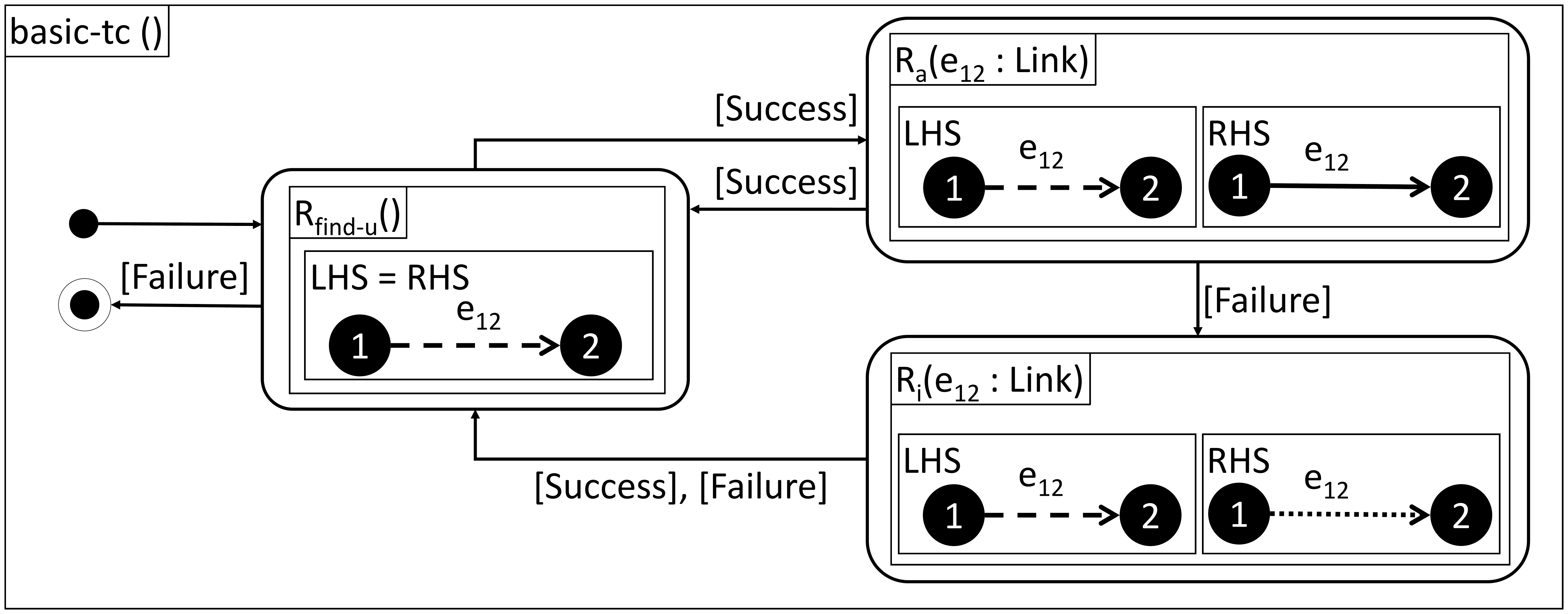}
    \end{center}
    \caption{\GToperationNN \basicTCOperation, a basic \TC algorithm that serves as template for further refinement}
    \label{fig:topology-control-algorithm-basic}
\end{figure}
In this example, the \activationRuleLong is always applicable because \linkVariableOneTwo is bound by \findUnclassifiedLinkRule and passed as parameter to \activationRule.
Therefore, the \inactivationRuleLong is never tried in \basicTCOperation.
The basic \TC algorithm \basicTCOperation
\begin{inparaenum}
\item 
enforces the \unclassifiedLinkConstraintLong because the execution may only reach the stop node if there are no more unclassified links in the topology,
\item
preserves the \inactiveLinkConstraintKTCLong because no link will ever be inactivated, but
\item 
fails to preserve the \activeLinkConstraintKTCLong because links are activated unconditionally.
\end{inparaenum}

\subsubsection{Context Event Rules}
\Cref{fig:context-event-rules} shows the \emph{context event (CE) rules}, which specify the possible modifications of the topology that are caused by the environment.
\begin{figure}
    \begin{center}
         \includegraphics[width=\textwidth]{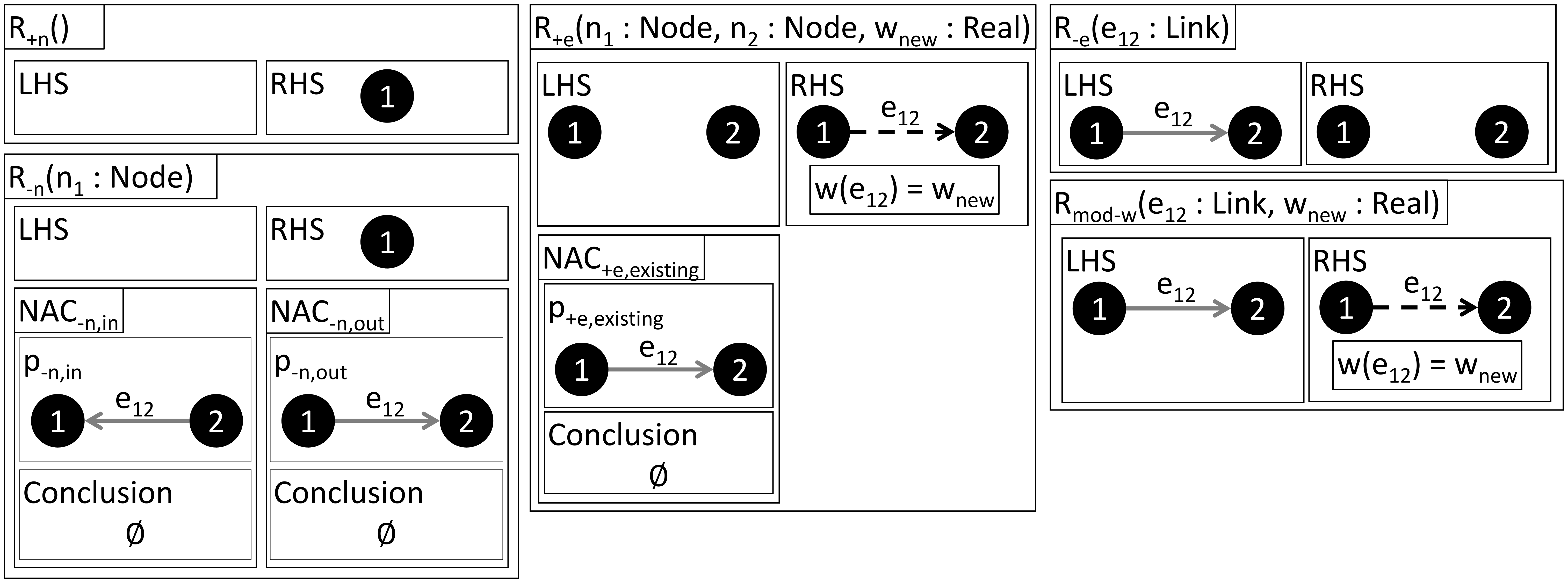}
        \caption{Context event rules}
        \label{fig:context-event-rules}
    \end{center}
\end{figure}
\begin{itemize}
    \item The \emph{node addition rule} \nodeAdditionRule creates a new \nodeNameLong{1} and adds it to the topology by establishing a corresponding \refToFig{nodes-topology} association between the node and the topology, which is not represented in the concrete syntax (as noted before).
    The empty \LHS indicates that the rule is always applicable.
        
    \item The \emph{node removal rule} \nodeRemovalRule removes a single \nodeNameLong{1} from the topology by deleting
    \begin{inparaenum}
        \item the corresponding \refToFig{nodes-topology} association and 
        \item the \nodeNameLong{1} itself.
    \end{inparaenum}
    For conciseness, we write out \nodeName{1} only in the parameter list, but abbreviate it as encircled 1 inside the rule.
    The two negative application conditions, \NAC{\text{-n},\text{in}} and \NAC{\text{-n},\text{out}}, prevent dangling links by ensuring that the node to be removed has neither incoming nor outgoing links.
    
    \item The \emph{link addition rule} \linkAdditionRule creates an unclassified link with the given source \nodeNameLong{1} and target \nodeNameLong{2}, assigns it the given weight \varWnew, and adds it to the topology by creating a corresponding \refToFig{links-topology} association between the link and the topology.
    The negative application condition \NAC{\text{+e,existing}} ensures that a rule application does not insert a parallel link from \nodeNameLong{1} to \nodeNameLong{2}.
    
    \item The \emph{link removal rule} \linkRemovalRule removes a given link \linkVariableOneTwo from the topology by deleting 
    \begin{inparaenum}
    \item the corresponding \refToFig{links-topology} association and 
    \item the link itself.
    \end{inparaenum}

    \item The \emph{weight modification rule} \weightModificationRule sets the weight of the given link \linkVariableOneTwo to \varWnew and unclassifies it.
\end{itemize}

\subsection{Constraint and Consistency Preservation}
\label{sec:constraint-preservation}
In the following, we introduce the concept of constraint preservation to describe how applying a \GT rule affects the consistency of a topology:
An \emph{application of a \GT rule \GTrule{x} preserves a graph constraint \constraint{y}} on a topology that fulfills \constraint{y} if the topology still fulfills~\constraint{y} after applying \GTrule{x}.
A \emph{\GT rule \GTrule{x} preserves a graph constraint~\constraint{y}} if \GTrule{x} preserves~\constraint{y} on any topology.
These definitions can be lifted to the concept of consistency in a natural way:
A \emph{\GT rule \GTrule{x} preserves weak (strong) consistency} if it preserves all graph constraints that make up weak (strong) consistency.
An application condition \AC{z} of a \GT rule \GTrule{x} is \emph{redundant \wrt a graph constraint \constraint{y}} if \AC{z} is always fulfilled under the assumption that \constraint{y} holds prior to each application of \GTrule{x}.
This means that we may remove \AC{z} from \GTrule{x} without threatening consistency preservation.

Given a topology that fulfills a \emph{positive} graph constraint, a rule application preserves the constraint on the topology
\begin{inparaenum}
    \item if each new match of its premise may be extended to a match of the conclusion, and 
    \item if each old match of its premise that also exists after the rule application may still be extended to a match of its conclusion after the rule application.
\end{inparaenum} 
Accordingly, given a topology that fulfills a \emph{negative} graph constraint, a rule application preserves the constraint on the topology if the rule application does not create a new match of the premise of the constraint.

\paragraph{Preservation of structural constraints}
The structural constraints \noParallelLinksConstraint and \noLoopsConstraint may only be violated by adding a link, \idest, by applying the \linkAdditionRuleLong.
This rule, however, has already been defined to preserve these constraints:
\begin{inparaenum}
\item The negative application condition \NAC{\text{+e,existing}} prevents any applications of \linkAdditionRule that would result in a parallel link between \nodeNameLong{1} and \nodeNameLong{2}, and
\item the required injectivity of matches of the \LHS enforces that the \nodeNameLong{1} and \nodeNameLong{2} are distinct, thereby preventing the new link \linkVariableOneTwo to have identical source and target nodes.
\end{inparaenum}

\paragraph{Examples of constraint violation}
\Cref{fig:example-constraint-preservation} shows four examples that illustrate how applying one of the \TC or \CE rules leads to a constraint violation.
These violations clearly call for a modification of the rules to avoid constraint violations.
In all examples, $k$ equals $2$.

In \Cref{fig:example-constraint-preservation-A}, the topology that results from applying the \activationRuleLong to \linkNameLong{12} violates the \activeLinkConstraintKTCLong because the weight of \linkName{12} is more than $k$-times greater than the weight of the weight-minimal link \linkName{13}.

In \Cref{fig:example-constraint-preservation-B}, the topology that results from applying the \linkRemovalRuleLong to \linkNameLong{13} violates the \inactiveLinkConstraintKTCLong because the match of the premise of the constraint can no longer be extended to a match of its conclusion.

In \Cref{fig:example-constraint-preservation-C}, the topology that results from applying the \inactivationRuleLong to \linkNameLong{13} violates the \activeLinkConstraintKTCLong.
Even though the links formed a triangle prior to applying \inactivationRule, the \inactiveLinkConstraintKTCLong was not violated because \linkName{13} was unclassified.

In \Cref{fig:example-constraint-preservation-D}, the topology that results from applying the \weightModificationRuleLong to \linkNameLong{13} with $\varWnew=3$ violates the \inactiveLinkConstraintKTCLong because the single match of the conclusion that extends the match of the premise is destroyed by the unclassification of link \linkName{13},
Remarkably, the \inactiveLinkConstraintKTCLong would \emph{not} have been violated if \linkName{13} had remained active because the weight of \linkName{13} has even decreased from 5 to 3.
\begin{figure}
\begin{center}
    \subcaptionbox{Violation of \activeLinkConstraintKTC caused by applying \activationRule \label{fig:example-constraint-preservation-A}}[.45\textwidth]{        \includegraphics[width=.45\textwidth]{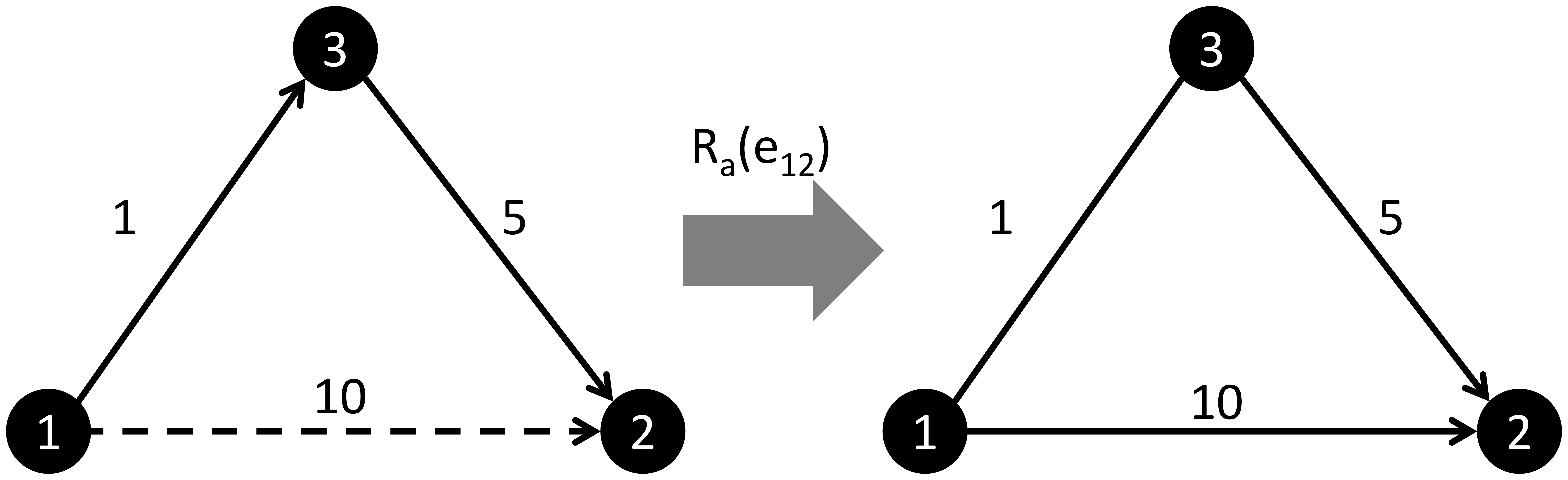}}
    \hspace{1em}
    \subcaptionbox{Violation of \inactiveLinkConstraintKTC caused by applying \linkRemovalRule\label{fig:example-constraint-preservation-B}}[.45\textwidth]{        \includegraphics[width=.45\textwidth]{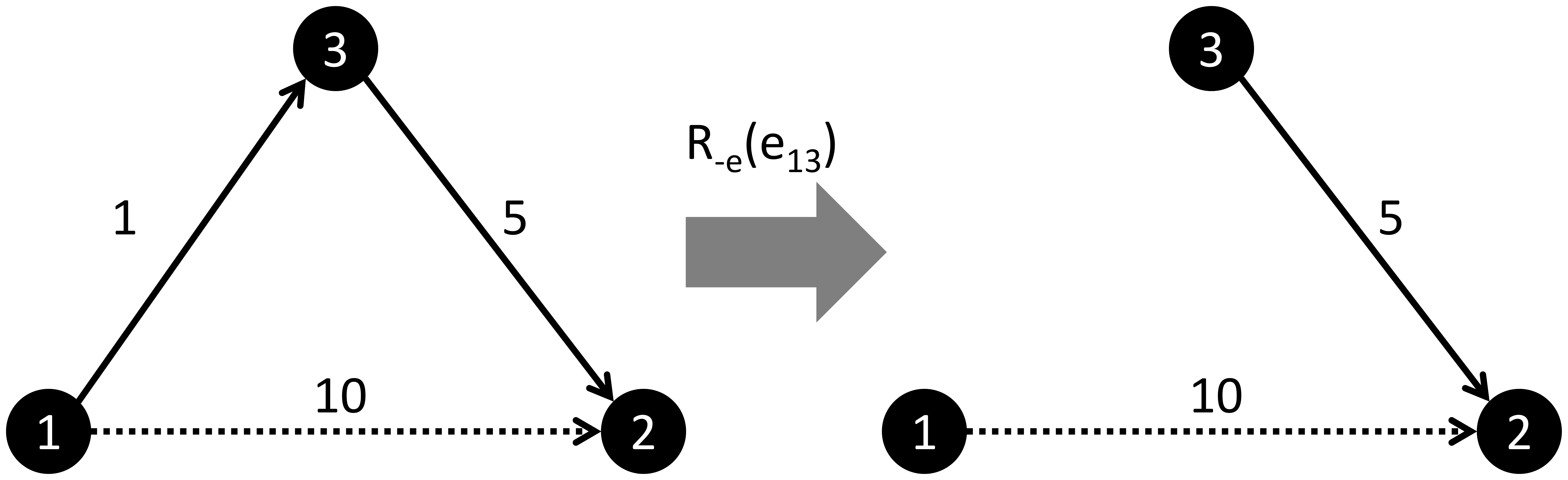}}
    
    \vspace{1ex}
    
    \subcaptionbox{Violation of  \activeLinkConstraintKTC caused by applying \inactivationRule\label{fig:example-constraint-preservation-C}}[.45\textwidth]{        \includegraphics[width=.45\textwidth]{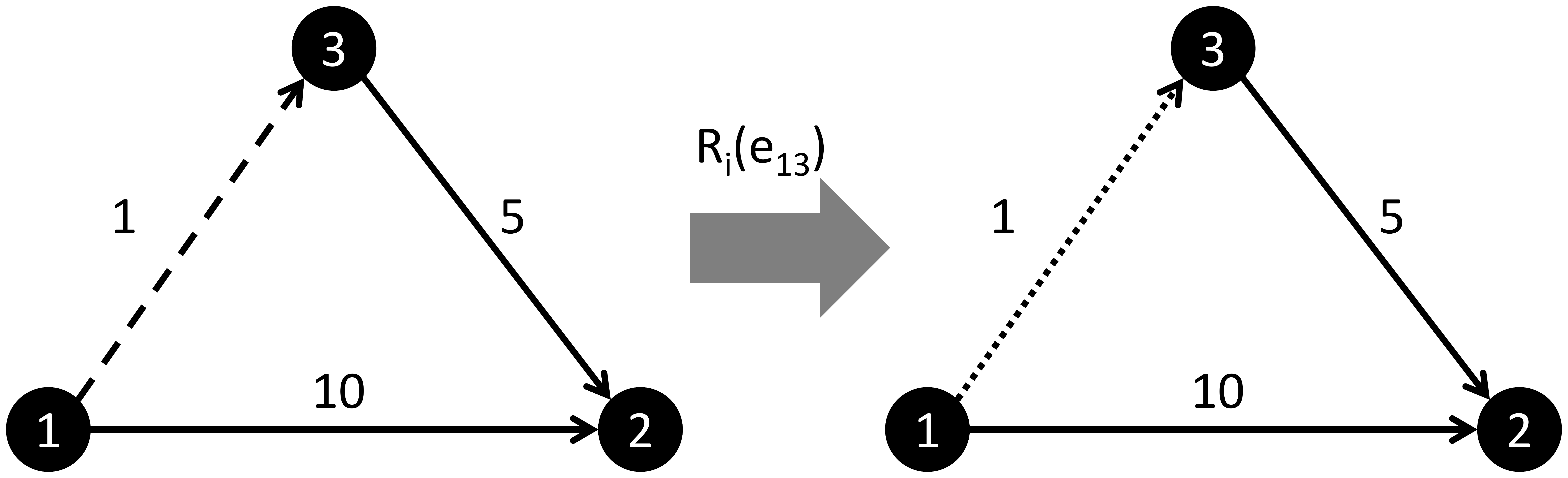}}
    \hspace{1em}
    \subcaptionbox{Violation of  \inactiveLinkConstraintKTC caused by applying \weightModificationRule\label{fig:example-constraint-preservation-D}}[.45\textwidth]{        \includegraphics[width=.45\textwidth]{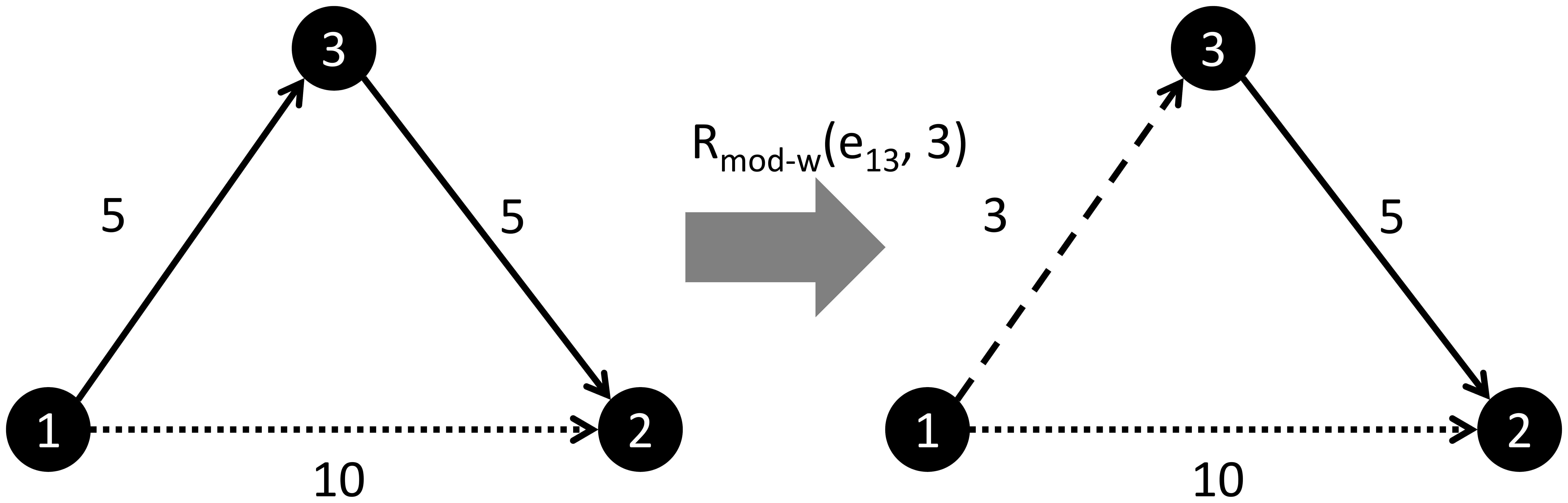}}
    \caption{Examples of constraint violations caused by rule applications ($k=2$)}
    \label{fig:example-constraint-preservation}
\end{center}
\end{figure}

\section{Refining Topology Control and Context Events Rules to Preserve Graph Constraints}
\label{sec:rule-refinement}

The examples at the end of \Cref{sec:constraint-preservation} clearly indicate that \basicTCOperation, the template \GToperationNN for developing \ktc, enforces the \unclassifiedLinkConstraintLong, but fails to preserve even weak consistency.
In this section, the \TC and \CE rules and the \GToperationNN \basicTCOperation are refined to achieve the following goals.
\begin{enumerate}
\item[\textbf{G1}] The refined \TC rules and \CE rules preserve weak consistency.
\item[\textbf{G2}] Each refined unrestrictable rule is applicable whenever its corresponding original rule is applicable (see \Cref{sec:restrictability}).
\item[\textbf{G3}] The refined \GToperationNN \basicTCOperation turns each weakly consistent input topology into a strongly consistent output topology.
\end{enumerate}
After accomplishing these goals, \basicTCOperation will be a valid \GT-based specification of \ktc.
For this reason, we will refer to the refined \GToperationNN \basicTCOperation as \tcOperation from now on.
The structure of this section is depicted in \Cref{fig:refinement-section-overview} and is shortly summarized in the following paragraphs.
\begin{figure}
    \begin{center}
    	\resizebox{\textwidth}{!}{%
    	\begin{tikzpicture}[node distance = 2.6cm]
    	
    	\sffamily
    
    	\node [rc] (1) at (0,0) {\small \textbf{Refinement algorithm}\\Approach and application\\ to \TC and CE rules\\ \lbrack Sec. \ref{sec:refinement-methodology}\rbrack};
    	
    	\node[bigArrow, right of=1](arrow1) {};
    	
    	\node[rc, right of=arrow1, node distance=2.9cm](2) {\small \textbf{Handler operations}\\Approach and application\\to unrestrictable rules\\ \lbrack Sec. \ref{sec:enforcing-applicability}\rbrack};
    	
    	\node[bigArrow, right of=2](arrow2) {};
    	
    	\node[rc, right of=arrow2, node distance=2.9cm](3) {\small \textbf{Pre-processing operations} for enforcing termination\\ \lbrack Sec. \ref{sec:enforcing-termination}\rbrack};
    	
    	\end{tikzpicture}}%
    \end{center}
    \caption{Structure of \Cref{sec:rule-refinement}}
    \label{fig:refinement-section-overview}
\end{figure}
Goal G1 aims at deriving the so-called \emph{weakest precondition}~\cite{Dijkstra1976} for each \GT rule and can be achieved by applying the constructive approach first presented in \cite{HW95} and refined later for attributes in \cite{DV14}.
This approach takes a set of \GTrulesNN and graph constraints as input and refines each \GTruleNN to ensure that all \GTrulesNN preserve all the graph constraints.
In \Cref{sec:refinement-methodology}, we describe the refinement algorithm in detail and apply it to the \TC and \CE rules, finally achieving goal G1.
However, only applying the constructive approach is not sufficient in our application scenario for the following reasons.
\begin{enumerate}
\item
The additional application conditions, which have been introduced during the refinement of the \TC and \CE rules, restrict their applicability.
While such restrictions are perfectly reasonable (and desirable) for \TC rules, restricting \CE rules is unrealistic because these rules represent (non-controllable!) modifications of the topology that are caused by the environment.
For similar reasons, the applicability of the \unclassificationRuleLong may not be restricted.
In the following, we subsume the \CE rules and the \unclassificationRuleLong under the term \emph{unrestrictable (GT) rules}.

In \Cref{sec:enforcing-applicability}, we propose to translate the application conditions of the unrestrictable rules that have been introduced during the rule refinement into appropriate handler (GT) operations.
This step is one of the major contributions of this paper.
In the course of this step, the additional application conditions of the unrestrictable \GTrulesNN are dropped under the assumption that the corresponding handler operation is always invoked immediately after the application of a \CE rule;
upon the completion of this step, goal G2 will be fulfilled.
\item
The additional application conditions of the \activationRuleLong and the \inactivationRuleLong may (and will) lead to situations where the currently considered unclassified link \linkVariableOneTwo can be processed by neither of these rules.
For this reason, \tcOperation may not terminate in certain situations.

In \Cref{sec:enforcing-termination}, we propose to insert a handler operation that is invoked prior to applying the \activationRuleLong or the \inactivationRuleLong.
This operation systematically unclassifies links that would prevent the application of both the \activationRuleLong and the \inactivationRuleLong.
We prove that the modified \GToperationNN \tcOperation always terminates, thereby achieving goal G3.
\end{enumerate}

\subsection{Refining Graph Transformation Rules to Preserve Graph Constraints}
\label{sec:refinement-methodology}

In this section, we tackle goal G1, the preservation of weak consistency, and introduce the refinement algorithm that builds on the constructive approach presented in \cite{HW95,DV14}.
We start by introducing basic concepts.
Then, we describe the general algorithm and apply it to the \TC and \CE rules.

\subsubsection{Basic Concepts: Post-conditions of \GT Rules and Gluings of Patterns}
\label{sec:refinement-basics}
In \Cref{sec:gratra}, we have introduced application conditions, which restrict the applicability of \GT rules and which are checked before the corresponding rule modifies any match. 
Therefore, application conditions are also called \emph{pre-condition}s. 
In the general case, a \GT rule may also be equipped with a set of \emph{post-condition}s, which have the same structure as pre-conditions, but are checked \emph{after} the application of their rules.
The premise of a post-condition extends the \RHS of its rule in the same fashion as the premise of a pre-condition extends the \LHS.
Conclusions of pre- and post-conditions extend their premises as explained beforehand.

Any violation of a pre-condition or a post-condition of a rule blocks its execution; 
as a consequence, a \GT engine has to roll back the application of a rule if this application results in a violation of one of its post-conditions.
Fortunately, any post-condition of a rule can be translated into an equivalent pre-condition (and vice versa).
It is even possible to translate arbitrary graph constraints into post- or pre-conditions of a given \GT rule such that this rule preserves these graph constraints;
this is the fundamental idea of the constructive approach in~\cite{HW95}.

The approach relies on the construction of so-called gluings, which are combinations of the \RHS (\LHS) of a \GTruleNN and the premises and conclusions of the regarded graph constraints.
More precisely, a \emph{gluing} \gluing{\ell,r} of patterns \leftPattern and \rightPattern is a graph pattern that represents one possible overlap of \leftPattern and \rightPattern.
We denote variables of \leftPattern with numbers ($1$, $2$, etc.), variables of \rightPattern with lowercase letters ($a$, $b$, etc.), and variables of \gluing{\ell, r} with uppercase letters plus their corresponding variables in \leftPattern and \rightPattern (\eg, $A[1,a]$).
Injective mappings \leftMapping and \rightMapping exist that assign to each node (link) variable of \leftPattern and \rightPattern a corresponding node (link) variable in~\gluing{\ell, r}.
Each node (link) variable in \gluing{\ell, r} has at least one corresponding node (link) variable in \leftPattern and \rightPattern, and at least one node variable of \gluing{\ell, r} must correspond to node variables in \leftPattern and \rightPattern, \idest, the patterns \leftPattern and \rightPattern must indeed overlap.
The attribute constraints of \gluing{\ell, r} correspond to the joint attribute constraints of \leftPattern and \rightPattern with appropriately relabeled node (link) variables.

For example, \Cref{fig:rule-refinement-unclrule-inactconstraint-gluings} (on page~\pageref{fig:rule-refinement-unclrule-inactconstraint-gluings}) and \Cref{tab:rule-refinement-unclrule-inactconstraint-gluings} show all six possible gluings \gluingUI{z\,,\,1 \leq z \leq 6} of the \RHS of the \unclassificationRuleLong and the premise of the \inactiveLinkConstraintKTCLong using a graphical and a compact tabular representation, respectively.
Each gluing has at most three node variables because each original pattern has two node variables and because at least one node variable of the two patterns must be glued together.
Each column of \Cref{tab:rule-refinement-unclrule-inactconstraint-gluings} corresponds to one gluing \gluingUI{k}, and each row corresponds to one of the up to three node variables \nodeVariableA, \nodeVariableB, and \nodeVariableC of the gluing.
In this example, the node variable \nodeVariableOne of the \unclassificationRuleLong and the node variable \nodeVariablea of the \inactiveLinkConstraintKTCLong are mapped to the node variable \nodeVariableA of gluing \gluingUI{1}.
Note that the resulting set of attribute constraints of gluing \gluingUI{1} is unsatisfiable: $\state(\linkName{AB}) = \ACT$ (resulting from \unclassificationRule) and $\state(\linkName{AB}) = \INACT$ (resulting from \inactiveLinkConstraintKTC).

\subsubsection{Description of Refinement Algorithm}
In~\cite{HW95}, Heckel and Wagner describe the constructive approach for enforcing constraint preservation of \GT rules, which forms the centerpiece of the first step in our approach.
On an abstract level, the \emph{refinement algorithm} works by iteratively considering all pairs of rules and constraints:
A rule \GTrule{x} and a graph constraint \constraint{y} serve as input to the \emph{refinement procedure}.
The refinement procedure results in zero or more additional application conditions of \GTrule{x}
These application conditions prevent any application of \GTrule{x} that may produce violations of \constraint{y}.
In our scenario, we may assume that \TC or \CE rules are only applied to topologies that are weakly consistent.
\Cref{fig:refinement-procedure-overview} gives a detailed overview of the refinement procedure of a particular rule \GTrule{x}  and a particular constraint \constraint{y}, which works as follows:
\begin{enumerate}[(1)]
    \item We construct all gluings \gluing{x,y,z} of the \RHS of rule \GTrule{x} and of the premise of constraint \constraint{y}.
    
    \item 
    Each \gluingLong{x,y,z} serves as the premise of a new post-condition \postcondition{x,y,z}.
    Then, the premise is extended to the \emph{basic conclusion pattern}~\conclusionPC{x,y,z,1} of the post-condition by extending the premise with all node and link variables and the attribute constraints that are part of the conclusion but not of the premise of the \constraintLong{j}.
    
    \item The set of \emph{reduced conclusion patterns}~\conclusionPC{x,y,z,r\,,\,r \geq 2} of the post-condition is the set of all patterns that results from merging one or more node variables of the basic \conclusionPCLong{x,y,z,1}.
    We may only merge node variables that are not contained in the premise.
    The incident link variables of merged node variables are merged accordingly.
    The \emph{conclusion}~\conclusionPC{x,y,z} of post-condition~\postcondition{x,y,z} consists of the basic conclusion pattern and the set of reduced conclusion patterns~\conclusion{x,y,z}{}.
    
    \item
    An \emph{application condition} \AC{x,y,z} is obtained from a post-condition \postcondition{x,y,z} by applying rule~\GTrule{x} in reverse order to the premise and to each conclusion pattern of \postcondition{x,y,z}.
    All node and link variables and all attribute constraints that appear in the \RHS but not in the \LHS of \GTrule{x} are removed from~\AC{x,y,z}, and copies of all node and link variables and attribute constraints that appear in the \LHS but not in the \RHS of \GTrule{x} are added to~\AC{x,y,z}.
\end{enumerate}
\begin{figure}
    \begin{center}
        \includegraphics[width=\textwidth]{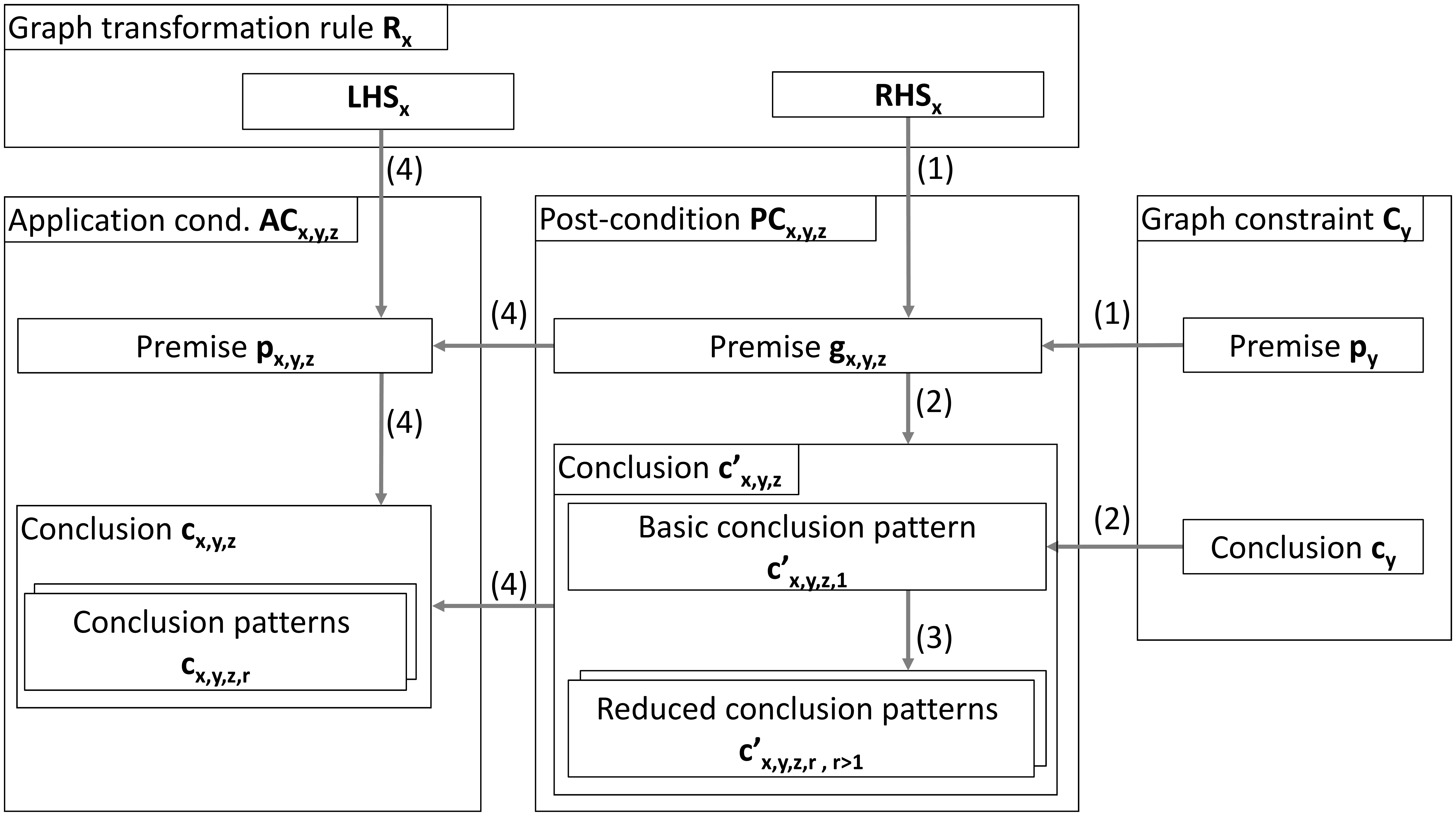}
    \end{center}
    \caption{Overview of the four steps of the rule refinement procedure}
    \label{fig:refinement-procedure-overview}
\end{figure}
To sum up, the refinement procedure first translates the graph constraint to post-conditions of the \GT rule in steps (1)--(3), which are then translated to equivalent pre-conditions in step (4).
We categorize gluings into three groups based on the satisfiability of their corresponding application condition:
\begin{itemize}
\item
\emph{Unsatisfiable gluing:}
If it is impossible to find a match of a gluing in any topology, we say that this gluing is \emph{unsatisfiable}.
For instance, a gluing that contains contradictory link state constraints for a particular link variable is unsatisfiable.
An unsatisfiable gluing \gluing{x,y,z} corresponds to an application condition that is always fulfilled because no match of the \LHS of the rule may ever be extended to a match of premise \premise{x,y,z}{}.
As explained in \Cref{sec:refinement-basics}, \gluingUI{1} is an example of an unsatisfiable gluing.

\item
\emph{Redundant application condition:}
An application condition of a \GT rule is \emph{redundant} if is fulfilled under the assumption that the structural graph constraints (\noParallelLinksConstraint and \noLoopsConstraint) are fulfilled prior to any application of the rule.

\item
\emph{Restrictive application condition:}
Any application condition that does not fall into one of the two above mentioned groups truly restricts the applicability of a rule, \idest, it is a (truly) restrictive application condition.
\end{itemize}
From these explanations, it is clear that we only need to add restrictive application conditions to the \GT rules and discard application conditions that originate from an unsatisfiable gluing or are redundant.

The rule refinement procedure sketched above has the following properties, which are essential for the development of a \TC algorithm that turns a weakly into a strongly consistent topology~\cite{HW95}:
\begin{inparaenum}
\item The added application conditions are \emph{strong enough} to prevent any application of \GTruleLong{x} that would transform a topology $G$ fulfilling \constraint{y} into a topology~$G'$ that violates \constraint{y}, \idest, they are sufficient.
\item The added application conditions are \emph{weak enough} to allow any application of \GTruleLong{x} that does not transform a topology~$G$  fulfilling \constraint{y} into a topology~$G'$ that violates \constraint{y}, \idest, they are necessary.
\end{inparaenum}

\Cref{tab:rule-refinement-subsection-overview} provides an overview of the subsequent sections that describe the application of the refinement algorithm to the \GTrulesNN.
\begin{table}
\begin{center}
    \caption{Overview of iterations of the refinement algorithm}
    \label{tab:rule-refinement-subsection-overview}
\begin{tabular}{p{.48\textwidth}|p{.25\textwidth}p{.25\textwidth}}
\toprule

\textbf{Rules} & \multicolumn{1}{c}{\textbf{\inactiveLinkConstraintKTC}} & \multicolumn{1}{c}{\textbf{\activeLinkConstraintKTC}} \\
 
\midrule

\textbf{\TC} (\unclassificationRule, \activationRule, \inactivationRule)
& \multicolumn{1}{c}{\Cref{sec:refinement-tc-inact}}
& \multicolumn{1}{c}{\Cref{sec:refinement-tc-act}}
\\
\textbf{Context events} (\nodeAdditionRule, \nodeRemovalRule, \linkAdditionRule, \linkRemovalRule, \weightModificationRule) 
& \multicolumn{1}{c}{\Cref{sec:refinement-ce-inact}}
& \multicolumn{1}{c}{\Cref{sec:refinement-ce-act}}
\\

\bottomrule

\end{tabular}
\end{center}
\end{table}

\subsubsection{Refining Topology Control Rules to Preserve the Inactive-Link Constraint}
\label{sec:refinement-tc-inact}

We begin with refining the \TC rules, \idest, the \unclassificationRuleLong, the \activationRuleLong, and the \inactivationRuleLong, to preserve the positive \inactiveLinkConstraintKTCLong.

\refinementPair{\unclassificationRule}{\inactiveLinkConstraintKTC}
The six resulting gluings \gluingUI{z\,,\,1 \leq z \leq  6} of the \RHS of the \unclassificationRuleLong and the premise of the \inactiveLinkConstraintKTCLong (the result of \refinementStep{1}) are shown in \Cref{fig:rule-refinement-unclrule-inactconstraint-gluings} and \Cref{tab:rule-refinement-unclrule-inactconstraint-gluings} using a graphical and a compact tabular representation, respectively.%
\begin{figure}
    \begin{center}
        \includegraphics[width=\textwidth]{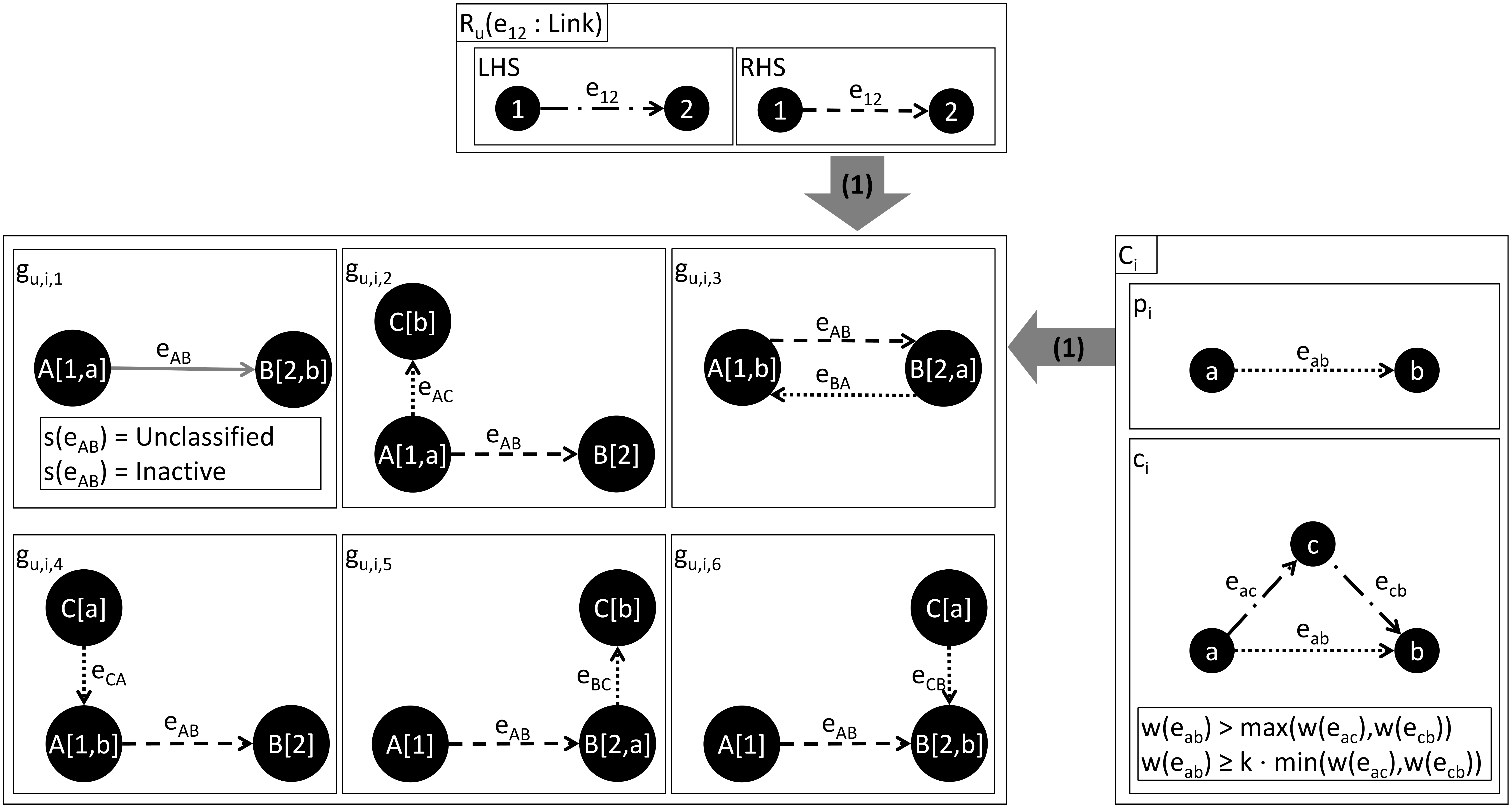}
    \end{center}
    \caption{Gluings \gluingUI{x} of the \RHS of the \unclassificationRuleLong and the premise of the \inactiveLinkConstraintKTCLong}
    \label{fig:rule-refinement-unclrule-inactconstraint-gluings}
\end{figure}
\begin{table}
    \begin{center}
        \caption{Overview of gluings of the \RHS of the \unclassificationRuleLong and the \inactiveLinkConstraintKTCLong}
        \label{tab:rule-refinement-unclrule-inactconstraint-gluings}
        \begin{tabular}{c|cccccc}
            \toprule
            \textbf{Variables in \unclassificationRule and Gluing \gluingUI{x}} &
            \multicolumn{6}{c}{\textbf{Origin in \inactiveLinkConstraintKTC} (per Gluing)}
            \\
            &\textbf{\gluingUI{1}}
            &\textbf{\gluingUI{2}}
            &\textbf{\gluingUI{3}}
            &\textbf{\gluingUI{4}}
            &\textbf{\gluingUI{5}}
            &\textbf{\gluingUI{6}}\\
            \midrule
            \textbf{$A=\leftMapping(1)$}&a&a&b&b&-&-\\
            \textbf{$B=\leftMapping(2)$}&b&-&a&-&a&b\\
            \midrule
            \nodeVariableC&-&b&-&a&b&a\\
            \bottomrule
        \end{tabular}    
    \end{center}
\end{table}
The gluings can be categorized into the aforementioned three categories as follows.
\begin{itemize}
    \item 
    \emph{Unsatisfiable gluing:}
    \GluingUILong{1} is unsatisfiable due to the conflicting link state constraints that require \linkVariableAB to be unclassified (from \unclassificationRule) and inactive (from \inactiveLinkConstraintKTC) at the same time, and is discarded during \refinementStep{2}.

    \item
    \emph{Redundant \PAC{}:}
    \Cref{fig:rule-refinement-unclrule-inactconstraint-gui3} illustrates that  \gluingUILong{3} corresponds to a redundant \PAC{}:
    The affected link \linkVariableAB is never part of the triangle of links that is a match of the conclusion of the \inactiveLinkConstraintKTCLong.
    Due to the assumption that the topology fulfills the \inactiveLinkConstraintKTCLong before applying the \unclassificationRuleLong, the \PAC{} is always fulfilled and may be safely ignored.
    \begin{figure}
        \begin{center}
            \includegraphics[width=\textwidth]{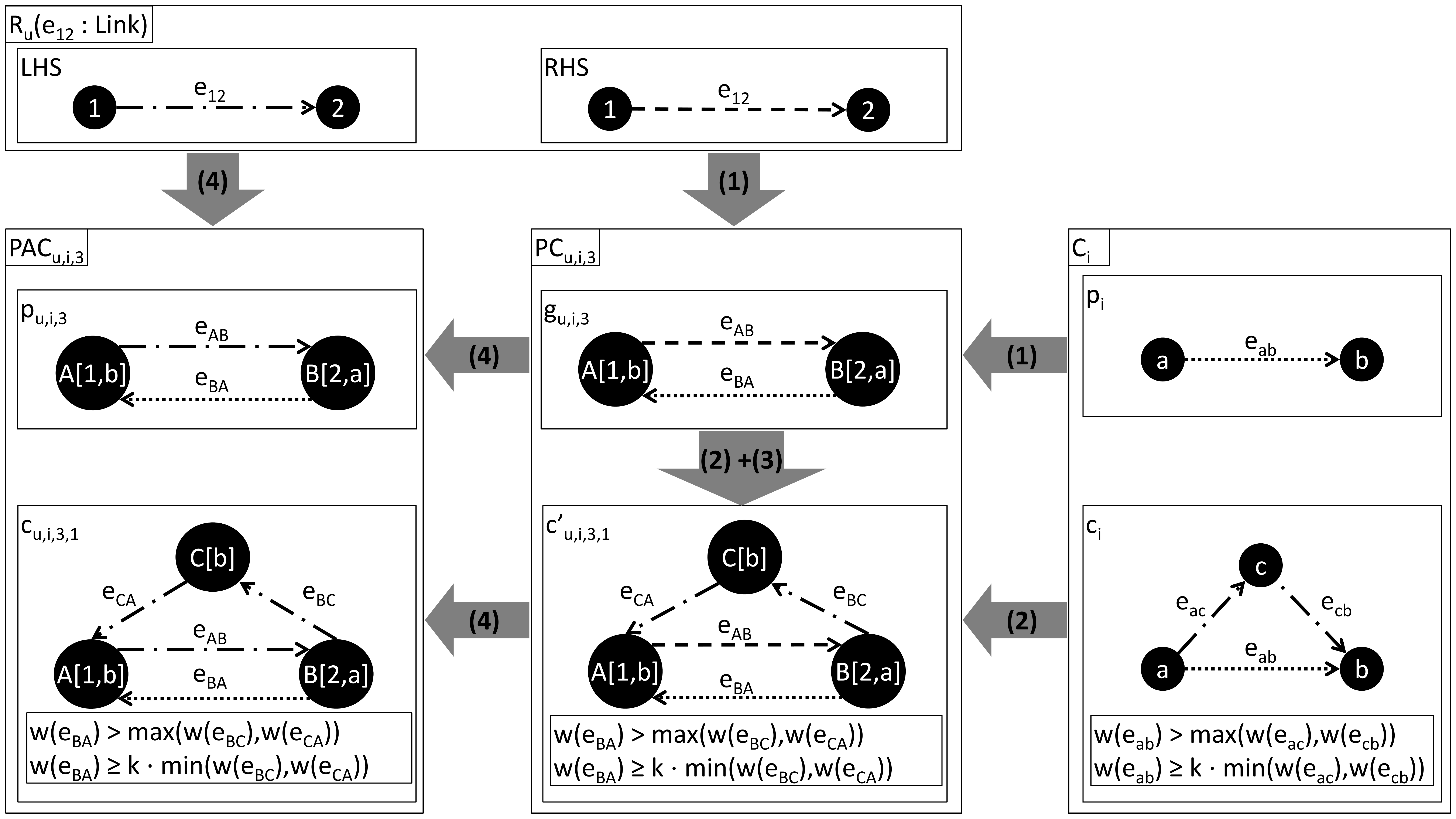}
        \end{center}
        \caption{Transformation of the gluing~\gluingUI{3} into the redundant application condition~ \PACui{3}}
        \label{fig:rule-refinement-unclrule-inactconstraint-gui3}
    \end{figure}
    
    \item 
    \emph{Restrictive \PACs:}
    The gluings \gluingUI{2}, \gluingUI{4}, \gluingUI{5}, and \gluingUI{6} correspond to four additional non-redundant \PACs{} of \unclassificationRule.
    In the following, we describe the application of the refinement procedure to \gluingUI{2} in detail.

\end{itemize}

\GluingUILong{2} is satisfiable;
therefore, it serves as the premise of the new post-condition~\postconditionUI{2}, as shown in \Cref{fig:rule-refinement-unclrule-inactconstraint-gui2}.
\begin{figure}
    \begin{center}
        \includegraphics[width=\textwidth]{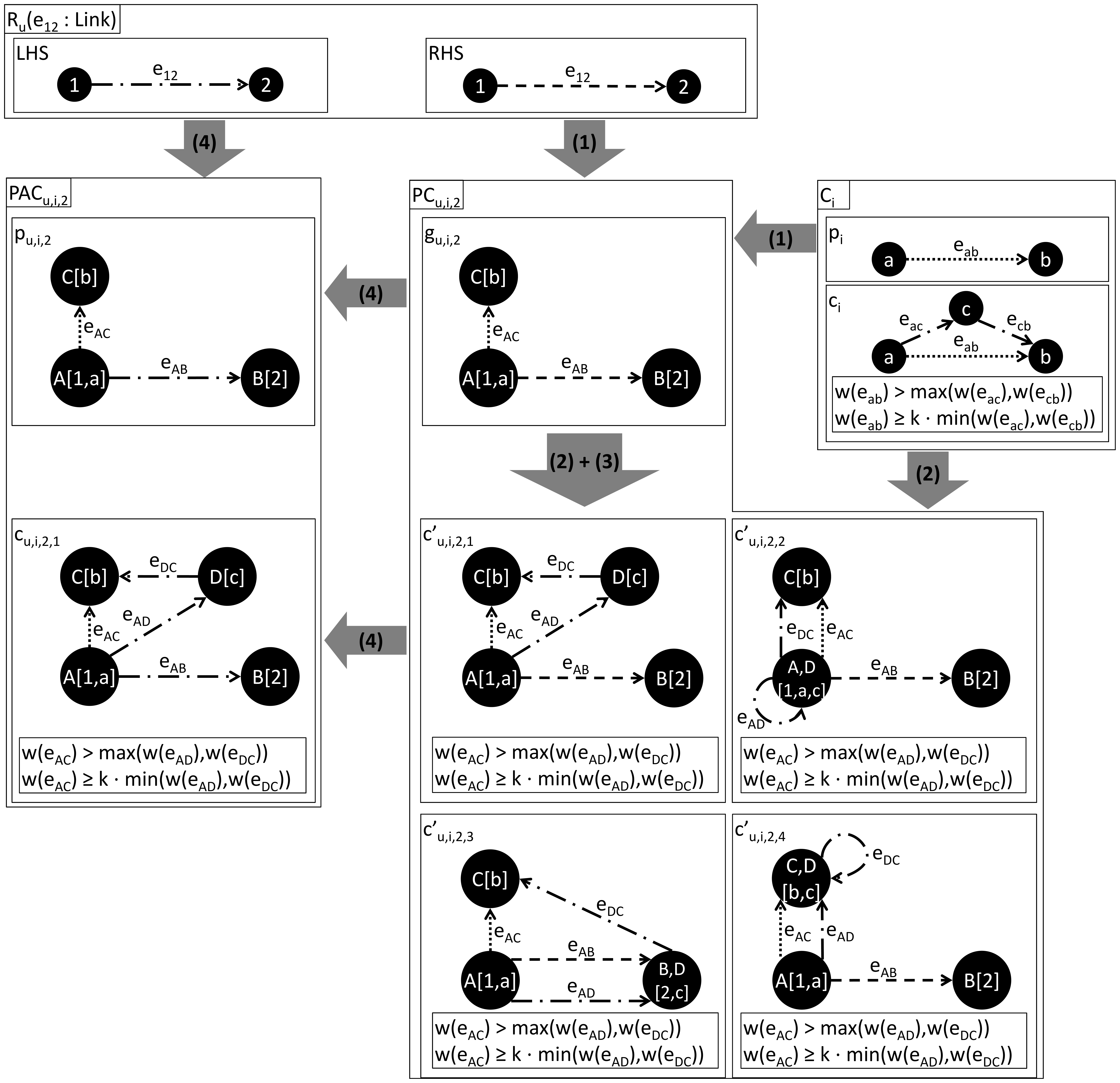}
    \end{center}
    \caption{Transformation of gluing~\gluingUI{2} into the positive application condition~\PACui{2}}
    \label{fig:rule-refinement-unclrule-inactconstraint-gui2}
\end{figure}
In \refinementStep{3}, we obtain the basic conclusion \conclusionUI{2,1}{'} by adding node variable $D$, which corresponds to node variable $c$ in \conclusionDI.
In \refinementStep{4}, three reduced conclusion patterns \conclusionUI{2,r\,;\,2\leq r \leq 4}{'} are obtained by merging node variable \nodeVariableD with the node variables~\nodeVariableA, \nodeVariableB, and \nodeVariableC, respectively.
All three reduced conclusions are not considered further on because they are unsatisfiable due to either loops (for \conclusionUI{2,2}{'}) or parallel link variables (for \conclusionUI{2,3}{'} and \conclusionUI{2,4}{'}).
Therefore, in \refinementStep{4}, only the basic conclusion \conclusionUI{2,1}{'} is transformed into a conclusion pattern of the positive application condition~\PACui{2} of the \unclassificationRuleLong.
Intuitively speaking, the new application condition requires that for each outgoing link of the source node of the affected link \linkVariableOneTwo, a triangle of links must exist that matches the conclusion of the \inactiveLinkConstraintKTCLong.
This additional triangle ensures that even if unclassifying link \linkVariableOneTwo may destroy a match of the conclusion of \inactiveLinkConstraintKTC, there is at least one more match of \conclusionDI, which ensures that \inactiveLinkConstraintKTC is preserved.
As a result of this refinement step, four new non-redundant application conditions \PACui{2}, \PACui{4}, \PACui{5}, and \PACui{6} were added to the \unclassificationRuleLong shown in \Cref{fig:rule-refinement-unclrule-inactconstraint-result}.
\begin{figure}
    \begin{center}
        \includegraphics[width=\textwidth]{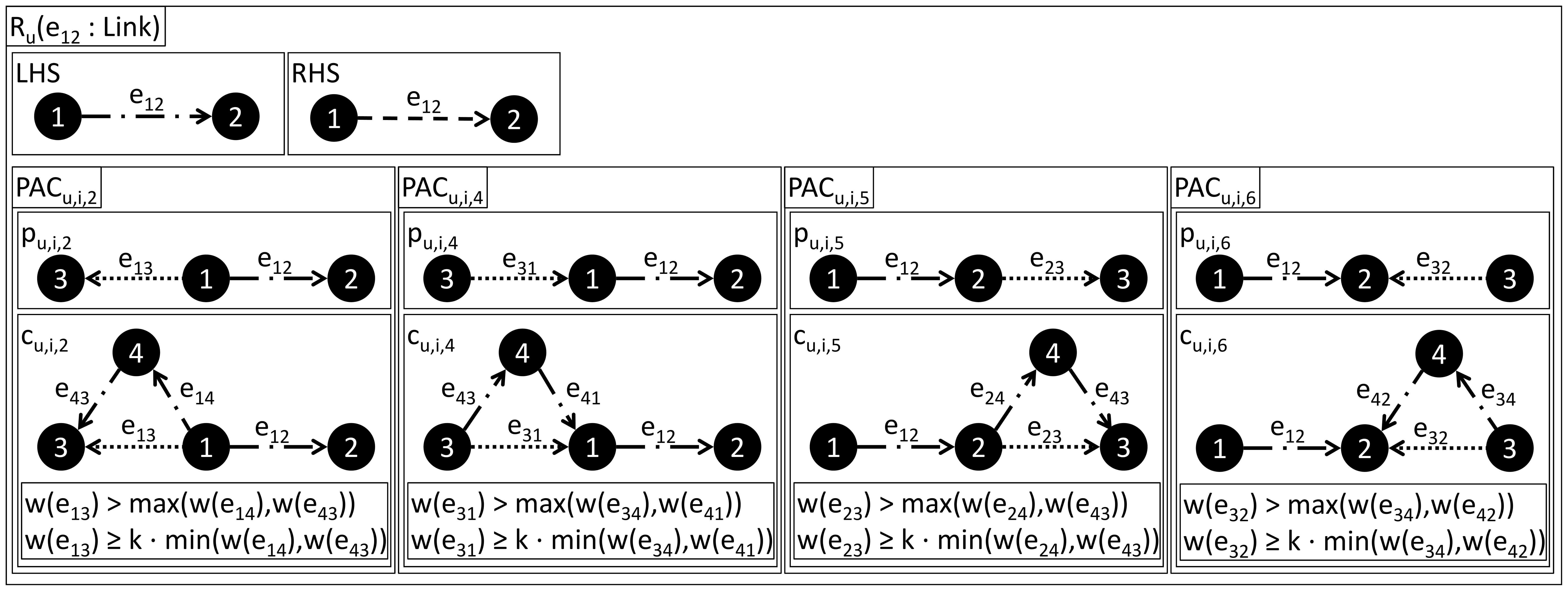}
    \end{center}
    \caption{Refined \unclassificationRuleLong with four additional application conditions \PACui{z\,;\,z \in \{2, 4, 5, 6\}}}
    \label{fig:rule-refinement-unclrule-inactconstraint-result}
\end{figure}

\refinementPair{\activationRule}{\inactiveLinkConstraintKTC}
The six possible gluings \gluingAI{z\,,\,1 \leq z \leq 6} of the \RHS of the \activationRuleLong and the premise of the \inactiveLinkConstraintKTCLong are almost identical to the \gluingsUILong{z\,,\,1 \leq z \leq 6}, shown in \Cref{fig:rule-refinement-unclrule-inactconstraint-gluings}.
The major difference is that the link state constraint of $\linkVariableAB$ originating from the \activationRuleLong is now $\state(\linkVariableAB) = \ACT$ in all gluings.
In this case, the gluings and their corresponding application conditions can be categorized as follows:
\begin{itemize}
    \item
    \emph{Unsatisfiable gluing:}
    Gluing  \gluingAI{1} is unsatisfiable due to the conflicting attribute constraints $\state(\linkVariableAB) = \ACT$ (from \activationRule) and $\state(\linkVariableAB) = \INACT$ (from \inactiveLinkConstraintKTC), and is discarded during \refinementStep{2}.

    \item
    \emph{Redundant \PACs:} 
    The remaining five gluings~\gluingAI{2} to \gluingAI{6} correspond to redundant \PACs as illustrated for \gluingAI{2} in \Cref{fig:rule-refinement-actrule-inactconstraint-gai2}.
    \begin{figure}
        \begin{center}
            \includegraphics[width=\textwidth]{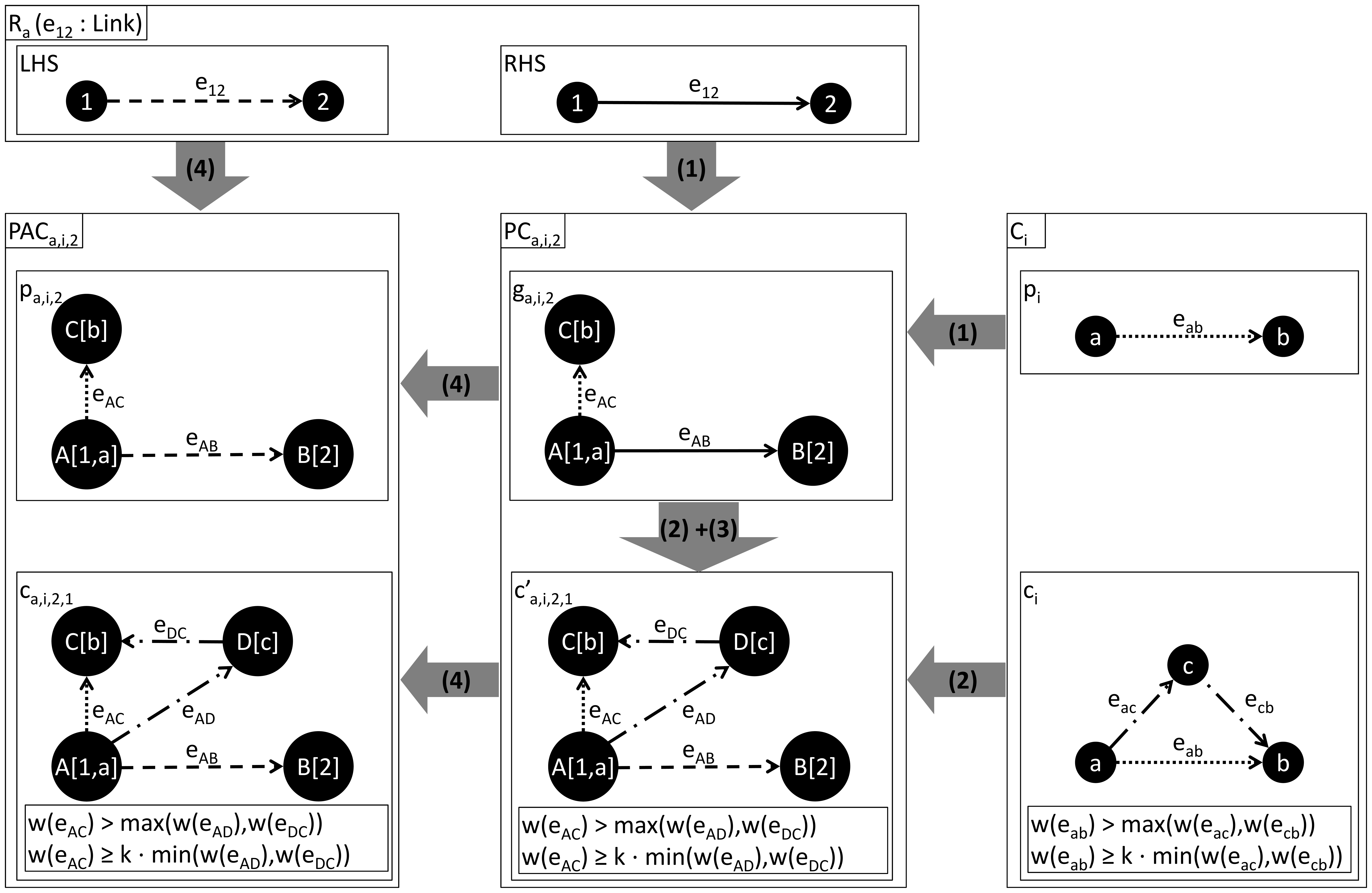}
        \end{center}
        \caption{Transformation of gluing~\gluingAI{2} into the redundant application condition~\PACai{2}}
        \label{fig:rule-refinement-actrule-inactconstraint-gai2}
    \end{figure}
    We show only the basic conclusion because, as discussed previously, the reduced conclusions are all unsatisfiable due to loops or parallel link variables.
    Note that, if any outgoing link of node variable \nodeVariableA was inactive prior to applying the \activationRuleLong, then it would be part of a triangle of links that matches the conclusion of \inactiveLinkConstraintKTC, and \linkVariableAB is not part of this triangle because it is unclassified.
    Therefore, the resulting application condition \PACai{2} is always fulfilled and can be safely ignored.
    
    \item
    \emph{Restrictive \PACs{}:}
    None of the gluings corresponds to a restrictive \PAC{}.
    Therefore, as a result of this refinement step, no application conditions need to be added to the \activationRuleLong.
\end{itemize}

\refinementPair{\inactivationRule}{\inactiveLinkConstraintKTC}
The six gluings \gluingII{z\,,\,1 \leq z \leq 6} of the \RHS of the \inactivationRuleLong and the premise of the \inactiveLinkConstraintKTCLong are similar to the \gluingsUILong{z\,,\,1 \leq z \leq 6}, shown in \Cref{fig:rule-refinement-unclrule-inactconstraint-gluings}.
The major difference is that the link state constraint of $\linkVariableAB$ is now $\state(\linkVariableAB) = \INACT$ in all gluings.
In this case, the gluings and their corresponding application conditions can be categorized as follows:
\begin{itemize}
    \item
    \emph{Unsatisfiable gluing:}
    None of the gluings~\gluingII{z\,,\,1 \leq z \leq 6} is unsatisfiable.
    
    \item 
    \emph{Restrictive \PAC{}:} \Cref{fig:rule-refinement-inactrule-inactconstraint-gii1} shows that the gluing~\gluingII{1} corresponds to an additional \PAC{} of the \inactivationRuleLong.%
    \begin{figure}
        \begin{center}
            \includegraphics[width=\textwidth]{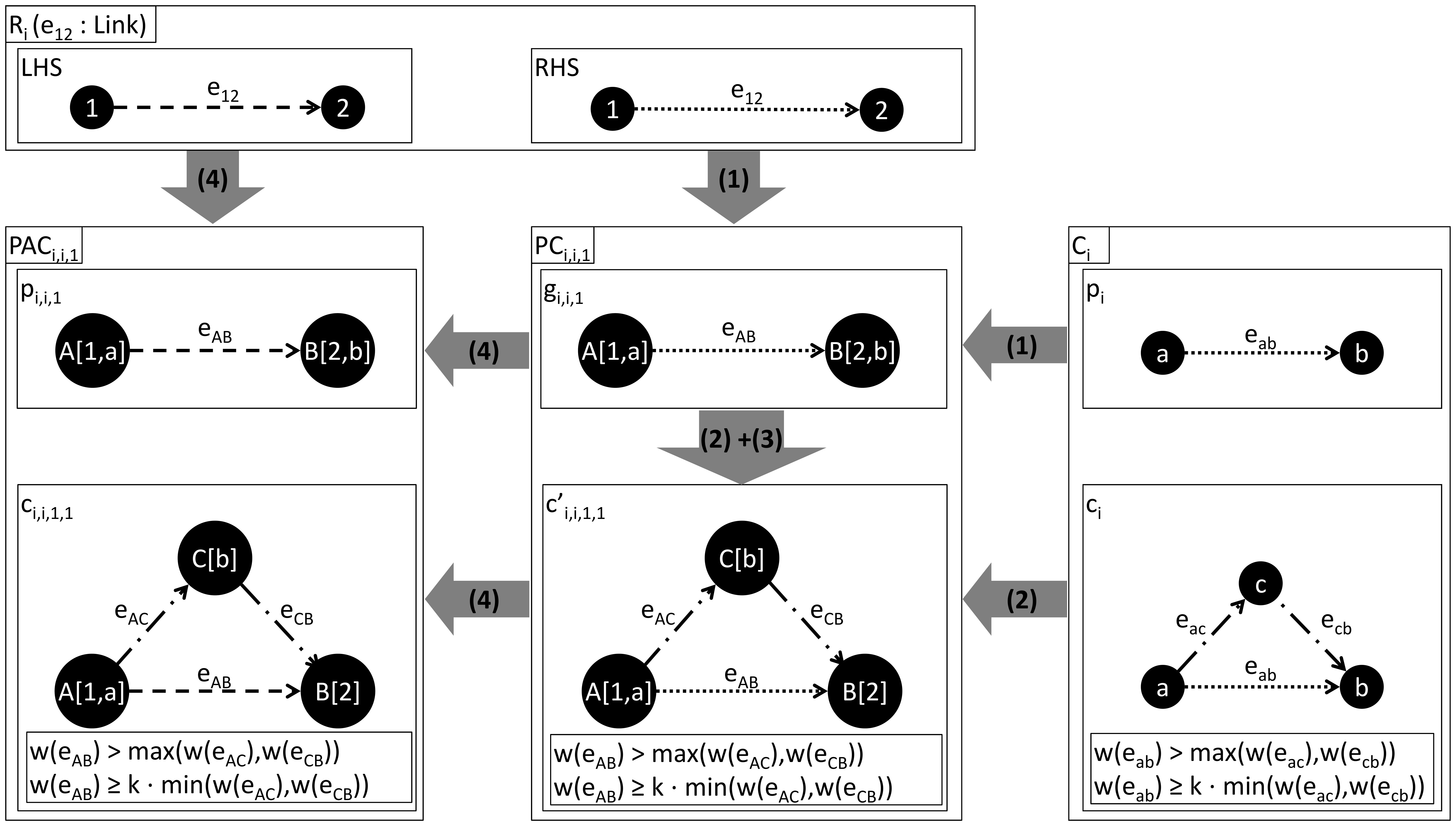}
        \end{center}
        \caption{Transformation of the gluing~\gluingII{1} into the restrictive application condition~\PACii{1}}
        \label{fig:rule-refinement-inactrule-inactconstraint-gii1}
    \end{figure}
    Gluing~\gluingII{1} serves as the premise of the post-condition \postconditionII{1}.
    We construct the conclusion to \postconditionII{1} by extending \gluingII{1} with a new node variable \nodeVariableC, which corresponds to the node variable \nodeVariablec in \inactiveLinkConstraintKTC.
    The new positive application condition~\PACii{1} is obtained by applying \inactivationRule in reverse order.

    \item
    \emph{Redundant \PACs:}
    As in the case of the \activationRuleLong, the remaining five gluings \gluingII{2} to \gluingII{6} result in redundant \PACs due to the assumption that the topology fulfills the \inactiveLinkConstraintKTCLong prior to any application of the \inactivationRuleLong.
\end{itemize}
As a result of this refinement step, the non-redundant positive application condition \PACii{1} was added to the \inactivationRuleLong.

\subsubsection{Refining Topology Control Rules to Preserve the Active-Link Constraint}
\label{sec:refinement-tc-act}

Next, we refine the \TC rules, \idest, the \activationRuleLong, the \inactivationRuleLong, and the \unclassificationRuleLong, to preserve the negative \activeLinkConstraintKTCLong.
The \activeLinkConstraintKTCLong is a negative constraint, which is transformed into additional negative application conditions of the affected \TC rules.
In this iteration, \refinementStep{2} and \refinementStep{3} of the refinement procedure are skipped because the conclusion of the \activeLinkConstraintKTCLong is empty.

\refinementPair{\unclassificationRule}{\activeLinkConstraintKTC}
\Cref{fig:rule-refinement-unclrule-actconstraint-gluings} gives an overview of the $12$ gluings \gluingUA{z\,,\,1 \leq z \leq 12} of the \RHS of the \unclassificationRuleLong and the premise of the \activeLinkConstraintKTCLong.
\begin{figure}
    \begin{center}
        \includegraphics[width=.8\textwidth]{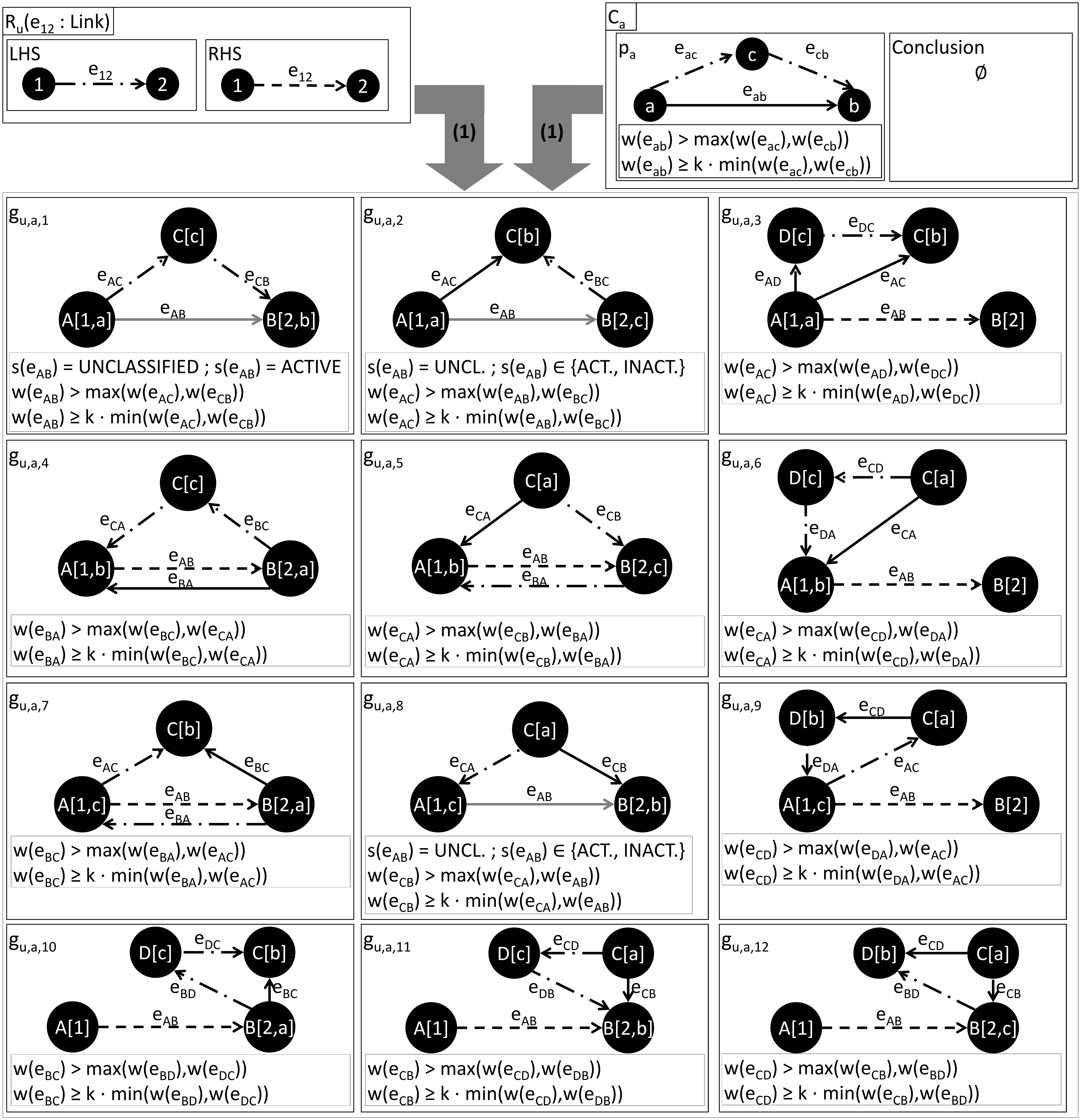}
    \end{center}
    \caption{All twelve gluings of the \unclassificationRuleLong and the \activeLinkConstraintKTCLong}
    \label{fig:rule-refinement-unclrule-actconstraint-gluings}
\end{figure}
In this case, the gluings and their corresponding application conditions can be categorized as follows:
\begin{itemize}
    \item \emph{Unsatisfiable gluings:} The gluings \gluingUA{1}, \gluingUA{2}, and \gluingUA{8} are unsatisfiable because of the conflicting combined link state constraints $\state(\linkVariableAB) = \UNCL \wedge \state(\linkVariableAB) = \ACT$ (for \gluingUA{1} and \gluingUA{2}) and $\state(\linkVariableAB) = \UNCL \wedge \state(\linkVariableAB) \in \{\ACT,\INACT\}$ (for \gluingUA{1} and \gluingUA{8}), and are discarded during \refinementStep{2}.

    \item \emph{Redundant \NACs:} As above, the remaining nine gluings correspond to redundant \NACs due to the assumption that the topology fulfills the \activeLinkConstraintKTCLong prior to any application of \activationRule.

    \item
    \emph{Restrictive \PACs{}:}
    None of the gluings corresponds to a restrictive \PAC{}.
    Therefore, as a result of this refinement step, no application conditions need to be added to the \unclassificationRuleLong.
\end{itemize}

\refinementPair{\activationRule}{\activeLinkConstraintKTC}
\Cref{fig:rule-refinement-actrule-actconstraint-gluings} gives an overview of the $12$ gluings \gluingAA{z\,,\,1 \leq z \leq 12} of the \RHS of the \activationRuleLong and the premise of the \activeLinkConstraintKTCLong.
\begin{figure}
    \begin{center}
        \includegraphics[width=.8\textwidth]{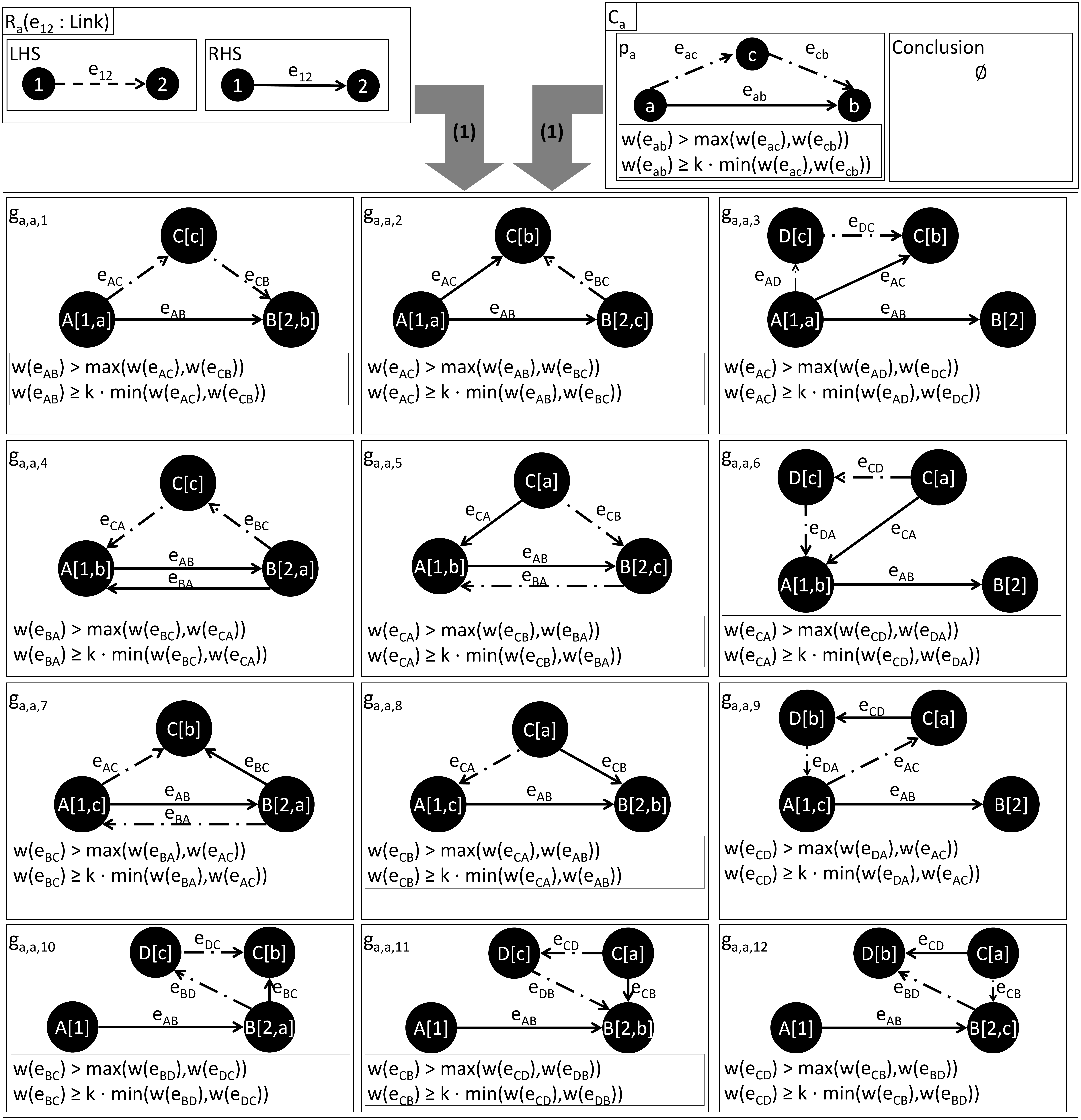}
    \end{center}
    \caption{All twelve gluings of the \activationRuleLong and the \activeLinkConstraintKTCLong}
    \label{fig:rule-refinement-actrule-actconstraint-gluings}
\end{figure}
These gluings are similar to the gluings shown in \Cref{fig:rule-refinement-unclrule-actconstraint-gluings}, but the link state constraints are non-conflicting here.
In this case, the gluings and their corresponding application conditions can be categorized as follows:
\begin{itemize}
    \item
    \emph{Unsatisfiable gluings:}
    None of the gluings~\gluingAA{z\,,\,1 \leq z \leq 12} is unsatisfiable.
    
    \item
    \emph{Restrictive \NACs:}
    Gluings~\gluingAA{1}, \gluingAA{2}, and \gluingAA{8} are transformed into three \NACs in \refinementStep{4}, as shown in \Cref{fig:rule-refinement-actrule-actconstraint-nacs}.
    \begin{figure}
        \begin{center}
            \includegraphics[width=\textwidth]{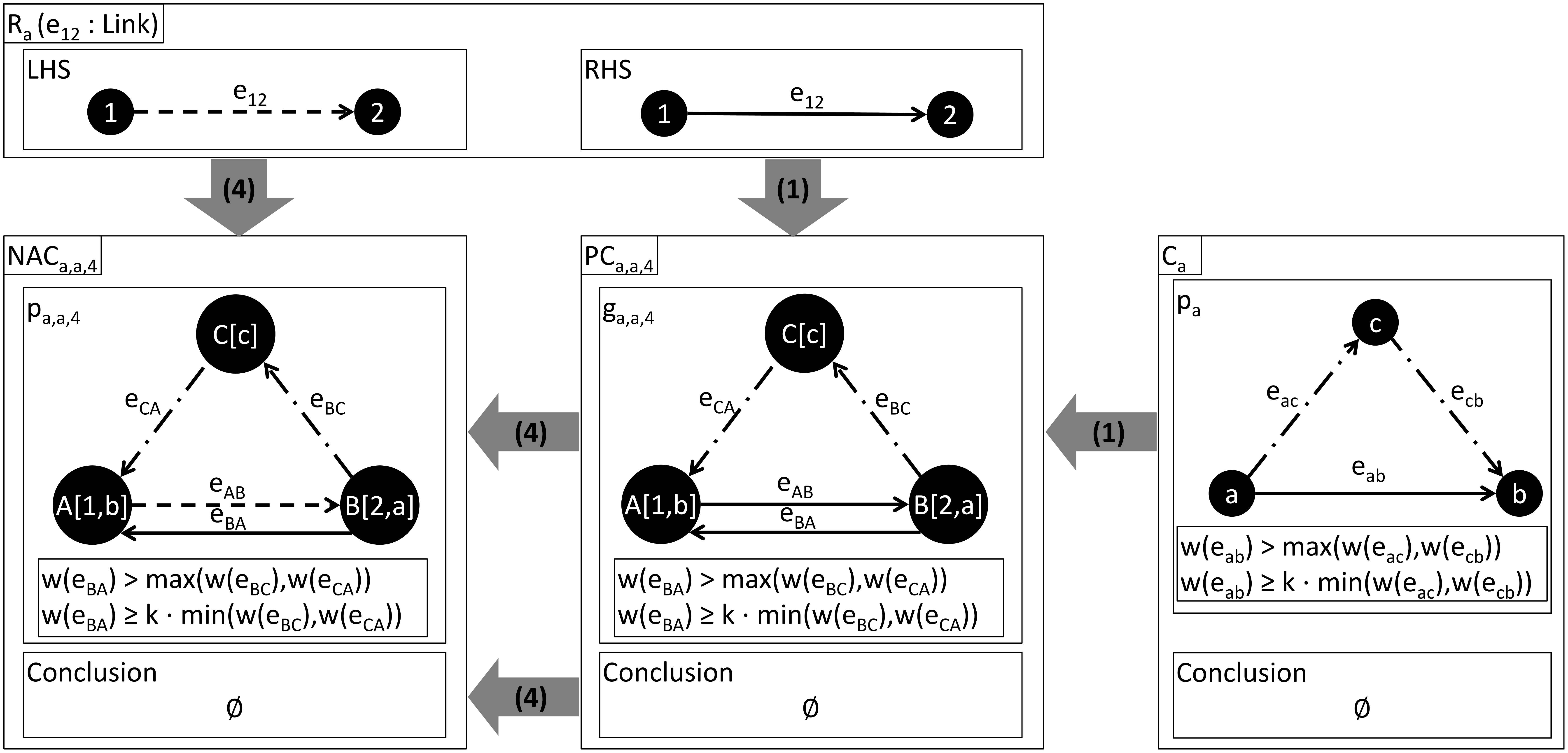}
        \end{center}
        \caption{Transformation of gluing~\gluingAA{4} into the redundant application condition~\NACaa{4}}
        \label{fig:rule-refinement-actrule-actconstraint-gaa4}
    \end{figure}
    For instance, for gluing~\gluingAA{1}, applying the \activationRuleLong in reverse order implies that the link state constraint $\state(\linkVariableAB) = \ACT$ is replaced with $\state(\linkVariableAB) = \UNCL$.

    \item \emph{Redundant \NACs:}
    The other nine gluings result in redundant \NACs;
    gluing~\gluingAA{4} serves as a representative example whose transformation is shown in \Cref{fig:rule-refinement-actrule-actconstraint-gaa4}:
    \begin{figure}
        \begin{center}
            \includegraphics[width=\textwidth]{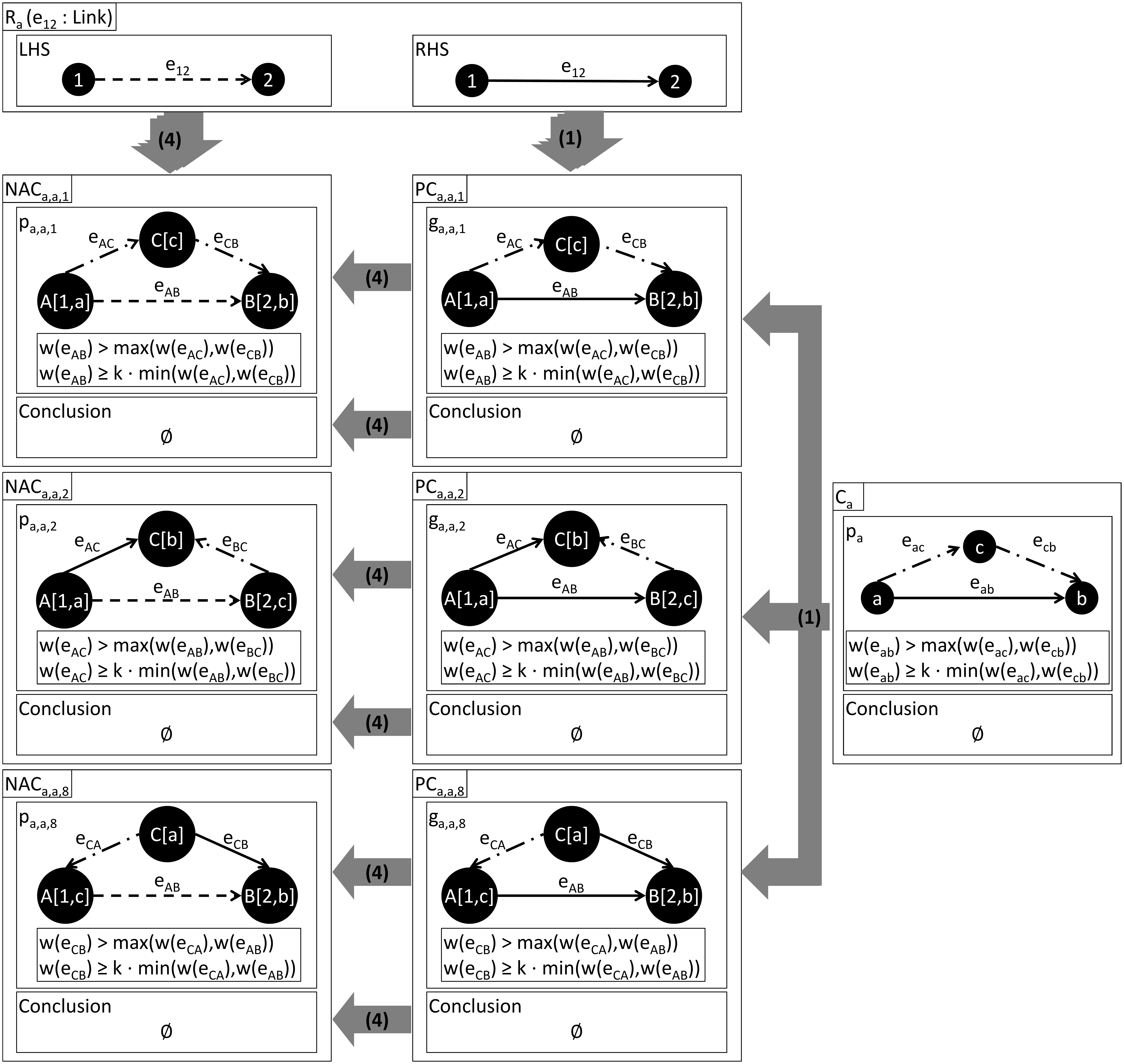}
            \caption{Transformation of the gluings \gluingAA{1} \gluingAA{2}, and \gluingAA{8} into three restrictive \NACs of the \activationRuleLong}
            \label{fig:rule-refinement-actrule-actconstraint-nacs}
        \end{center}
    \end{figure}
    The gluing~\gluingAA{4} serves as the premise for the post-condition~\postconditionAA{4}, and the set of conclusions of the post-condition is empty.
    The application condition is obtained by setting the link state constraint of \linkVariableAB to $\state(\linkVariableAB) = \UNCL$.
    The additional \NAC{} represents the requirement that the \activeLinkConstraintKTCLong must hold for all anti-parallel links of \linkName{12}, whose source node is the target node of \linkName{12} and vice versa.
    The constructed application condition is redundant due to the assumption that the topology is weakly consistent prior to applying the \activationRuleLong.
\end{itemize}
As a result of this refinement step, the three non-redundant new \NACs \NACaa{1}, \NACaa{2}, and \NACaa{8} were added to the \activationRuleLong.

\refinementPair{\inactivationRule}{\activeLinkConstraintKTC}
The refinement of the \inactivationRuleLong is similar to the previous case.
We obtain $12$ possible gluings \gluingIA{z\,,\,1 \leq z \leq 12} similar to those shown in \Cref{fig:rule-refinement-actrule-actconstraint-gluings}.
The major difference is that the link state constraint of $\linkVariableAB$ originating from \unclassificationRule is now $\state(\linkVariableAB) = \INACT$ in all gluings.
In this case, the gluings and their corresponding application conditions can be categorized into three groups:
\begin{itemize}
    \item 
    \emph{Unsatisfiable gluing:} Gluing  \gluingIA{1} is unsatisfiable due to the conflicting attribute constraints $\state(\linkVariableAB) = \INACT$ (from \inactivationRule) and $\state(\linkVariableAB) = \ACT$ (from \activeLinkConstraintKTC) and is discarded during step (2).

    \item 
    \emph{Restrictive \NACs:} As before, applying the refinement procedure to the gluings \gluingIA{2} and \gluingIA{8} results in two new non-redundant \NACs as shown in \Cref{fig:rule-refinement-inactrule-actconstraint-nacs}.
    
    \item 
    \emph{Redundant \NACs:}
    As above, the remaining nine gluings correspond to redundant \NACs due to the assumption that the topology fulfills the \activeLinkConstraintKTCLong prior to any application of the \inactivationRuleLong.
\end{itemize}
As a result of this refinement step, the two new non-redundant \NACs \NACia{2} and \NACia{8} were added to the \inactivationRuleLong.
\begin{figure}
    \begin{center}
        \includegraphics[width=\textwidth]{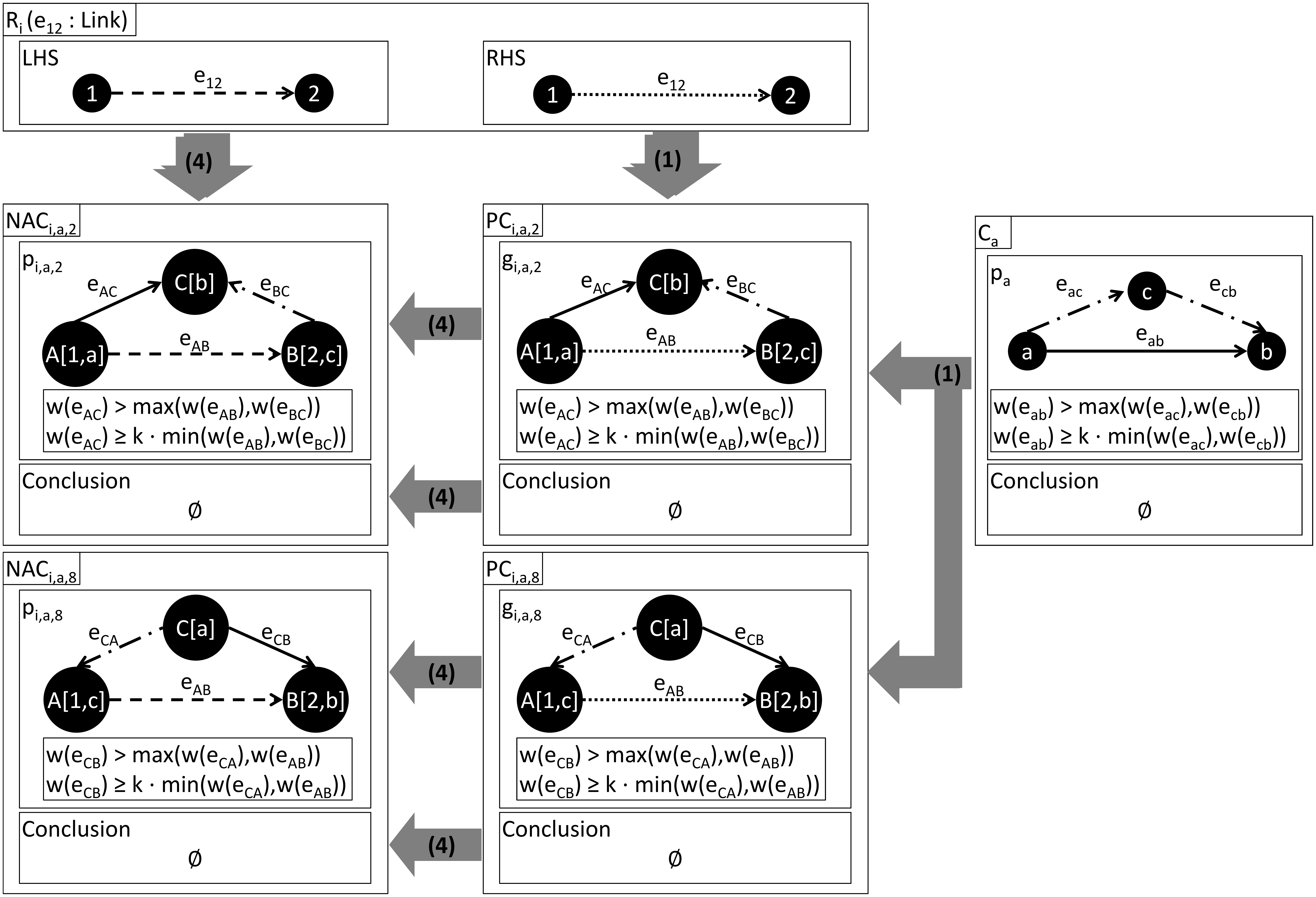}
    \end{center}
    \caption{Transformation of the gluings \gluingIA{2} and \gluingIA{8} into the restrictive application conditions \NACia{2} and \NACia{8}}
    \label{fig:rule-refinement-inactrule-actconstraint-nacs}
\end{figure}

\subsubsection{Refining Context Event Rules to Preserve the Inactive-Link Constraint}
\label{sec:refinement-ce-inact}
After describing the refinement procedure for the \TC rules in detail, the corresponding steps for the \CE rules in this section are presented in an abbreviated form.
Recall that, in general, applying a rule to a topology preserves a positive constraint 
\begin{inparaenum}
    \item if all new matches of the premise of the constraint that result from applying the rule may be extended to matches of the conclusion of this constraint, and
    \item if for each existing match of its premise that also exists after the rule application, at least one corresponding match of the conclusion exists.
\end{inparaenum}

\refinementPair{\nodeAdditionRule}{\inactiveLinkConstraintKTC}
The \nodeAdditionRuleLong preserves the \inactiveLinkConstraintKTCLong because a rule application of \nodeAdditionRule adds an isolated node.
Neither a new match of the premise is created nor a match of the conclusion is destroyed by adding a node to the topology.
Therefore, the \inactiveLinkConstraintKTCLong is already preserved by the \nodeAdditionRuleLong.

\refinementPair{\nodeRemovalRule}{\inactiveLinkConstraintKTC}
The \nodeRemovalRuleLong preserves the \inactiveLinkConstraintKTCLong because a rule application of \nodeRemovalRule removes only isolated nodes, which may not be part of a match neither of the premise nor of the conclusion of \inactiveLinkConstraintKTC.

\refinementPair{\linkAdditionRule}{\inactiveLinkConstraintKTC}
The \linkAdditionRuleLong preserves the \inactiveLinkConstraintKTCLong because a rule application of \linkAdditionRule adds an unclassified link to a topology.
Adding an unclassified link neither produces a match of the premise of \inactiveLinkConstraintKTC that cannot be extended to a match of the conclusion nor destroys any match of the conclusion of \inactiveLinkConstraintKTC.

\refinementPair{\linkRemovalRule}{\inactiveLinkConstraintKTC}
The \linkRemovalRuleLong does not preserve the \inactiveLinkConstraintKTCLong because a match of the conclusion may be destroyed without destroying the corresponding match of the premise.
The refinement procedure of the \linkRemovalRuleLong and the \inactiveLinkConstraintKTCLong proceeds completely analogous to the refinement procedure of the \unclassificationRuleLong and the \inactiveLinkConstraintKTCLong.
For brevity, we only present the resulting refined \linkRemovalRuleLong in 
\Cref{fig:rule-refinement-linkremovalrule-inactconstraint-result}.
Note that the premises of the application conditions only match if the  \linkNameLong{12} is either active or inactive and if it has an incident inactive link.
\begin{figure}
    \begin{center}
        \includegraphics[width=.7\textwidth]{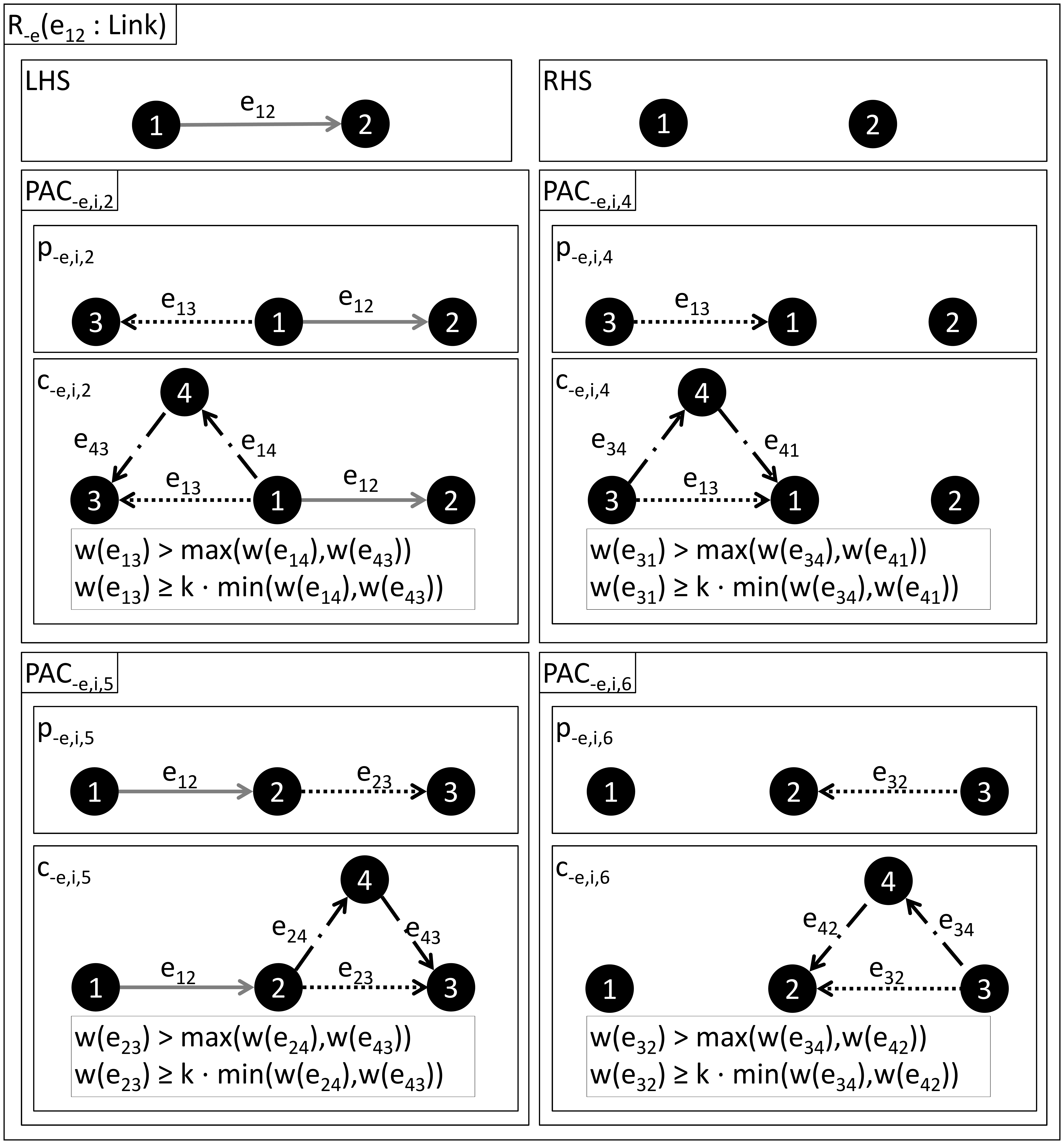}
    \end{center}
    \caption{Final refined \linkRemovalRuleLong}
    \label{fig:rule-refinement-linkremovalrule-inactconstraint-result}
\end{figure}

\refinementPair{\weightModificationRule}{\inactiveLinkConstraintKTC}
The \weightModificationRuleLong may violate the \inactiveLinkConstraintKTCLong in the very same way as the \linkRemovalRuleLong does.
Therefore, we omit the details of the derivation of the application conditions and only present the refined \unclassificationRuleLong in
\Cref{fig:rule-refinement-distmodrule-inactconstraint-result}.
One particularity here is the attribute constraint that requires $\weight(\linkVariableOneTwo)$ to equal $\varWnew$ after applying \weightModificationRule.
This attribute constraint is still present in the gluings that yield the post-conditions (not shown here), but it is missing from the resulting application conditions \PACmoddi{2}, \PACmoddi{4}, \PACmoddi{5}, and \PACmoddi{6} because applying \weightModificationRule in reverse order removes this attribute constraint. 
\begin{figure}
    \begin{center}
        \includegraphics[width=.7\textwidth]{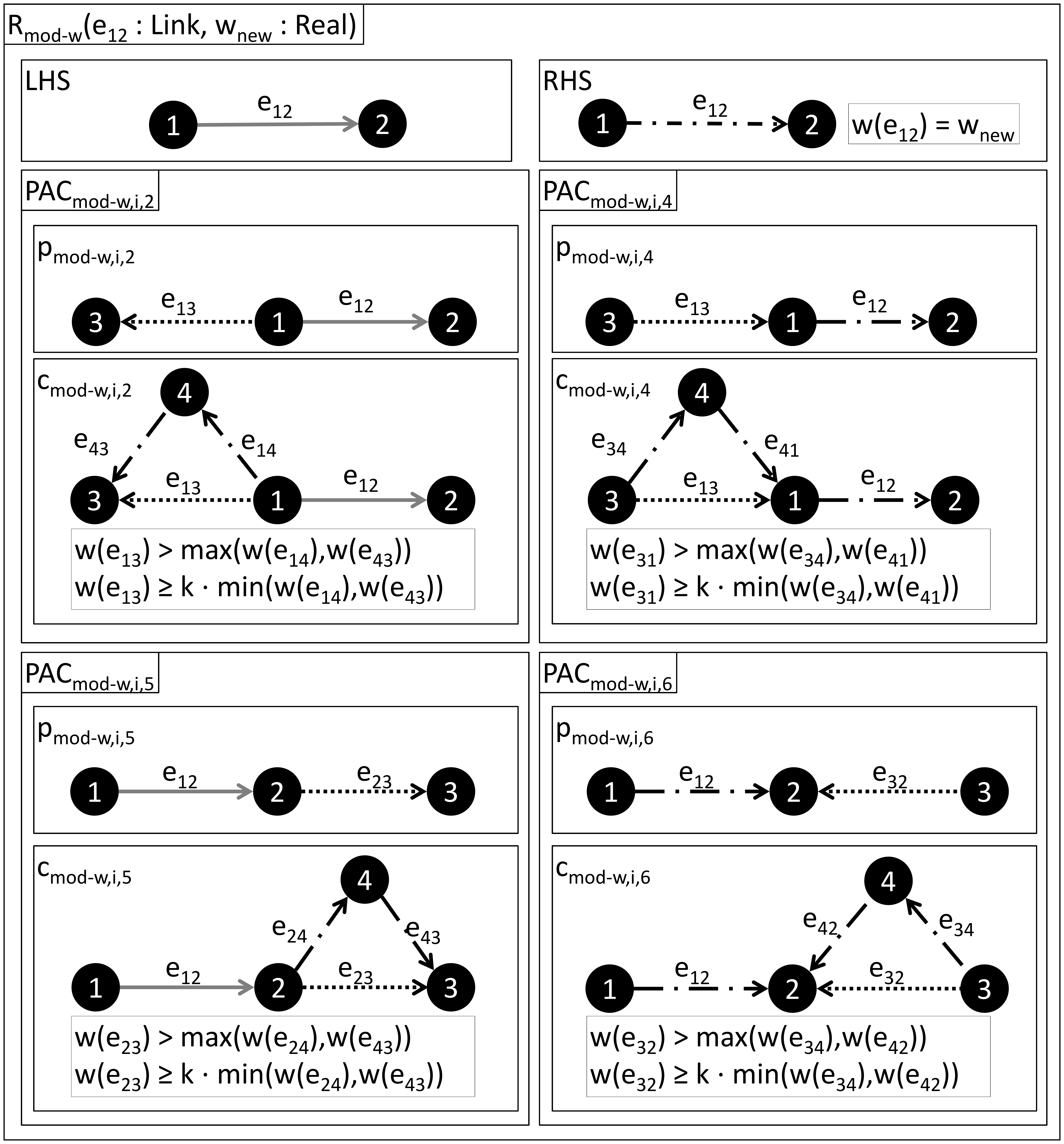}
    \end{center}
    \caption{Final refined \weightModificationRuleLong}
    \label{fig:rule-refinement-distmodrule-inactconstraint-result}
\end{figure}

\subsubsection{Refining Context Event Rules to Preserve the Active-Link Constraint}
\label{sec:refinement-ce-act}

As in the previous iteration, we abbreviate the description of the application of the refinement procedure in the following.
Recall that applying a rule to a topology preserves a negative constraint if the rule application may not produce any new matches of the premise of the constraint.

\refinementPair{\nodeAdditionRule}{\activeLinkConstraintKTC}
The \nodeAdditionRuleLong preserves the \activeLinkConstraintKTCLong because an application of \nodeAdditionRule adds an isolated node and, therefore, never produces a new match of the premise of \activeLinkConstraintKTC.

\refinementPair{\nodeRemovalRule}{\activeLinkConstraintKTC}
The \nodeRemovalRuleLong preserves the \activeLinkConstraintKTCLong because a rule application of \nodeRemovalRule removes an isolated node and, therefore, never produces a new match of the premise of \activeLinkConstraintKTC.

\refinementPair{\linkAdditionRule}{\activeLinkConstraintKTC}
The \linkAdditionRuleLong preserves the \activeLinkConstraintKTCLong because a rule application of \linkAdditionRule adds an \emph{unclassified} link to a topology, which may never produce a new match of the premise of \activeLinkConstraintKTC.

\refinementPair{\linkRemovalRule}{\activeLinkConstraintKTC}
The \linkRemovalRuleLong preserves the \activeLinkConstraintKTCLong because a rule application of \linkRemovalRule removes a link and, therefore, never produces a new match of the premise of \activeLinkConstraintKTC.

\refinementPair{\weightModificationRule}{\activeLinkConstraintKTC}
The \weightModificationRuleLong preserves the \activeLinkConstraintKTCLong because a rule application of \weightModificationRule unclassifies a link, which cannot result in a new match of the premise of \activeLinkConstraintKTC.

\subsubsection*{Intermediate Summary}
\label{sec:refinement-intermediate-summary}
After applying the refinement procedure to all combinations of \GT rules and \ktc-specific constraints, all refined rules preserve weak consistency.
\Cref{fig:ktc-algorithm-after-step2} shows the specification of the \ktc algorithm after \Cref{sec:refinement-methodology}.
The output topology fulfills the \unclassifiedLinkConstraintLong because the stop node is only reached if the topology contains no more unclassified links.
\begin{figure}
    \begin{center}
        \includegraphics[width=.9\textwidth]{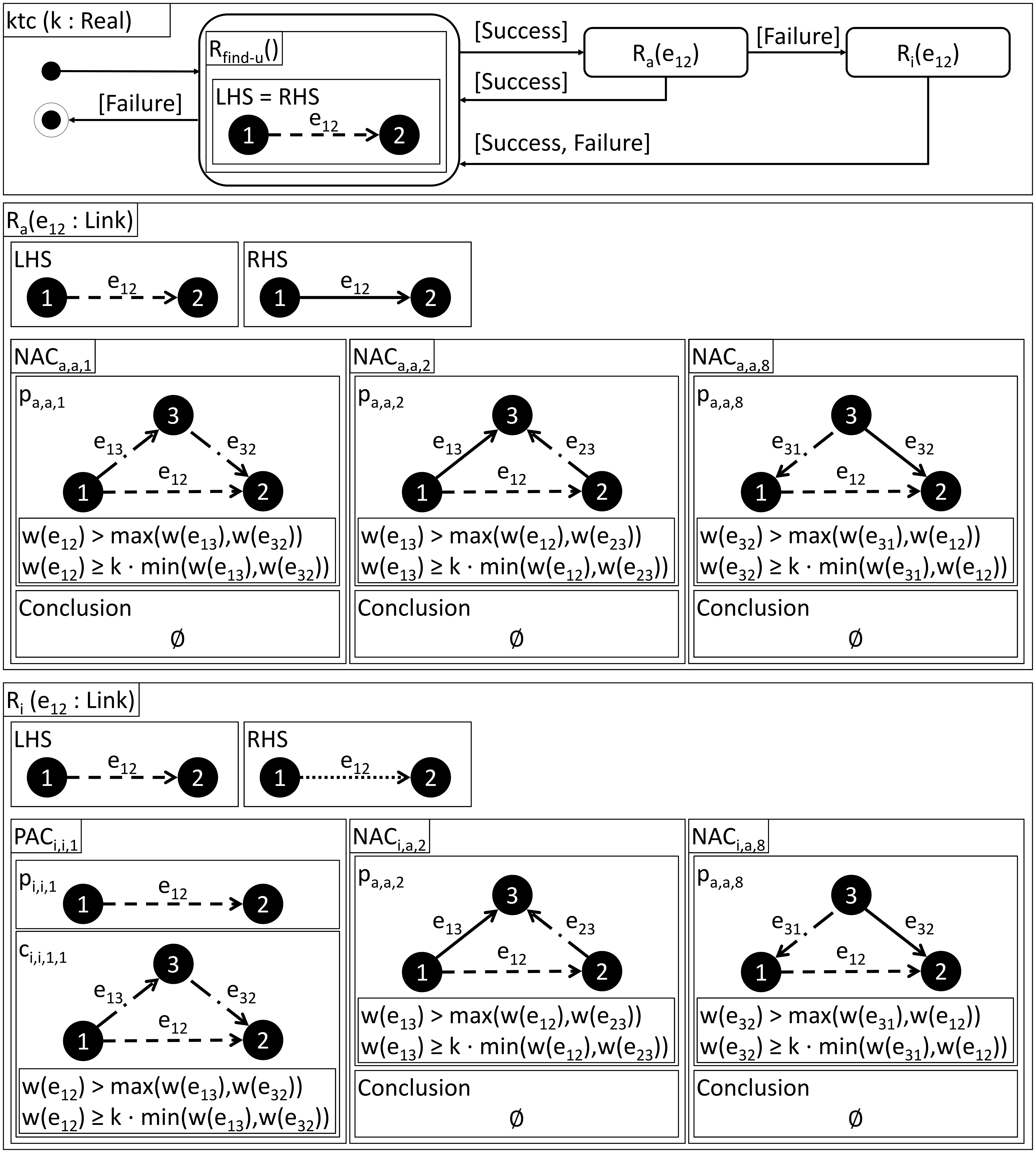}
    \end{center}
    \caption{\GToperationNN \tcOperation, the specification of \ktc after rule refinement in \Cref{sec:refinement-methodology}}
    \label{fig:ktc-algorithm-after-step2}
\end{figure}

\subsection{Restoring Applicability of Unrestrictable Rules via Handler Operations}
\label{sec:enforcing-applicability}

\Cref{sec:refinement-methodology} resulted in additional application constraints for three unrestrictable \GT rules:
the \unclassificationRuleLong, the \linkRemovalRuleLong, and the \weightModificationRuleLong.
The \weightModificationRuleLong, for instance, is currently only applicable if the resulting topology is still weakly consistent.
As explained in \Cref{sec:restrictability}, unrestrictable \GTrulesNN represent (voluntarily or involuntarily) non-controllable modifications of the topology.
Therefore, the goal of this section is to \emph{restore} the applicability of all unrestrictable \GTrulesNN (G2), but without sacrificing the guarantee that applying any of the unrestrictable rules preserves weak consistency (G1).

The following proposed solution is one of the main contributions of this paper.
We suggest to systematically turn any added application condition \AC{x,y} of an unrestrictable \GTruleLong{x} into appropriate \emph{restoration operations} \restorationOperation{\AC{x,y}}.
The purpose of this restoration operation is to resolve any violations of weak consistency caused by not checking \AC{x,y}.

\Cref{fig:restoration-operations-overview} shows the general idea of deriving restoration operations from application conditions of a \GTruleLong{x}:
Each application of \GTrule{x} is replaced with a corresponding handler operation \handlerOperation{\GTrule{x}}.
\begin{figure}
\begin{center}
\includegraphics[width=.9\textwidth]{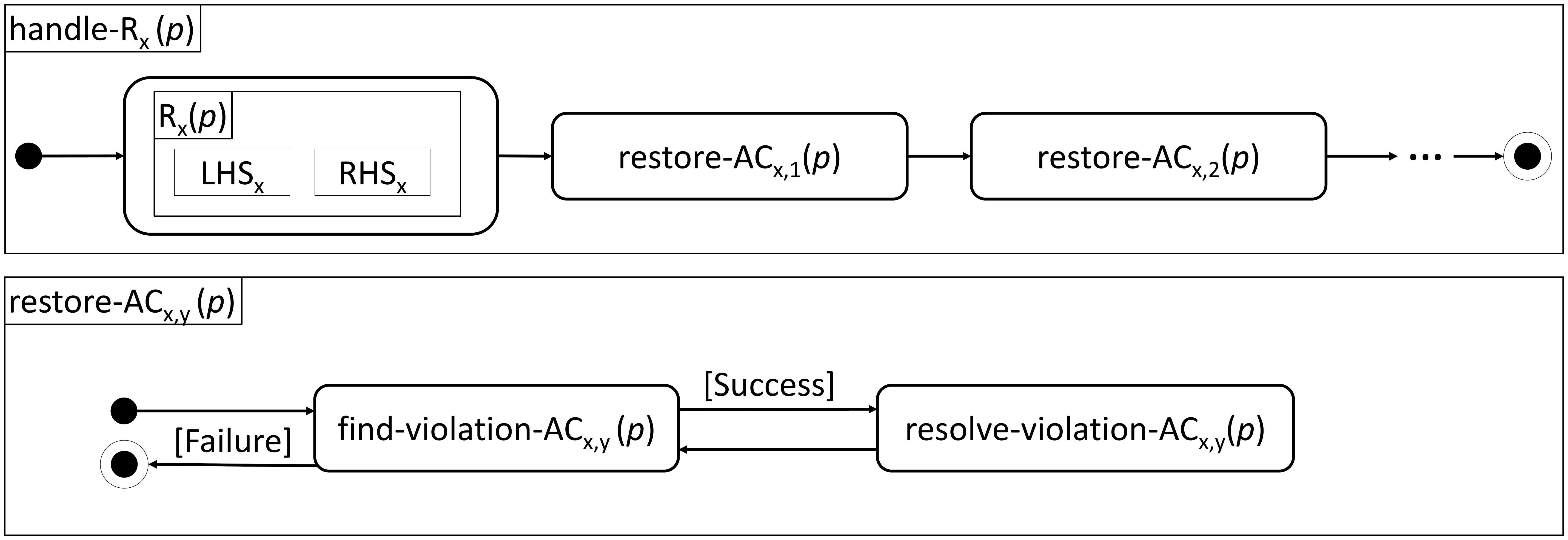}
\caption{Schematic overview of derivation of restoration operations}
\label{fig:restoration-operations-overview}
\end{center}
\end{figure}
For each \ACLong{x,y} to be dropped, a restoration operation \restorationOperation{\AC{x,y}} is invoked after applying \GTrule{x}.
The restoration operation \restorationOperation{\AC{x,y}} iteratively identifies all violations of \AC{x,y} by means of a \emph{violation identification operation} and resolves the identified violation by means of a \emph{violation resolution operation}.
A restoration operation terminates only if all constraint violations have been eliminated.
In the schematic overview shown in \Cref{fig:restoration-operations-overview}, we assume that \GTrule{x} takes a (potentially empty) list $p$ of parameters (\eg, the link variable \linkVariableOneTwo of \unclassificationRule).
The corresponding handler operation \handlerOperation{\GTrule{x}} has the same list of parameters \emph{p} as \GTrule{x}.
The parameter list \emph{p} is forwarded to \GTrule{x} and the restoration operations \restorationOperation{\AC{x,y}}.
The following properties have to be proved for each restoration operation.
\begin{itemize}
\item The resolution of a particular constraint violation should not result in new constraint violations at the affected match of \GTrule{x}.
Otherwise, the restoration operation may not terminate.
Still, it may be necessary to create constraint violations either of other application conditions or at a different match of the current application condition.
Such violations have to be handled subsequently (\eg, by recursively invoking restoration operations).

\item
A general assumption in our scenario is that constraint violations can be resolved by unclassifying links.
This is a valid assumption because the application conditions were derived from the constraints of weak consistency. Weak consistency is guaranteed to hold on a completely unclassified topology (see \Cref{thm:unclassified-topology-fulfills-weak-consistency}).
Therefore, a na\"{i}ve violation resolution operation could simply unclassify all links in the topology.
\end{itemize}
\begin{table}
\begin{center}
\caption{Overview of operations introduced during the restoration step (\Cref{sec:enforcing-applicability})}
\label{tab:restoration-operations-overview}
\begin{tabular}{p{0.3\textwidth}p{0.6\textwidth}}
\toprule

\textbf{Conventional Name} & \textbf{Description} \\

\midrule

\handlerOperation{\GTrule{x}} & Handler operation for \GTrule{x}, which ensures that weak consistency is preserved when applying the unrestricted version of \GTrule{x}\\

\restorationOperation{\AC{x,y}} & Restoration operation for \AC{x,y}, which restores application condition \AC{x,y} by eliminating all violations of \AC{x,y} that result from applying \GTrule{x}\\

\violationIdentificationOperation{\AC{x,y}} & Violation identification operation for \AC{x,y}, which identifies some violation of \AC{x,y} \\

\violationResolutionOperation{\AC{x,y}} & Violation resolution operation for \AC{x,y}, which resolves a given violation of \AC{x,y} \\

\bottomrule
\end{tabular}
\end{center}
\end{table}
In the following subsections, we illustrate the concrete derivation of violation identification and resolution operations, based on affected rules \unclassificationRule, \linkRemovalRule, and \weightModificationRule. 

\subsubsection{Deriving Repair Operations for the Unclassification Rule}

Applying the refinement algorithm to the \unclassificationRuleLong resulted in four additional application conditions---\PACui{2}, \PACui{4}, \PACui{5}, and \PACui{6}---as shown in \Cref{fig:rule-refinement-unclrule-inactconstraint-result}.
These application conditions essentially require that the affected link \linkVariableOneTwo may only be inactivated if each incident inactive link of the source and target node of the affected link can be extended to a triangle that matches the conclusion of the \inactiveLinkConstraintKTCLong.

\Cref{fig:repair-operations-for-unclassification-rule} shows the \handlerOperationLong{\unclassificationRule} and one  \restorationOperation{\PACui{2}}.
\begin{figure}
    \begin{center}
        \includegraphics[width=\textwidth]{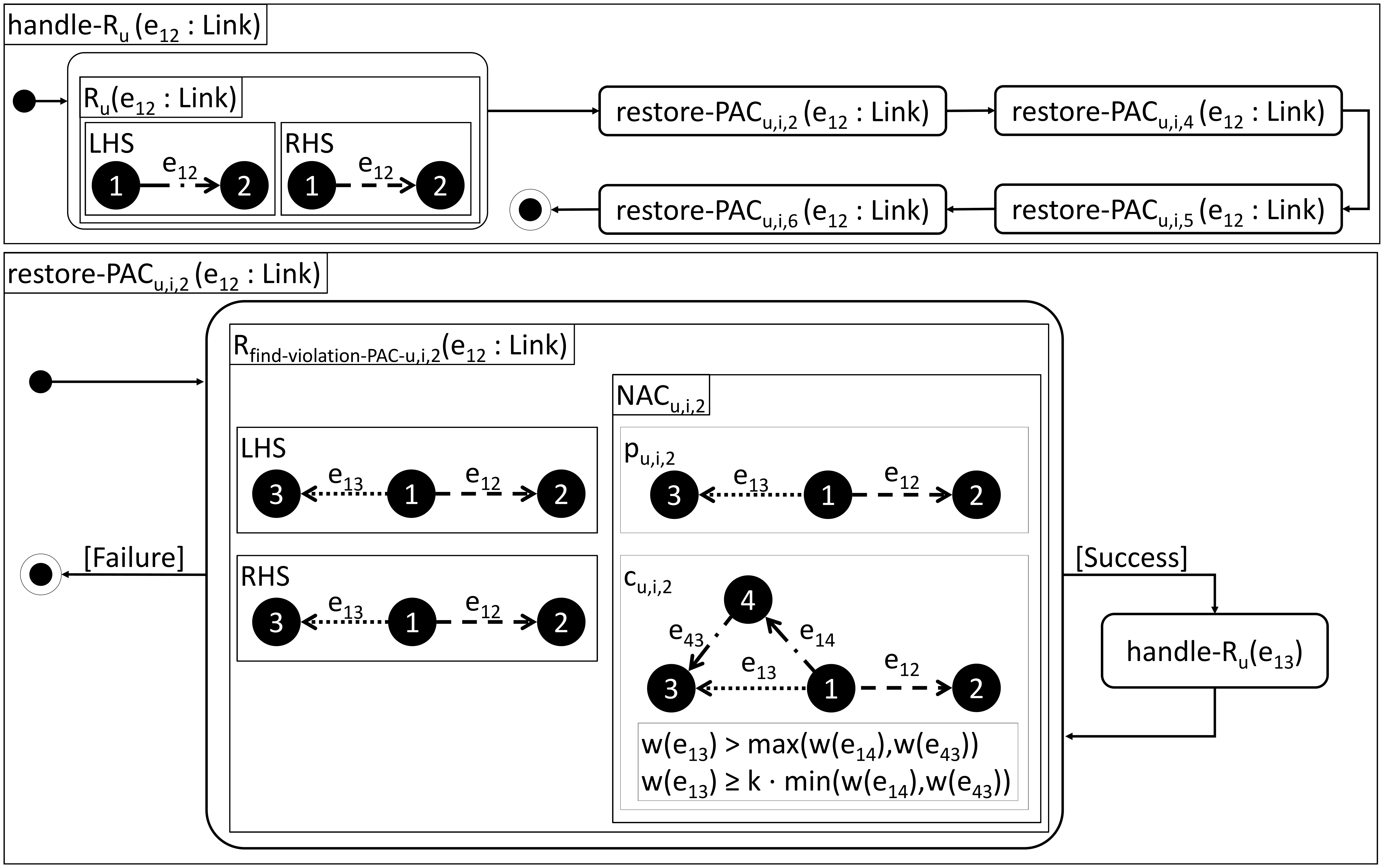}
    \end{center}
    \caption{Handler operation \handleUnclassification for \unclassificationRuleLong}
    \label{fig:repair-operations-for-unclassification-rule}
\end{figure}
The parameter list of \unclassificationRule contains one link variable \linkVariableOneTwo.
Corresponding to the four application conditions, we create four restoration operations: \restorationOperation{\PACui{2}}, \restorationOperation{\PACui{4}}, \restorationOperation{\PACui{5}}, and \restorationOperation{\PACui{6}}.
For conciseness, we only show and discuss the first restoration rule, \restorationOperation{\PACui{2}}.
The violation identification operation of \PACui{2} consists of the sole rule  \violationIdentificationRule{\PACui{2}} and is, therefore, inlined here.
This rule essentially identifies some inactive link \linkVariableOneThree for which \linkVariableOneTwo is part of the sole match of the conclusion of the \inactiveLinkConstraintKTCLong.
To identify such situations, the positive application condition \PACui{2} is turned into the  \emph{negative} \NACui{2} of \violationIdentificationRule{\PACui{2}}.
For each identified link \linkVariableOneThree, unclassifying \linkVariableOneTwo would result in a violation of \PACui{2}---and thereby of \unclassifiedLinkConstraint---if no countermeasures are taken because the sole match of the conclusion of \PACui{2} at \linkVariableOneThree would be destroyed.

To resolve the identified, immanent constraint violation at \linkVariableOneThree, we apply the default violation resolution strategy and unclassify \linkVariableOneThree in this case.
A more sophisticated strategy could try to establish an additional match of the conclusion of \inactiveLinkConstraintKTC, \eg, by classifying other links appropriately.
While this strategy is certainly a means to reduce the number of links that need to be unclassified, it is not generally applicable because we may easily come up with topologies that, for instance, do not contain alternate triangles. 
Therefore, we chose the default strategy here, which means that the violation resolution operation is a recursive call to \handleUnclassification with the problematic link variable \linkVariableOneThree as parameter.
The restoration operations for the other three \ACs (\PACui{4}, \PACui{5}, and \PACui{6}) are derived in an analogous manner.

The recursive invocation of \handleUnclassification is guaranteed to terminated.
The reason is that, at each invocation of \handleUnclassification, at least one classified link is being unclassified by invoking \unclassificationRule prior to the recursive invocations of \handleUnclassification.

\subsubsection{Deriving Repair Operations for the Context Event Rules}

Steps for deriving restoration operations for the \linkRemovalRuleLong and the \weightModificationRuleLong are completely analogous to the steps for the \unclassificationRuleLong and are, therefore, omitted for brevity.

\subsubsection{Context Event Rules and Connectivity}
A crucial assumption of our approach is that the topology is weakly consistent whenever \tcOperation is invoked, \idest, \tcOperation is designed to \emph{preserve} weak consistency, not to enforce it.
Between two invocations of \tcOperation, an arbitrary number of \CEs may occur.
After the modifications to the \CE rules described in this section, we may safely assume that an appropriate restoration operation is invoked whenever a \CE occurs.
Based on this assumption, we can prove the following theorem.
\begin{theorem}
    The refined \CE rules preserve weak connectivity on physically connected topologies.
\end{theorem}
\begin{proof-sketch}
    The refined \CE rules (with the restoration operations) preserve weak consistency.
    According to \Cref{thm:weak-consistency-weak-connectivity}, weak consistency implies weak connectivity;
    therefore, the claim follows.
\end{proof-sketch}
\subsection{Enforcing Termination of \tcOperation via Handler Operations}
\label{sec:enforcing-termination}

In this section, we first illustrate that the current version of \tcOperation (shown in \Cref{fig:ktc-algorithm-after-step2}) may run into situations that lead to an infinite execution in \Cref{sec:termination-example}.
Based on these observations and the concept of restoration operations, as described in \Cref{sec:enforcing-applicability}, we introduce a pre-processing operation that ensures that one of the two \TC rules, \activationRule or \inactivationRule, is applicable in each iteration of the \TC algorithm.
This step fulfills goal G3.

\subsubsection{Example of Non-Termination and Problem Analysis}
\label{sec:termination-example}
As \Cref{fig:sample-topology-for-shared-NACs} shows, there are topologies for which the current \TC algorithm may not terminate:
Starting with an entirely unclassified topology, the \TC algorithm activates \linkName{23} and inactivates \linkName{13}.
Now, the unclassified link \linkName{12} may be neither activated nor inactivated because the \activationRuleLong and the \inactivationRuleLong share two pairs of \NACs that may prevent the applicability of \emph{both} \TC rules in this situation.
In general, if \NACaa{2} prevents the application of \activationRule for a certain unclassified link, then \NACia{2} also prevents the application of \inactivationRule.
The same holds for \NACaa{8} and \NACia{8}, respectively.
\begin{figure}
\begin{center}
\subcaptionbox{Initial situation
\label{fig:sample-topology-for-shared-NACs}}[.7\textwidth]{        \includegraphics[width=.7\textwidth]{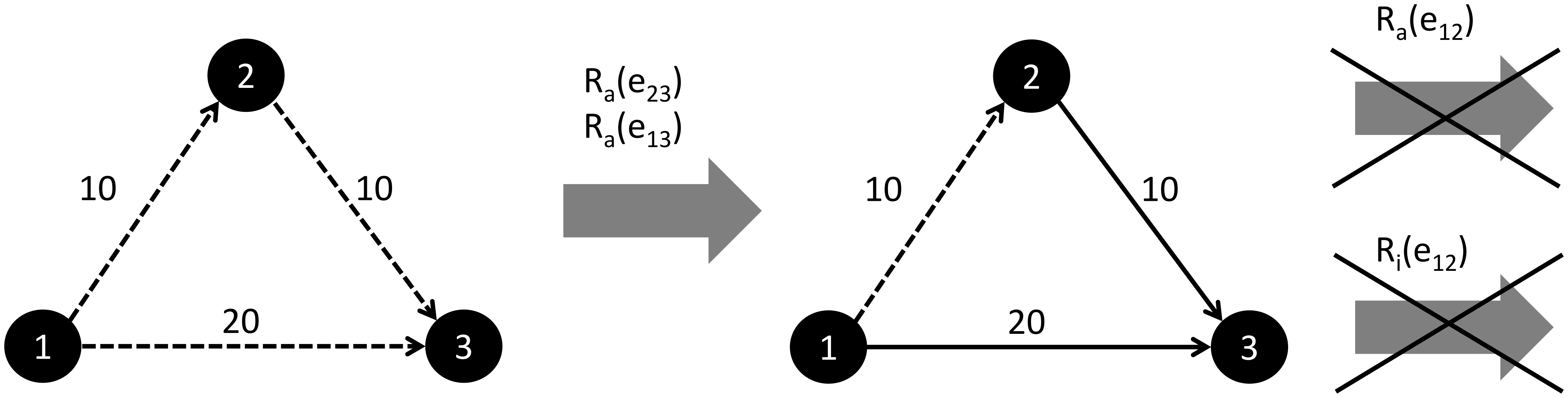}}

\vspace{2ex}

\subcaptionbox{Possible solution: re-unclassifying weight-maximal link \linkName{13}\label{fig:sample-topology-for-shared-NACs-solution}}[\textwidth]{        \includegraphics[width=\textwidth]{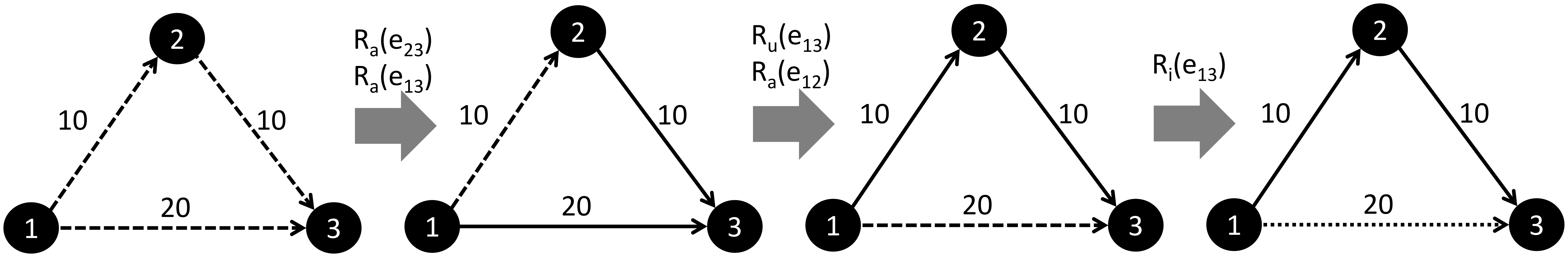}}
\end{center}
\caption{Example of non-termination of \tcOperation (\Cref{fig:ktc-algorithm-after-step2}) and a possible solution}
\end{figure}
\Cref{fig:sample-topology-for-shared-NACs-solution} depicts a possible solution for this particular situation:
When processing the unclassified \linkNameLong{12}, the wrong classification of \linkNameLong{13} is reverted by unclassifying \linkName{13}.
Afterwards, the \activationRuleLong is successfully applied to \linkNameLong{12}, leaving only one unclassified link, \linkName{13}, which finally gets inactivated.
The preceding discussion shows that 
\begin{inparaenum}
\item processing links in an \enquote{unfavorable} order may cause wrong classifications, and 
\item at least two possible solutions to this problem exist.
\end{inparaenum}

As a first solution, we could pre-determine the order of processing of the unclassified links to ensure that less (or even no) unclassifications are necessary.
Currently, the processing order is non-deterministic:
the \findUnclassifiedLinkRuleLong may return any of the unclassified links in the topology.
Pre-determining the processing order brings the benefit of no or few required re-classifications of links, but it may be non-trivial to identify such a suitable link processing order.

As a second solution, we could re-unclassify links whose state is known to be \enquote{wrong.}
In the case of \ktc, \enquote{wrong} means that a particular link blocks the activation and inactivation of an unclassified link.
This solution brings the benefit that the processing order needs not be pre-determined, leaving this decision to the underlying \GT engine or code generator that executes the \GT rules.
The downside is that links may be re-unclassified many times if the re-unclassification logic is not properly specified.

\subsubsection{Introducing Handler Operations}
We propose to apply the second solution for the \ktc case study.
We take the very same approach as in \Cref{sec:enforcing-applicability}:
Our goal is to eventually drop two critical pairs of negative application conditions: \NACaa{2} and \NACia{2} plus \NACaa{8} and \NACia{8}.
To this end, we treat these \ACs in the same way as the additional \ACs of the unrestrictable rules:
We may replace the application of the \activationRuleLong and the \inactivationRuleLong with invocations of appropriate handler operations.
In \handlerOperation{\activationRule}, we may then introduce the restoration operations \restorationOperation{\NACaa{2}} and \restorationOperation{\NACaa{8}}, and
in \handlerOperation{\inactivationRule}, we may then introduce the restoration operations \restorationOperation{\NACia{2}} and \restorationOperation{\NACia{8}}.
The remaining \ACs---\PACii{1} and \NACaa{1}---are considered later.

This time, however, we proceed slightly differently, achieving the same result but with improved performance and readability:
Instead of inserting a restoration operation for \emph{each} of the two affected rules, we only invoke the restoration operation \emph{once} prior to applying the first rule (\activationRule in this case).
This is valid because the derived restoration operations would be identical for both rules, and no new constraint violations may be created by invoking the \activationRuleLong, because the \inactivationRuleLong is only tried if the first rule application has failed.

\Cref{fig:ktc-algorithm-after-step3} depicts the resulting, and final, version of the operation \tcOperation, the additional handler operation, and the \activationRuleLong and \inactivationRuleLong after removing the \NACs that are now always fulfilled.
The handler operation consists of two loops that unclassify the weight-maximal links \linkName{13} and \linkName{23} in each match of the premise of corresponding \NACs.
For each unclassification of an active link \linkName{13}/\linkName{32}, the repair operation \repairUOperation of the \unclassificationRuleLong (as shown in \Cref{fig:repair-operations-for-unclassification-rule}) is invoked to repair possible constraint violations.

Note that the premise of \NACaa{1} and the conclusion of \PACii{1} are equivalent, \idest, each match of \NACaa{1} is also a match of \PACii{1}.
This means that the \activationRuleLong in \Cref{fig:ktc-algorithm-after-step3} is applicable to a given link \linkName{12} \emph{if and only if} the \inactivationRuleLong is not applicable to the same link.
For this reason, we could remove \PACii{1} from \inactivationRule for performance and readability reasons.
\begin{figure}
    \begin{center}
        \includegraphics[width=\textwidth]{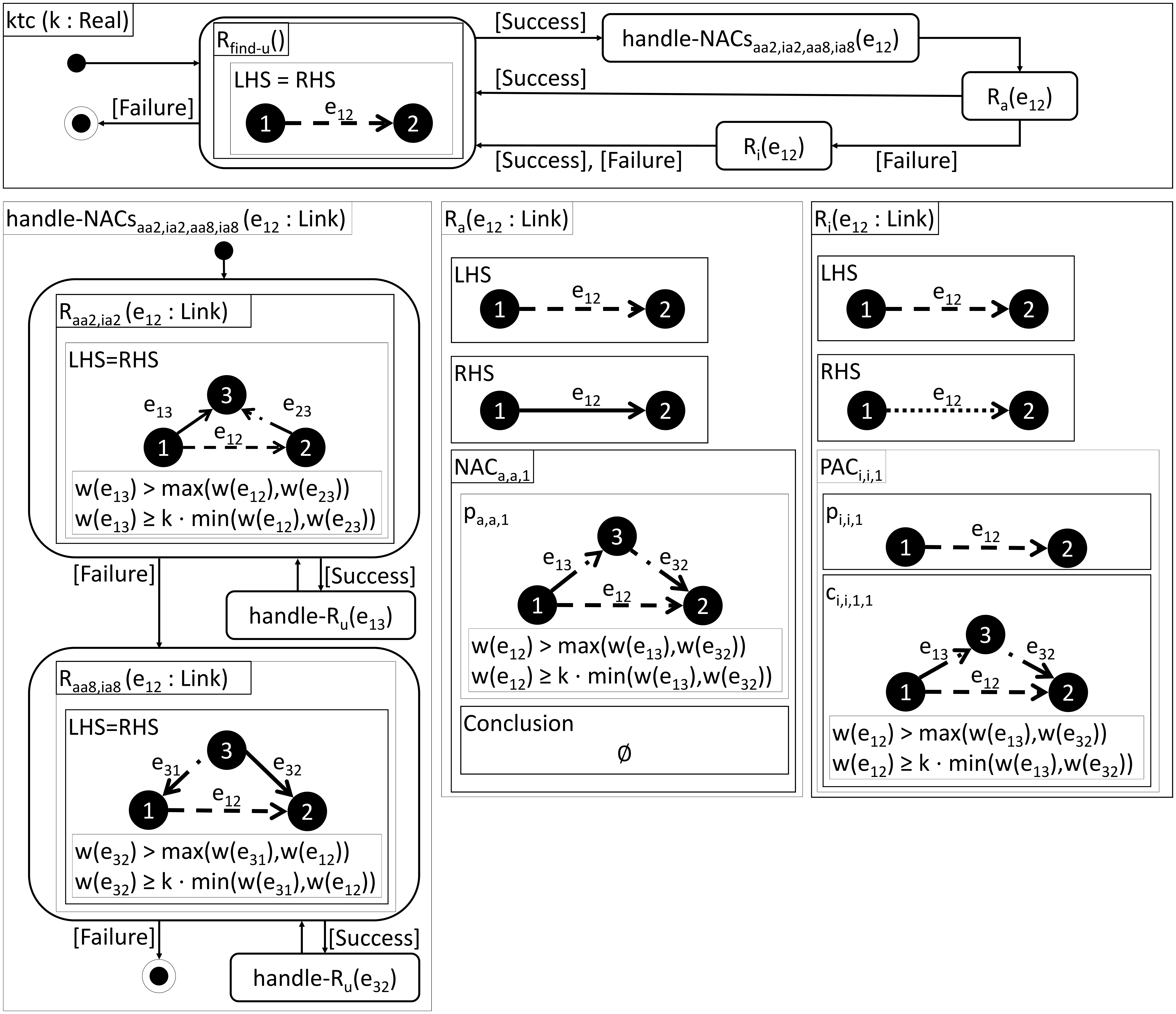}
    \end{center}
    \caption{\GToperationNN \tcOperation, the final specification of \ktc after enforcing termination in \Cref{sec:enforcing-termination}}
    \label{fig:ktc-algorithm-after-step3}
\end{figure}

\subsubsection{Proof of Termination}
After showing that the \TC algorithm produces strongly consistent topologies, which fulfill the \activeLinkConstraintKTCLong, the \inactiveLinkConstraintKTCLong, and the \unclassifiedLinkConstraintLong, it remains to show that the \TC algorithm, shown in \Cref{fig:ktc-algorithm-after-step3}, terminates.

\begin{theorem}[Termination]
    \tcOperation terminates for any input topology.
\end{theorem}
\begin{proof}   
    We have to show that the number of rule applications of the \findUnclassifiedLinkRuleLong is finite for any input topology.
    Consider a topology with link set $E$.    
    We consider the sequence of all link states $s_i(e_1),\hdots,s_i(e_m)$ with $m:=|E|$ after the $i$-th execution of \findUnclassifiedLinkRule, where the links are ordered according to their weight:
    \begin{align*}
        \forall v,w \in \{1,\dots,m\}: \weight(e_v) < \weight(e_w) \Rightarrow v < w
    \end{align*}
    Note that the right-to-left implication does not hold because no order is imposed on two links \linkName{v}, \linkName{w} that have an identical weight value.
    We compare two sequences of link states, $s_{i}$ and $s_j$, as follows:
    $s_i \prec s_j$ if and only if%
    \begin{inparaenum}[(i)]
        \item some link $e_v$ is unclassified in $s_i$ and active or inactive in $s_j$, and 
        \item the states of all links with a weight less than $e_v$ are identical in $s_i$ and $s_j$,
    \end{inparaenum} formally:
    \begin{align*}
    s_i \prec s_j \;:\Leftrightarrow\; 
        &\exists v \in \{1,\dots,m\}: s_i(e_v) = \UNCL \;\wedge\; s_j(e_v) \in \{\ACT,\INACT\}\\
        &\wedge \forall w \in \{1,\dots,v-1\} : s_i(e_w) = s_j(e_w)
    \end{align*}
    Note that any sequence that consists of active and inactive links only is an upper bound for $\prec$.
    
    We now show that $s_{i-1} \prec s_{i}$ for $i > 1$.
    Let $e_v$ be the link that is bound by applying the \findUnclassifiedLinkRuleLong.
    The pre-processing operation unclassifies only links $e_w$ if $\weight(e_w) > \weight(e_v)$ with $w > v$.
    The \activationRuleLong or the \inactivationRuleLong activate or inactivate $e_v$, respectively.
    
    Therefore, $s_{i-1}\prec s_i$ because
    \begin{inparaenum}
        \item the first $v{-}1$ elements of $s_{i-1}$ and $s_i$ are identical, and
        \item $s_{i-1}(e_w) = \UNCL$ and $s_i(e_w) \in \{\ACT,\INACT\}$.
    \end{inparaenum}
    The termination follows because any ordered sequence $s_1 \prec s_2 \prec \dots$ has a finite length.
\end{proof}

\paragraph{Example}
To illustrate the comparison of link state sequences, \Cref{tab:sample-topology-for-shared-NACs-solution-link-states} shows the link state sequences for the previous example in \Cref{fig:sample-topology-for-shared-NACs-solution}.
The columns are sorted by increasing weight.
In the third iteration, the pre-processing operation unclassifies \linkName{13}, which ensures that the \activationRuleLong may be applied to \linkName{12}.
\begin{table}
    \caption{Link states after each iteration of the example in \Cref{fig:sample-topology-for-shared-NACs-solution}}.
    \label{tab:sample-topology-for-shared-NACs-solution-link-states}
\begin{center}
\begin{tabular}{c|ccc}
    \toprule
    \textbf{Iteration}  & \multicolumn{3}{c}{\textbf{State after Iteration}}\\
     & $\state(\linkName{12})$ & $\state(\linkName{23})$ & $\state(\linkName{13})$ \\
    \midrule
    (initial) & \UNCL & \UNCL & \UNCL \\
    1 & \UNCL & \ACT & \UNCL \\
    2 & \UNCL & \ACT & \ACT \\
    3 & \ACT & \ACT & \UNCL \\
    4 & \ACT & \ACT & \INACT \\
    \bottomrule
\end{tabular}
\end{center}
\end{table}

This proof concludes the refinement of the \TC rules.
We have shown that the refined rules preserve the \activeLinkConstraintKTCLong and the \inactiveLinkConstraintKTCLong and that introducing the pre-processing operation enforces the termination of the \TC algorithm.
\subsection{Summary of Rule Refinement}
\label{sec:refinement-summary}

We conclude this section with a short overview of its results.
First, to ensure that the \TC and \CE rules preserve weak consistency (G1), we first applied the constructive approach by Heckel and Wagner~\cite{HW95} in \Cref{sec:refinement-methodology}.
Then, we introduced the complementary restoration step in \Cref{sec:enforcing-applicability}.
This step is necessary to fulfill the goal that the unrestrictable \GTrulesNN---the \unclassificationRuleLong and the \CE rules---should always be applicable (G2).
The restoration step systematically transforms all additional application conditions into restoration operations that eliminate (anticipated) constraint violations caused by applying the original \GT rules.
Finally, we observed that an unfavorable processing order of unclassified links may lead to a non-termination of \tcOperation, which contradicts goal G3.
We solved this problem by re-using the concept of restoration operations:
We transformed the problematic application conditions into appropriate restoration operations.
Thereby, we enforced the applicability of at least one of the \TC rules in each iteration of the \TC algorithm.

\section{Simulation-Based Evaluation}
\label{sec:evaluation}

In this section, we evaluate the incremental specification of \ktc constructed in \Cref{sec:rule-refinement}.
For this evaluation, we use the \GT tool \eMoflon~\cite{LAS14} to automatically transform the specification of \ktc into executable Java code.
For simulating \WSNs, we use the network simulation platform \Simonstrator~\cite{RSRS15}.

In \Cref{sec:eval-rqs}, we present the research questions to be answered during this evaluation.
In \Cref{sec:eval-setup}, we outline the evaluation setup.
In \Cref{sec:eval-rq-correctness,sec:eval-rq-incrementality,sec:eval-rq-generalizability,sec:eval-rq-performance}, we present and discuss the results to answer the four identified research questions.
In \Cref{sec:eval-threats}, we discuss threats to validity and summarize the results of the evaluation.

\subsection{Research Questions}
\label{sec:eval-rqs}
We emphasize that the focus of this paper is \emph{not} to develop a new \TC algorithm but to derive a novel, incremental version of an established \TC algorithm using the correct-by-construction approach presented in the preceding sections.
Therefore, the following research questions focus on the general validity, applicability, and performance of our approach.
Our goal is not to perform a detailed analysis of the \TC algorithm \wrt network-specific metrics (\eg, network lifetime, energy consumption).
The answers to the following research questions are summarized in \Cref{tab:answers-to-rqs} at the end of this section.

\textbf{\RQCorrectnessLong:}
The implementation that is derived from the \GT-based specification should always be correct.
The implementation is correct if the topology is always weakly consistent and  if \tcOperation always terminates and returns a strongly consistent topology.
More specifically, in \Cref{sec:eval-rq-correctness}, we will address the following questions:
\begin{itemize}
    \item Is the topology always weakly consistent after \CE handling?
    \item Does \tcOperation always terminate? Is the resulting topology strongly consistent?
\end{itemize}    

\textbf{\RQIncrementalityLong:}
Incrementality is an fundamental property of the \TC{} algorithms developed using our approach.
We quantitatively assess incrementality in terms of the number of link state modifications required to handle \CEs and to re-optimize the topology.
More specifically, in \Cref{sec:eval-rq-incrementality}, we will address the following questions:
\begin{itemize}
\item 
Which properties of a topology (\eg, node count, node density) influence the cost of handling \CEs and invoking \TC?
\item How large is the cost of executing incremental \ktc compared to batch \ktc?
\end{itemize}

\textbf{\RQPerformanceLong:}
From a practical perspective, it is crucial that simulations may be performed in sensible time.
While we only have limited influence on the execution time of the actual simulation (\eg, due to simulated network traffic), the overhead of executing \tcOperation should remain reasonable.
The purpose of this research question is to evaluate this overhead quantitatively.
More specifically, in \Cref{sec:eval-rq-performance}, we will address the following questions:
\begin{itemize}
    \item What fraction of execution time is spent on executing \tcOperation?
    \item Is the total execution time of \tcOperation sensible for topologies of realistic size?
\end{itemize}

\textbf{\RQGeneralizabilityLong:}
The proposed methodology shall help reducing the gap between specification and implementation that is pertinent in the \WSN community.
This research question addresses the general applicability of the proposed approach.
More specifically, in \Cref{sec:eval-rq-generalizability}, we will discuss the following questions:
\begin{itemize}
\item 
Is meta-modeling suitable to specify all relevant properties of \WSN topologies?
\item 
Are \GT rules suitable to specify \CEs and \TC operations?
\item 
Are graph constraints and the proposed refinement algorithm suitable to specify and ensure the preservation of consistency constraints and optimization goals of typical \TC algorithms?
\item
Is a \TC developer able to apply the described methodology?
\end{itemize}

\subsection{Evaluation Setup}

\label{sec:eval-setup}
Following best practice for simulation studies (\eg, \cite{Kurkowski2005,Hiranandani2013}), we document the evaluation setup in detail to foster reproducibility of our results.
The technical platforms of this evaluation are the \GT tool \eMoflon \cite{LAS14} and the network evaluation platform \Simonstrator~\cite{RSRS15} with its contained network simulator \PFS~\cite{SGRNKS11}, as shown in \Cref{fig:evaluation-setup} and as described in detail in the following.
\begin{figure}
    \begin{center}
        \includegraphics[width=\textwidth]{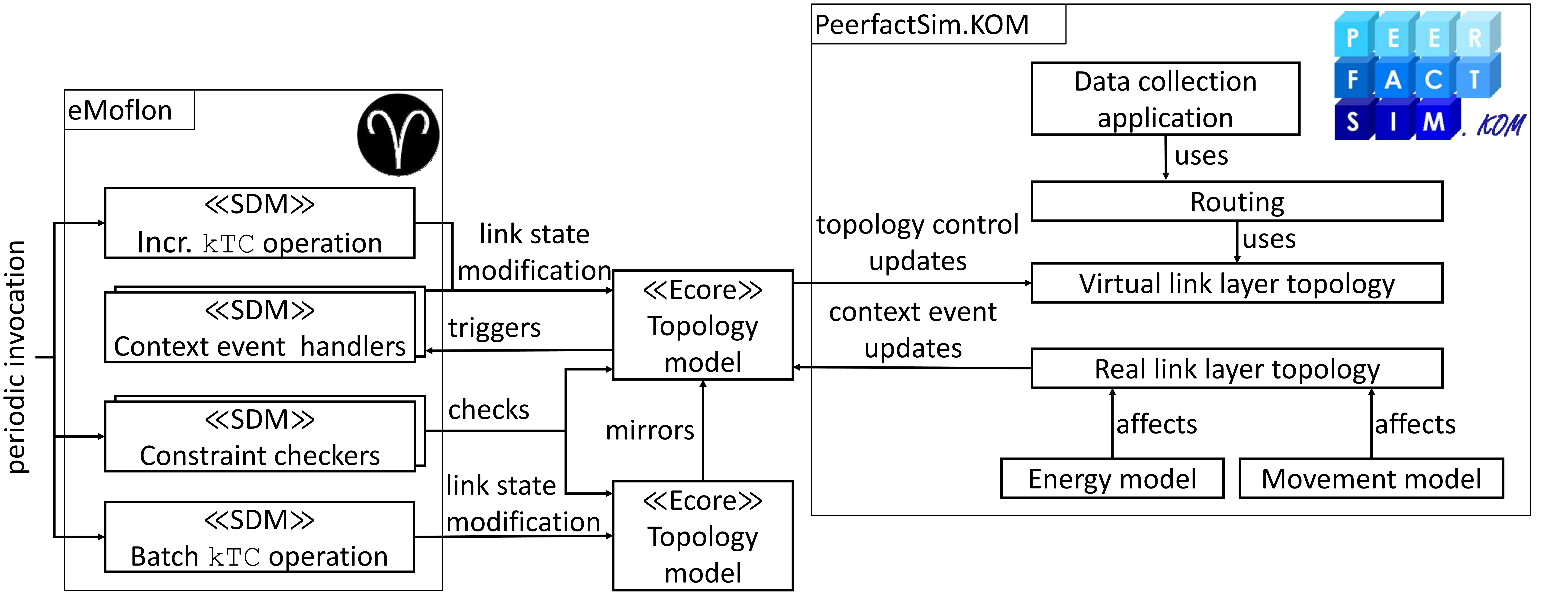}
        \caption{Overview of the evaluation setup}
        \label{fig:evaluation-setup}
    \end{center}
\end{figure}

\paragraph{eMoflon}
The \CE handlers and the \GToperationNN \tcOperation have been specified using a visual syntax of story diagrams in \eMoflon.
The specification is compiled into EMF-compliant Java code, which operates on an EMF-based topology model.
The topology meta-model has been specified in Ecore~\cite{SBMP08}.
Whenever the topology model is modified by a \CE in the simulator, the corresponding handler operation is triggered.
The \TC algorithm \tcOperation periodically updates the topology model via \LSMs.

\paragraph{Simulator}
\PFS~\cite{SGRNKS11} is a time-discrete network simulator that allows to simulate protocols on the underlay (\idest, the physical and link layer of the ISO-OSI reference model~\cite{Zimmermann1980}) as well as on the overlay (\idest, network, transport and application layer~\cite{Zimmermann1980}).

\paragraph{Topology control via virtual topologies}
Upper-layer protocols typically operate on the topology view provided by lower-layer protocols.
We created a generic \emph{topology framework} that introduces \emph{virtual} topologies:
The virtual topology is a view of the corresponding real topology, which is used by upper-layer applications instead of the real topology.
In this evaluation, we consider the real and virtual \emph{link layer} topology.
Using a virtual link layer topology allows us to apply \TC transparently, \idest, without the upper-layer protocols noticing that they operate on a modified view of the real topology.
The real link layer topology is affected by two types of \CEs:
\begin{itemize}[(i)]
\item
The \emph{energy model} of a node simulates the depletion of its battery over time.
The energy model determines the initial charge of the battery and the amount of energy that transmitting a message of a certain size requires.
We say that a node is \emph{alive} at a certain point in time if its battery is not empty.
\item
The \emph{movement model} specifies the movement of each node inside the bounds of the simulation area.
In this scenario, we use a modified version of the Gauss-Markov movement model~\cite{Camp2002}.
To simulate the typically low dynamics in \WSNs, we introduced a \emph{hesitation probability} of \evalParHesitation:
This is the probability that a node does not move within a time step.
Node movement causes link weight modifications and link additions/removals if one node enters/leaves another node's transmission range.
\end{itemize}
The modifications of the topology model caused by \CEs trigger the corresponding handler operations.
In the reverse direction, \LSMs of the topology model that result from an invocation of \tcOperation are propagated to the virtual link layer topology.
The virtual link layer topology contains all links in the original link layer topology that are either active or unclassified, which ensures that the link layer topology remains connected even in the presence of \CEs.

In this evaluation, \TC is invoked periodically every \evalParTCInterval (simulation time).
The parameter $k$ of \ktc{} is set to \evalParK, which is a typical value~\cite{KSSVMS14}.
The set of \CEs is collected during this period and handled prior to invoking \TC{}.
We subsume the subsequent execution of \CE handling and \TC invocation under the term \emph{\TC run}, \idest, every \evalParTCInterval, a \TC run is performed.

\paragraph{Overlay}
On top of the virtual link layer topology, a typical (many-to-one) data collection application~\cite{Yick2008} is running:
All \emph{sensor} nodes periodically send data messages to a central \emph{base station} node, which has a fixed power supply (modeled using a very large battery capacity of \evalParBatteryTarget).
The transmitted messages are routed by a global knowledge routing algorithm, which operates on the complete virtual link layer topology.

\Cref{tab:evaluation-parameters-overview} summarizes the parameters of the simulation setting that are fixed for all configurations.
We run each simulation in \emph{terminating mode}, \idest, we abort the simulation after \evalParSimulationDuration, regardless of the remaining number of alive nodes.
This means that, in each simulation run, \TC is performed $\frac{\evalParSimulationDuration}{\evalParTCInterval} - 1 =  \evalParMeasurementsPerRun$ times because we neglect the initial \TC run (at \evalParTCInterval of simulated time) for technical reasons.
All experiments were performed on a 64-bit machine with an Intel i7-2600~CPU (\SI{2}{cores}, \SI{3.7}{\giga\hertz}) and \SI{8}{\giga\byte} of RAM running Windows 7 Professional.
\begin{table}
\begin{minipage}{\textwidth}
\begin{center}

    \caption{Common parameter values of all simulation configurations}
    \label{tab:evaluation-parameters-overview}
    \begin{tabular}{ll}
        \toprule
        \textbf{Parameter} & \textbf{Value}\\
        \midrule
        \multicolumn{2}{l}{\textbf{Simulation}}\\
        \midrule
        Network simulator & \PFS / \Simonstrator 2.4~\cite{RSRS15}\\
        GT-tool & eMoflon 2.15.0~\cite{LAS14}\\
        Number of runs per configuration & \evalParSeedCount\\
        Seeds & 1\dots\evalParSeedCount\\
        Simulation type & terminating (after \evalParSimulationDuration of simulated time)\\
        Measurements per run & \evalParMeasurementsPerRun\\
        Code availability & as SHARE virtual machine \cite{GM11}\footnote{For details, please see \url{https://github.com/Echtzeitsysteme/CorrectByConstructionTC-JVLC16}}\\
        World size & \emph{variable}  (see \Cref{tab:evaluation-configurations-overview})\\
        \midrule
        \multicolumn{2}{l}{\textbf{Node configuration}}\\
        \midrule
        Source node count (initial) & \emph{variable} (see \Cref{tab:evaluation-configurations-overview}) \\
        Target node count & 1\\
        Transmission radius & \SI{131}{\meter}\\
        Movement model & Gauss-Markov~\cite{Camp2002} with hesitation ($p=\evalParHesitation$)\\
        Placement model & uniform at random\\       
        Initial battery level (sources) & \evalParBatterySource\\
        Initial battery level (target) & \evalParBatteryTarget\\
        \midrule
        \multicolumn{2}{l}{\textbf{Link layer and topology control}}\\
        \midrule
        Link count (initial) & \emph{variable} (see \Cref{tab:evaluation-configurations-overview})\\
        Link layer & IEEE 802.11 Ad-Hoc\\
        \TC algorithm & incremental \ktc\\
        $w$ calculation & Euclidean distance between incident nodes\\
        $k$ parameter & \evalParK\\
        Interval of \TC runs& \evalParTCInterval\\
        \TC runs per simulation run & \evalParMeasurementsPerRun\\
        \midrule
        \multicolumn{2}{l}{\textbf{Overlay}}\\
        \midrule
        Application & Data collection\\
        Communication pattern & many-to-one\\
        Routing & Global knowledge\\
        Message size & \SI{1}{\kilo\byte}\\
        Message transmission frequency & 1 message every \SI{10}{\second}\\
        \bottomrule
    \end{tabular}
\end{center}
\end{minipage}
\end{table}

To obtain a set of \emph{(evaluation) configurations}, we vary the node count of the initial topology as well as the side length of the quadratic area onto which the nodes are distributed uniformly at random.
In total, we investigate six different evaluation configurations, summarized in \Cref{tab:evaluation-configurations-overview}, that result from varying the initial node count $n \in \{100, 1000\}$ and the world size $w \in \{\SI{750}{\meter}, \SI{500}{\meter}, \SI{250}{\meter}\}$ for $n = 100$ and $w \in \{\SI{2000}{\meter}, \SI{1000}{\meter}, \SI{1500}{\meter}\}$ for $n = 1000$.
Each configuration is simulated \evalParSeedCount times to obtain representative results.
In the table, the configurations are ordered from the smallest and sparsest initial topology (\cfgSmallSparse) to the largest and densest topology (\cfgLargeDense).
\begin{table}
    \begin{center}
        \caption{Variable parameters in the six evaluation configurations (\textbf{$n$}: initial node count,
            \textbf{$m$}: initial link count in real link layer, 
            \textbf{$w$}: side length of quadratic simulation area,
            \textbf{$h$}: average hop count to base station, 
            \textbf{$d_{\text{out}}$}: average out-degree,
            values of $m$, $h$, and $d_{\text{out}}$ are averaged over all \evalParSeedCount repetitions)}
        \label{tab:evaluation-configurations-overview}
        \begin{tabular}{crrrrr}
            \toprule
            \multicolumn{1}{c}{\textbf{ID}} & 
            \multicolumn{1}{c}{\textbf{$n$}} & 
            \multicolumn{1}{c}{\textbf{$w$ [m]}} & 
            \multicolumn{1}{c}{$m$} & 
            \multicolumn{1}{c}{\textbf{$h$}} & 
            \multicolumn{1}{c}{\textbf{$d_{\text{out}}$}}
            \\
            \midrule
            \cfgSmallSparse&100&750&\numprint{812}&3.6&7.7\\
            \cfgSmallMedium&100&500&\numprint{1651}&2.1&16.0\\
            \cfgSmallDense&100&250&\numprint{5161}&1.2&51.2\\
            \midrule
            \cfgLargeSparse&\numprint{1 000}&\numprint{2 000}&\numprint{12683}&7.8&12.1\\
            \cfgLargeMedium&\numprint{1 000}&\numprint{1 500}&\numprint{22092}&5.5&21.6\\
            \cfgLargeDense&\numprint{1 000}&\numprint{1 000}&\numprint{48072}&3.6&47.5\\
            \bottomrule
        \end{tabular}
    \end{center}
\end{table}
The node (out-)degree $d_{\text{out}}$ is an indicator of the density of the topology:
A high node degree indicates that a topology is dense, while a low node degree indicates a sparse topology.
The initial topology of the two sparsest topologies is physically disconnected in 3 out of \evalParSeedCount simulation runs for \cfgSmallDense and 1 out of \evalParSeedCount simulation runs for \cfgLargeSparse.
The lack of connectivity is not a problem because the topology
\begin{inparaenum}
\item may become physically connected due to node movements and
\item will eventually become physically disconnected as soon as more and more nodes run out of energy.
\end{inparaenum}
The average hop count $h$ to the base station is an indicator of the extent of the considered topology.

\paragraph{Batch \ktc}
This evaluation shall illustrate benefits of incremental \TC compared to batch \TC.
For this reason, we also implemented batch \ktc using SDM, which operates on a separate topology model that mirrors the structure of the topology model that is modified by  incremental \ktc.
Prior to each application of batch \ktc, all links in the mirrored topology are being unclassified.
In the following, we will refer to the incremental and batch variants of \ktc as \iktc and \bktc, respectively.

\subsection{\RQCorrectnessLong}
\label{sec:eval-rq-correctness}

A central contribution of our approach is that the derived \TC algorithms and \CE handlers are \emph{correct by construction}.
Recall that \emph{correctness} means that
\begin{inparaenum}
\item the \TC algorithm always terminates and produces strongly consistent output topologies for weakly consistent input topologies, and 
\item the \CE handlers preserve weak consistency.
\end{inparaenum}
Evidently, this property only holds if all steps of the methodology are performed correctly.
\Cref{tab:rq-correctness-overview} summarizes steps in the methodology where correctness of the derived \TC algorithm and \CE handlers may be corrupted, and it shows the steps that we took to mitigate or eliminate these problems.
\begin{table}
\begin{center}
\caption{Threats to correctness and applied mitigation strategies (\RQCorrectness)}
\label{tab:rq-correctness-overview}
\begin{tabular}{p{2.5cm}|p{3cm}p{7cm}}
    
\toprule

\textbf{Step} & \textbf{Potential Threat} & \textbf{Mitigation Strategy}\\

\midrule

Constraint specification (\Cref{sec:constraints})&
The constraints do not represent the intended objective.&
We chose the well-established \ktc algorithm, which has already proved to achieve energy savings in a testbed evaluation~\cite{SZS15}, and carefully translated its basic ideas into graph constraints.

\\

\midrule

Rule refinement (\Cref{sec:refinement-methodology})&
The derived application conditions fail to preserve the specified graph constraints.&
The translation of graph constraints into application conditions is theoretically correct, but has been performed by hand for this evaluation, which may not be accurate.
To increase confidence in our implementation, the graph constraints that constitute weak and strong consistency are checked during simulation after \CE handling and \TC invocation, respectively.
The large number of \CEs and \TC invocations during the simulation serve as test suite for the implementation.
To mitigate the potential errors of performing the manual translation, an automatic derivation of the application conditions would be possible.\\

\midrule

Derivation of \CE handlers (\Cref{sec:enforcing-applicability})&
The \CE handlers fail to eliminate all possible violations of their original application conditions.&
The check of weak consistency after every \CE handling serves increases the confidence in the derived handlers.\\

\midrule

Code generation (\Cref{sec:evaluation})&
The generated Java code fails to implement the \GT-based specification.&
We rely on \eMoflon, which is practically proven and harnessed by a test suite containing more than 200 unit and system tests.\\

\bottomrule

\end{tabular}
\end{center}
\end{table}

\paragraph{Metric}

Clearly, observing zero constraint violations in the simulation study is necessary, but not sufficient, to confirm the correctness of the implementation.
Therefore, we checked for weak consistency after invoking the \CE handling and for strong consistency after invoking \iktc and \bktc.
The checkers for weak and strong consistency were implemented using SDM and were unit-tested to ensure correct functionality.

\paragraph{Results}
In total, no consistency violations were detected.
Context event handling and \TC were invoked 120 times per simulation run, which amounts to \numprint{3600} consistency checks per configuration and \numprint{21 600} constraint checks during the whole evaluation.

\paragraph{Discussion}
Every single consistency check is a test of the incremental (and batch) implementation.
The total amount of \numprint{21 600} successful consistency checks during the simulation considerably increases the confidence in the correctness of our implementation.
Still, these results do not constitute a formal proof of correctness of the implementation.
Therefore, we may give the following answers to \RQCorrectness:
During our simulation study, \tcOperation always terminated with strongly consistent output topologies, and the \CE handling always preserved weak consistency.

\subsection{\RQIncrementalityLong}
\label{sec:eval-rq-incrementality}

\paragraph{Metrics}
In the literature, no consensus exists considering how to assess incrementality of an algorithm.
In this evaluation, we will assess incrementality in terms of boundedness~\cite{Ramalingam1993}:
An algorithm is \emph{incremental \wrt to boundedness} if the cost of reacting to a change depends on the size of the change and of properties that can be calculated with local knowledge only (\eg, the out-degree of a node).
In our setting, changes are sets of \CEs and the cost of reacting to a change is the required number of \LSMs during the \CE handling and the \TC invocation after a set of \CEs.
More precisely, we define the \emph{\metricScopeLong} of a \TC run $r$ as follows:
\begin{align*}
    \text{scope of } r =\,& \text{number of \LSMs required for \CE handling of } r\\
    +\,& \text{number of \LSMs required for \TC invocation of } r
\end{align*}
Our analysis of the \CE and \TC operations reveals the following results concerning the theoretical boundedness of the scope for all \CEs, as summarized in \Cref{tab:eval-rq-incr-theoretical-boundedness}:
\begin{table}
\begin{center}
\caption{Theoretical boundedness \wrt scope per type of \CE}
\label{tab:eval-rq-incr-theoretical-boundedness}
\begin{tabular}{lccc}
\toprule
\multicolumn{1}{c}{\textbf{Context Event}}&
\multicolumn{1}{c}{\textbf{CE Handl. Bounded?}}&
\multicolumn{1}{c}{\textbf{TC Invoc. Bounded?}}&
\multicolumn{1}{c}{\textbf{\TC Run Bounded?}}
\\
\midrule
Node addition & \OK & \OK &\OK\\
Node removal & \OK & \OK&\OK\\
Link addition & \OK & \notOK&\notOK\\
Link removal  & \notOK & \notOK&\notOK\\
Weight mod. & \notOK & \notOK&\notOK\\
\bottomrule
\end{tabular}
\end{center}
\end{table}
Node addition and removal entail no \LSMs and, therefore, \tcOperation is bounded for these \CEs.
Link addition is bounded during \CE handling because the new link is unclassified and may not violate weak consistency (see \Cref{sec:refinement-ce-inact,sec:refinement-ce-act}).
In contrast, the subsequent \TC invocation may be unbounded due to the \NAC{} handling in \tcOperation (see \Cref{sec:enforcing-termination}).
In total, \iktc is theoretically unbounded for link additions.
For link removal and weight modifications, the \CE handling (and therewith the \TC invocation) is already unbounded due to the link unclassifications required for restoring weak consistency (see \Cref{sec:enforcing-applicability}).

Still, it is essential that incrementality holds for representative, realistic scenarios.
For this reason, we pratically evaluate the scope of \iktc for each configuration.
To assess the influence of the average node out-degree in the topology, we measure the \emph{\metricDegreeNormalizedScopeLong} of a \TC run $r$, which is defined as follows: 
\begin{align*}
    \text{\metricDegreeNormalizedScopeLong of } r
    = 
    \frac%
    {\text{scope of } r}%
    {\text{average node out-degree in real link layer topology}}
\end{align*}
The degree-normalized scope is measured in \LSMs (per link).
To compare \iktc with \bktc, we measure the \emph{\metricIkTCVsBkTCLong} of a \TC run $r$, which is defined as follows:
\begin{align*}
\text{\metricIkTCVsBkTCLong}
= 
\frac%
{\text{number of LSMs during } r \text{ when applying \iktc} }%
{\text{number of LSMs during } r \text{ when applying \bktc} }
\end{align*}

\paragraph{Results}

The four plots in \Cref{fig:eval-rq-incr-scope-by-degree-vs-node-count} show the \metricDegreeNormalizedScopeLong \vs the number of alive nodes in the topology of the four most extreme scenarios (in terms of size and node density): \cfgSmallSparse, \cfgSmallDense, \cfgLargeSparse, and \cfgLargeDense.
Each of the four plots contains a total of \evalParMeasurementsPerRun data points that correspond to the \evalParMeasurementsPerRun \TC runs (\idest, \CE handling and \TC invocation), each of which processes the set of \CEs that has been collected during the previous \evalParTCInterval of simulation time, at \SI{20}{\minute}, \SI{30}{\minute}, \dots, \evalParSimulationDuration.
For technical reasons, there is no data point for the first \TC run at \evalParTCInterval.
Each data point represents the average \metricDegreeNormalizedScopeLong of the \evalParSeedCount results at the same particular point in time in all \evalParSeedCount simulation runs.
Please note that the y-axis is logarithmic due to the large range of the data.

\begin{figure}
\begin{center}
    \def\subcaptionboxWidthTwo{.45\textwidth}
    \def\filename{rq2_incr_scope-normalized-by-node-outdgree_vs_node-count_log}    
    \def\figALabel{\cfgSmallSparse
        \label{fig:eval-rq-incr-scope-by-degree-vs-node-count-small-sparse}}
    \def\figBLabel{\cfgSmallDense
        \label{fig:eval-rq-incr-scope-by-degree-vs-node-count-small-dense}}
    \def\figCLabel{\cfgLargeSparse
        \label{fig:eval-rq-incr-scope-by-degree-vs-node-count-large-sparse}}
    \def\figDLabel{\cfgLargeDense
        \label{fig:eval-rq-incr-scope-by-degree-vs-node-count-large-dense}}
    \def\figA{\rootOfConfiguration{0099}{0750}/\filename}
    \def\figB{\rootOfConfiguration{0099}{0250}/\filename}
    \def\figC{\rootOfConfiguration{0999}{2000}/\filename}
    \def\figD{\rootOfConfiguration{0999}{1000}/\filename}
    \subcaptionbox%
        {\figALabel}[\subcaptionboxWidthTwo]%
        {\includegraphics[width=\subcaptionboxWidthTwo]{\figA}}
    \subcaptionbox%
        {\figBLabel}[\subcaptionboxWidthTwo]%
        {\includegraphics[width=\subcaptionboxWidthTwo]{\figB}}
    
    \subcaptionbox%
        {\figCLabel}[\subcaptionboxWidthTwo]%
        {\includegraphics[width=\subcaptionboxWidthTwo]{\figC}}
    \subcaptionbox%
        {\figDLabel}[\subcaptionboxWidthTwo]%
        {\includegraphics[width=\subcaptionboxWidthTwo]{\figD}}
\end{center}
\caption{Comparison of logarithmic \metricDegreeNormalizedScopeLong and \metricAliveNodeCountLong for \RQIncrementality}
\label{fig:eval-rq-incr-scope-by-degree-vs-node-count}
\end{figure}

The four plots in \Cref{fig:eval-rq-incr-iktc-vs-bktc} show the \metricIkTCVsBkTCLong \vs the number of alive nodes.
As in \Cref{fig:eval-rq-incr-scope-by-degree-vs-node-count}, each data point represents the results of a \TC run at a particular point in time during the simulation, averaged across all \evalParSeedCount simulation runs per configuration.
Please note that also here the y-axis is logarithmic due to the large range of the data.

\begin{figure}
    \begin{center}
        \def\subcaptionboxWidthTwo{.45\textwidth}
        \def\filename{rq2_incr_i-ktc-vs-b-ktc_vs_node-count-alive_log}
        \def\figALabel{\cfgSmallSparse
            \label{fig:eval-rq-incr-iktc-vs-bktc-small-sparse}}
        \def\figBLabel{\cfgSmallDense
            \label{fig:eval-rq-incr-iktc-vs-bktc-small-dense}}
        \def\figCLabel{\cfgLargeSparse
            \label{fig:eval-rq-incr-iktc-vs-bktc-large-sparse}}
        \def\figDLabel{\cfgLargeDense
            \label{fig:eval-rq-incr-iktc-vs-bktc-large-dense}}
        \def\figA{\rootOfConfiguration{0099}{0750}/\filename}
        \def\figB{\rootOfConfiguration{0099}{0250}/\filename}
        \def\figC{\rootOfConfiguration{0999}{2000}/\filename}
        \def\figD{\rootOfConfiguration{0999}{1000}/\filename}
        \subcaptionbox%
        {\figALabel}[\subcaptionboxWidthTwo]%
        {\includegraphics[width=\subcaptionboxWidthTwo]{\figA}}
        \subcaptionbox%
        {\figBLabel}[\subcaptionboxWidthTwo]%
        {\includegraphics[width=\subcaptionboxWidthTwo]{\figB}}
        
        \subcaptionbox%
        {\figCLabel}[\subcaptionboxWidthTwo]%
        {\includegraphics[width=\subcaptionboxWidthTwo]{\figC}}
        \subcaptionbox%
        {\figDLabel}[\subcaptionboxWidthTwo]%
        {\includegraphics[width=\subcaptionboxWidthTwo]{\figD}}
    \end{center}
    \caption{Comparison of logarithmic \metricIkTCVsBkTCLong and \metricAliveNodeCountLong for \RQIncrementality}
    \label{fig:eval-rq-incr-iktc-vs-bktc}
\end{figure}

\paragraph{Discussion}

The plots in \Cref{fig:eval-rq-incr-scope-by-degree-vs-node-count} allow us to investigate in how far the number of alive nodes and the node degree correlate with the \metricDegreeNormalizedScopeLong of a \TC run.
In all plots, the \metricDegreeNormalizedScopeLong tends to decrease with increasing \metricAliveNodeCountLong.
This is a very positive result because it shows that the scope, when normalized by the node out-degree, is even positively affected by the size of the topology.

Finally, we discuss the benefit of applying \iktc in comparison to \bktc using  \Cref{fig:eval-rq-incr-iktc-vs-bktc}.
Pleasantly, for all configurations, the fraction of \LSMs required for performing \iktc in comparison to \bktc decreases significantly.
In all four configurations, \iktc outperformed \bktc in between \SI{89}{\percent} and \SI{99}{\percent} of all \TC runs per configuration, \idest, between \numprint{106} (for \cfgLargeMedium) and \numprint{118} (for \cfgLargeSparse) out of \evalParMeasurementsPerRun \TC runs.
%
%
%
To sum up, we may answer \RQIncrementality as follows.
\begin{itemize}
\item \iktc is theoretically unbounded and, therefore, not strictly incremental.
Still, from a practical point of view, we showed that the cost in terms of \LSMs is bounded and mainly determined by the average node out-degree, which is a local property.
Thus, \iktc behaves incrementally in the considered scenarios.
\item
\iktc was able to clearly outperform \bktc in our simulation study.
The benefit of invoking \iktc even increases with increasing number of alive nodes, \idest, with increasing topology size.
\end{itemize}
\subsection{\RQPerformanceLong}
\label{sec:eval-rq-performance}

\paragraph{Metrics}
To assess the performance of the integration of \eMoflon and the \Simonstrator, we measure and compare the CPU time of executing the \GT-based implementation (\idest, the \CE handling, the \TC algorithm, the consistency checks) and the remaining simulation.

\paragraph{Results}

\Cref{tab:eval-rq-performance-overview} provides an overview of the execution time measurements for each configuration, averaged over all \evalParSeedCount runs per configuration.
We list 
\begin{inparaenum}
\item 
the time for the algorithmic part (\emph{GT algo}), \idest, handling \CEs and invoking \TC, 
\item 
the time required for invoking the (optional) consistency checks (\emph{GT checks}), which were active during the whole simulation study to increase the confidence in the correctness of our implementation (see \Cref{sec:eval-rq-correctness}), and
\item 
the time for performing the remainder of the simulation (\emph{Sim.}).
\end{inparaenum}
The configurations are sorted increasinly by size and density of the initial topology.
\begin{table}[htbp]
    \centering
    \caption{Quantitative comparison of execution time for \RQPerformance. $\textbf{n+m}$: Topology size.}
    \label{tab:eval-rq-performance-overview}
\begin{tabular}{crrrrrr}
    \toprule
    \multicolumn{1}{c}{\textbf{ID}} & 
    \multicolumn{1}{c}{\textbf{n+m}} & 
    \multicolumn{1}{p{1.5cm}}{\centering \textbf{GT Algo [min]}} & 
    \multicolumn{1}{p{1.6cm}}{\centering \textbf{GT Checks [min]}} & 
    \multicolumn{1}{p{1.3cm}}{\centering \textbf{Sim. [min]}} &
    \multicolumn{1}{p{1.5cm}}{\centering \textbf{Total [min]}} & 
    \multicolumn{1}{p{1.5cm}}{\centering \textbf{GT Algo / Total [\%]}} \\
    \midrule
    \cfgSmallSparse & \numprint{912} & 0.01  & 0.00  & 0.76  & 0.77  & 0.9 \\
    \cfgSmallMedium & \numprint{1751} & 0.02  & 0.01  & 0.76  & 0.79  & 2.7 \\
    \cfgSmallDense & \numprint{5261} & 0.69  & 0.19  & 0.93  & 1.81  & 38.2 \\
    \midrule
    \cfgLargeSparse & \numprint{13683} & 0.13  & 0.08  & 96.98 & 97.18 & 0.1 \\
    \cfgLargeMedium & \numprint{23092} & 0.56  & 0.33  & 125.12 & 126.00 & 0.4 \\
    \cfgLargeDense & \numprint{49072} & 7.46  & 2.83  & 170.52 & 180.81 & 4.1 \\
    \bottomrule
\end{tabular}%
\end{table}%

\paragraph{Discussion}
The numbers in \Cref{tab:eval-rq-performance-overview} are impressive:
For three out of six configurations, the overhead of invoking the GT-based implementation is below \SI{1}{\percent} (\cfgSmallSparse, \cfgLargeSparse, \cfgSmallMedium).
In further two out of six configurations, the overhead is below \SI{5}{\percent} (\cfgSmallMedium, \cfgLargeDense).
Only in the densest configuration \cfgSmallDense with an initial average out-degree of \approximately 50, the overhead is remarkably high (\SI{38}{\percent}), which is acceptable because a configuration of such density reflects a border case.
When comparing the execution time of the algorithm and the consistency checking part of the \GT implementation, we recognize that between approx. \SI{25}{\percent} and \SI{33}{\percent} of the execution time of the \GT implementation is used for checking consistency.
For this reason, it may be reasonable to disable consistency checking as soon as the developer has gained sufficient confidence in her implementation of the \GT algorithm to gain even more performance.
Still, we experienced that checking consistency continuously during the simulation is a valuable feature because it aids in identifying even subtle bugs due to the huge number of checked topologies.
The fact that the execution time of the \GT algorithm for the largest configuration (\idest, \cfgLargeDense) makes up only \SI{4}{\percent} of the total execution time shows that our implementation allows simulations even of large topologies and that it is definitely not the bottleneck for such large scenarios.
To sum up, we may answer \RQPerformance as follows.
\begin{itemize}
\item
The fraction of CPU time that is requird for performing \TC and \CE greatly varies depending on the configuration.
For five out of six configurations, the fraction is below \SI{5}{\percent}.
\item
The execution time of \iktc is definitely reasonable for the considered topologies.
\end{itemize}

\subsection{\RQGeneralizabilityLong}
\label{sec:eval-rq-generalizability}

In the following, we qualitatively discuss the general applicability of the proposed approach.

\subsubsection{
Is meta-modeling suitable to specify all relevant properties of \WSN topologies?
}

Based on our experience, we claim that meta-modeling does \emph{not} impose any restrictions on the complexity of the specified topologies.
Meta-modeling concepts that are available in both theory and practice are, \eg, arbitrary attributes of nodes, links, and whole topologies (single-, set-, or list-valued, potentially using custom datatypes) as well as (multiple) inheritance.
More advanced algorithms that also take overlay constraints into account~\cite{KSSVMS14,SKSVSM15,SPSBM16} would require to model paths and constraints on paths, which is also easily possible using meta-modeling.

\subsubsection{
Are \GT rules suitable to specify \CEs and \TC operations?
}

We believe that \GT is suitable to model \CEs and \TC operations in general because graph-based models and \GT rules are natural means to represent the structure and modification of topologies.

\subsubsection{
Are graph constraints and the proposed constructive approach suitable to specify and ensure the preservation of consistency constraints and optimization goals of typical \TC algorithms?
}


Graph constraints as first introduced in \cite{HW95} and adopted in this paper are suitable to specify required (in case of positive graph constraints) and forbidden (in case of negative graph constraints) substructures in topologies of finite size but, thanks to \cite{DV14}, with arbitrary complex attribute constraints.
Still, our running example, \ktc, reveals that certain consistency constraints and optimization goals cannot be expressed using this type of graph constraints:
Physical, weak and strong connectivity are not expressible because of the required constraints over paths of links of arbitrary length.
Additional typical constraints are, \eg, that the output topology should be $k$-connected, a directed acyclic graph (DAG) or even a tree/forest.
In \cite{HabelRadke2010,Radke2010}, Habel and Radke present $HR^{*}$ constraints, a new type of graph constraints that allow to express path-related properties.
Fortunately, the constructive approach is applicable to $HR^{*}$ constraints in principle;
in future work, we will investigate the applicability of $HR^{*}$ constraints to specify, \eg, connectivity of topologies.
An additional feature of meta-models that was not required in this case study is type inheritance.
In \cite{Taentzer2005}, the authors describe how the constructive approach may be applied to meta-models with edge and node type inheritance.
The propose to flatten the original meta-model into an equivalent meta-model without type inheritance based on which the original constructive approach of \cite{HW95,DV14} may be applied.

From our experience, we may say that at least the optimization goals of a large class of \TC algorithms may be specified using graph constraints.
These \TC algorithms operate on local knowledge, \eg, only of neighbors that are at most two hops away, so-called local algorithms~\cite{JRS03}.
This restriction is necessary because sensor nodes tend to have limited working memory (to store the local topology) and because acquiring the neighborhood topology of a node is the more expensive in terms of bandwidth and execution the larger the required topology becomes~\cite{SPSBM16}.
Therefore, we may answer the question as follows:
Graph constraints are suitable to specify the optimization goals of a huge class of \TC algorithms, while crucial consistency constraints currently need to be proved manually.
In total, this is still a considerable advantage over the current state of the art, which only allows to verify or test that a \TC algorithm fulfills the desired consistency constraints and optimization goals a posteriori.

\subsubsection{Is a \TC developer able to apply the described methodology?}

If we want to establish a new development methodology for \TC algorithms, it is crucial that the target audience, \idest, the developers of \TC algorithms, is capable of applying it.

The first step of our methodology is to capture the relevant properties (\idest, classes with their attributes and associations) of the topologies that shall be processed by the \TC algorithm (see \Cref{sec:metamodeling}).
To successfully perform this step, a developer needs to have a working background in meta-modeling, which is realistic as several industrial success stories prove~\cite{Hermann2013}.

The second step of our methodology is the specification of optimization goals and consistency constraints in terms of graph constraints (see \Cref{sec:constraints}).
While this is certainly a non-trivial task, the adoption of a constraint specification languages (\eg, the Object Constraint Language~\cite{Warmer1998}) in practice indicates that this step should also be feasible for a developer.
In \Cref{sec:connectivity}, the constraint-based formulation of \ktc was used to prove that the output topology is connected whenever the input topology fulfills weak consistency.
A developer may need to prove this and similar properties at this step for other \TC algorithms.
However, when analyzing the \TC literature, we recognize that developers of \TC algorithms already perform this step today:
They prove properties such as planarity or connectivity of the topology, boundedness of node degree or termination based on their formal representation of the proposed \TC algorithm (\eg, in \cite{SWBM12,XH12}).
For this reason, we believe that it is reasonable for a \TC developer to succeed in proving properties using a graph constraint-based specification of the \TC algorithm.

The third step of our methodology is to describe relevant \CEs and \TC operations using \GT rules (see \Cref{sec:gratra}).
As described in~\cite{Hermann2013}, it is also a realistic assumption that developers acquire these skills.
Additionally, the \GT rules presented in \Cref{sec:gratra} will be reusable, maybe with modifications, in a large number of cases.
For instance, the \nodeRemovalRuleLong and the \linkRemovalRuleLong may probably be re-used out-of-the-box, while the \nodeAdditionRuleLong, the \linkAdditionRuleLong, and the \weightModificationRuleLong may have to be updated to reflect the particular node and link attributes of the considered class of topologies.
These considerations are backed by observations in the literature that \TC algorithms (not only in the \WSN domain) are typically composed of elementary topology modification operations (\eg, link removal)~\cite{Stein2016}.

The fourth step of our methodology is the refinement of the specified \GT rules to preserve the specified graph constraints (see \Cref{sec:refinement-methodology}).
While the concrete applications of the refinement algorithm as described in \Cref{sec:refinement-tc-inact,sec:refinement-tc-act,sec:refinement-ce-inact,sec:refinement-ce-inact} were carried out manually for this paper, the rule refinement algorithm may in principle be automated.
As discussed in \Cref{sec:eval-rq-correctness}, performing this step manually is error-prone and tedious.
Therefore, a powerful tool support is required to aid the developer in analyzing and filtering the resulting application conditions to come up with a set of rules that is readable and efficient to evaluate.

The fifth step of our methodology is the derivation of restoration operations for unrestrictable \GT rules that were refined with additional application conditions during the refinement step (see \Cref{sec:enforcing-applicability}).
In this paper, we do not provide a precise algorithm for automatically deriving restoration operations from a given application condition.
Yet, such an algorithm may exist at least for graph constraints as used in our approach.
In this case, tool support for automating should be developed as for the refinement step.

The sixth step of our methodology is the simulation-based evaluation of the specified \TC algorithm.
The existing tool support, consisting of a model-based adapter between \eMoflon and the \Simonstrator, may be re-used if the \Simonstrator is deemed appropriate by the developer to carry out the evaluation.
The only required modification is an update of the synchronization objects in the \Simonstrator and the topology model.
If a different network simulator or a different \GT tool is to be used to carry out the evaluation, more development effort is required.
Still, the existing tool support may serve as a template for implementing the desired adapter between the \GT tool and the network simulator of the developer's choice.

In summary, the meta-modeling part, the specification of \GT rules and graph constraints, and the adaptation of the integration between the \GT tool and the simulator can realistically be carried out by a developer of \TC algorithms.
The refinement algorithm and the derivation of restoration operations, which form the core parts of our methodology, still require advanced tooling to support the \TC developer.

\subsection{Discussion of Threats to Validity and Summary of Evaluation}
\label{sec:eval-threats}
We considered threats to \emph{external validity}~\cite{Cook1979}, \idest, the general applicability of our approach, as the major threat to validity.
While this paper discussed the construction of one particular \TC algorithm in detail due to space limitations, we mitigated this threat by dedicating \RQGeneralizability to it (\Cref{sec:eval-rq-generalizability}).
The performance evaluation in \Cref{sec:eval-rq-performance} mitigates this threat from a practical point of view:
Our proof-of-concept implementation of the integration of \eMoflon and the \Simonstrator is applicable to \WSN topologies of realistic size and does not impose an unreasonable overhead.

We considered threats to \emph{internal validity}, \idest, the soundness of the presented results, when discussing \RQCorrectnessLong.
We analyzed in detail the possible sources of specification and implementation problems (\eg, the manual creation of restoration operations) and applied measures to minimize or mitigate them (\eg, by continuously checking consistency during the numerous iterations in the simulation study).

In \Cref{tab:answers-to-rqs}, we concisely summarize the results of this evaluation.

\begin{table}
    \def\rowskip{\\[6pt]}
\begin{center}
\caption{Summary of answers to the research questions}
\label{tab:answers-to-rqs}
\begin{tabular}{p{2.4cm}cp{9cm}}
    
\toprule

\textbf{RQ} & \textbf{Sec.} & \textbf{Result}\\

\midrule

\textbf{\RQCorrectnessLong} & \ref{sec:eval-rq-correctness} & 
We carefully analyzed potential sources of errors during the steps of the methodology and described the applied mitigation strategies.
All of the \numprint{21 600} consistency checks during the simulation study were successful, which fosters the confidence in the correctness of our implementation.
\rowskip
\textbf{\RQIncrementalityLong} & \ref{sec:eval-rq-incrementality} & 
The execution of \CE handlers is unbounded and, therefore, the developed algorithm is theoretically not strictly incremental.
However, the simulation showed that the scope of \CEs is limited by a  constant that is independent of the number of alive nodes.
\rowskip
\textbf{\RQPerformanceLong}& \ref{sec:eval-rq-performance} & 
In the smallest and densest configuration, \cfgSmallDense, \GT consumed \SI{38}{\percent} of the execution time.
In all other five of six configurations, \GT consumed less than \SI{5}{\percent} of the execution time.
\rowskip
\textbf{\RQGeneralizabilityLong} & \ref{sec:eval-rq-generalizability} & 
Meta-modeling and \GT rules are suitable to specify arbitrary \WSN topologies and their modifications by \CEs and \TC algorithms.
Typical optimization goals of the large class of local \TC algorithms can naturally be modeled using graph constraints;
certain consistency constraints such as connectivity may not be expressed using graph constraints.
A developer of \TC algorithms should be capable of applying our methodology in practice if the existing tooling is extended appropriately.
\rowskip
\bottomrule

\end{tabular}
\end{center}
\end{table}

\section{Related Work}
\label{sec:related-work}

In this section, we survey related work from the fields of modeling and checking consistency properties and formal methods as well as model-based software engineering in the communication network engineering domain.

\paragraph{Specifying, checking and enforcing consistency properties}
Model checking \cite{RSV04} is an analysis technique used to
verify particular properties of a system a posteriori, \idest, after a specification of the desired properties and an implementation of the regarded system has been constructed.
If an abstract symbolic problem specification~\cite{Mil00} is missing, model checking tools are often limited to a finite model size.
The approach in this paper constructively integrates constraints a priori at design time into an already executable formal specification of a TC algorithm rather than checking whether specified properties are fulfilled a posteriori (either at design time or at runtime). 
This correct-by-construction approach (\eg,~\cite{BBBCJNS11,DB05,BFMSFSW05}) has the positive side-effect that we can prove that topologies of any size are consistent without any needs to find a finite abstraction of the infinite set of all possible topologies.

In \cite{HW95}, graphical consistency constraints, which express the fact that particular combinations of nodes and edges should be present in or absent from a graph, are translated into application conditions of \GT{} rules.
This technique has been generalized to the framework of HLR categories~\cite{EEPT06} and extended to cope with attributes~\cite{DV14}.
The basic idea is to translate consistency conditions, characterizing consistent graphs, into post-conditions of \GT{} rules, which may then be transformed into application conditions of the rules.
This paper applies and extends this generic methodology for the practical and 
complex application scenario of developing \TC{} algorithms.
Our extension consists of an additional, domain-specific step that translates some of the derived application conditions into pre- and post-processing rules that ensure that the corresponding application conditions are always fulfilled and may, thus, be removed from the rule.
These new domain-specific relaxation and re-enforcement techniques for constraints and application conditions are needed to bypass a critical---and from a practical point of view often unrealistic assumption---made in~\cite{HW95}:
A GT-rule with additional application conditions preserves some graph constraints if the application conditions may block any harmful execution of this rule, and if the processed graph never contains any constraint violations.
Both assumptions are violated in our scenario, where \CEs{} may add, modify, and remove topology elements without any restrictions.

Consistency properties may be expressed using a variety of formal frameworks.
In this paper, we use graphical consistency constraints~\cite{HW95}, which are suitable to specify properties based on a fixed number of graph elements, \eg, the triangle of links in \ktc{}.
However, certain consistency properties, such as connectivity, may not be expressed in this framework.
For such properties, more expressive formal frameworks such as (counting) second-order monadic logic~\cite{HW95, Rad13} or nested graph constraints~\cite{HP05,Fli15} exist.
In this paper, connectivity of topologies is the only consistency property that cannot be expressed using graphical consistency constraints;
therefore, we manually proved that this property is fulfilled.

In~\cite{HHS02}, model transformations are categorized based on the way that their application affects the consistency of a model.
In their categorization, a transformation may basically either preserve or establish consistency.
Consistency is preserved by a transformation if models that were consistent before the transformation are consistent afterwards as well.
Consistency is established by a transformation if models that might not be consistent before the transformation are consistent afterwards.
In this paper, preservation of weak consistency is achieved because we assume that the initial network topology is weakly consistent and because the refined \TC{} and \CE{} rules preserve the \TC{} graph constraints that imply weak consistency.

A prominent alternative language for specifying constraints is the Object Constraint Language (OCL)~\cite{Warmer2003} and its visual extension \textsf{VisualOCL}~\cite{Bottoni2001}.
Current research focuses on translating OCL constraints into graph constraints~\cite{Radke2015,Bergmann2014}.
In Story Decision Diagrams~\cite{Klein2007}, a complex (graph) constraint is formalized using a directed acyclic graph data structure that resembles Story Diagrams, which have been employed in this paper to define the control flow of \TC operations.
While the graph constraints considered in this paper could easily be expressed using Story Decisions Diagrams, to the best of our knowledge, no work exists that allows to transform postconditions formulated as Story Decision Diagram into preconditions of \GT rules.

\paragraph{Formal methods in communication network engineering}
The need for supporting or even driving the development of novel network communication protocols using formal methods has been recognized already in the 1980s~\cite{BS80}, and even today, formal methods are continuously gaining importance in the networking engineering community~\cite{QH15}.

Formal analysis of consistency properties of network protocols has already revealed special cases in which the implemented algorithms violate crucial (topological) constraints~\cite{Zave2008,Zave2012,LMW12,HI13}.
In~\cite{Zave2008}, the SPIN model checker~\cite{Hol04} is used to analyze the Session Initiation Protocol (SIP) for specification flaws.
In~\cite{Zave2012}, the lightweight modeling language Alloy~\cite{Jac02} is used to find specification flaws in the famous peer-to-peer protocol Chord~\cite{SMKKB01}.
In~\cite{HI13}, statistical model checking is applied to analyze AODV and DYMO, two famous routing protocols for wireless mesh networks.
In~\cite{KMH08}, model checking is applied to detect bugs and to point at their causes in the \TC{} algorithm LMST, leading to an improved implementation thereof.
In contrast to the aforementioned contributions, this paper focuses on preserving consistency properties in a constructive way rather than analyzing existing protocols a posteriori.

In~\cite{Papadopoulos2016}, a three-stage development process of \WSN algorithms---consisting of formalization, simulation evaluation, and testbed evaluation---is emphasized as an important means of verifying \WSN algorithms.
While the authors focus on the role of testbed simulations (compared to simulation), the focus of this paper is to reduce the gap between the formalization and the simulation part.

In~\cite{Nareyek2001}, Nareyek highlights that avoiding consistency checks at runtime by constructively integrating constraints into the preconditions of rules of decision agents may also increase the performance.

\paragraph{Visual languages and models in the context of telecommunications}

Visual and model-based languages techniques have shown to be suitable to describe~\cite{TGM00} and to construct~\cite{JRSST09} optimization algorithms for communication networks.
Message sequence charts \cite{DRR04}, for instance, have already been used to document protocol standards of communication systems; 
the same is true for object-oriented modeling languages like ROOM~\cite{SGW94}, which was one of the first (truly) object-oriented modeling tools successfully used by the telecommunication industry.
Some examples of distributed system visual programming languages are G-Net~\cite{CDC93,NKMD96}, Prograph~\cite{CGP89}, DVispatch~\cite{MH98}, which rely on different programming paradigms and formalisms (such as Petri-Nets~\cite{DA94} and Flow-Graphs~\cite{GW06}).
In~\cite{KSSVMS14}, an approach for the \GT{}-based rapid prototyping of variants of \TC{} algorithms is proposed.
In a case study, the \GT{} tool eMoflon~\cite{LAS14} is used to derive two batch variants---\emph{l-kTC} and \emph{g-kTC}---of the \TC{} algorithm \ktc{}~\cite{SWBM12,SZS15}, which filter the selected links of the marking step of \TC{} based on additional, application-specific constraints.
In~\cite{SKSVSM15}, the batch algorithms l-kTC and g-kTC are evaluated in detail in a simulation study using the \textsf{Simonstrator} platform~\cite{RSRS15}.
In this paper, we use the same tooling as in \cite{KSSVMS14,SKSVSM15} for the simulation study, but we focus on deriving incremental \TC{} algorithms that are correct by construction rather than introducing complex additional consistency constraints.

In the following, we list further approaches to the graph-based and model-based construction of \TC algorithms.
In~\cite{BMPPT15}, an extension to \textsf{Agilla}~\cite{FRL09}, a model-driven middleware for developing \WSN{} applications, is proposed that introduces an instruction set for reading the current remaining energy of sensor nodes.
In~\cite{RDBPP15}, a model-driven development process is proposed that allows to design \TC{} algorithms based on meta-models and state charts and then to generate platform-specific code from this representation.
The focus of their approach lies on a complete tool chain from a platform-independent model to the platform-specific, compilable code;
correctness properties of the developed algorithms have to be ensured manually.
In~\cite{WBSO07}, a model-driven tool for developing \WSN{} applications called \textsf{Matilda} is presented.
Based on a custom UML profile~\cite{OMG06} for sensor networks, the authors describe the domain structure using class diagrams, concrete network instances using instance diagrams, and the behavior of the sensor nodes using sequence diagrams.
The authors propose to extend their development process with techniques for checking correctness.
In~\cite{TSFH15}, a model-based approach for developing \WSN{} algorithms in a \emph{stepwise} manner is proposed, which separates the concerns of designing the \WSN{} application logic from the handling of network-related concerns.
The focus of their approach is on decomposing the development process of \WSN{} applications.
In~\cite{BKLLXZ04}, the modeling and simulation framework \textsf{VisualSense} is presented, which operates on top of the modeling framework for heterogeneous devices \textsf{Ptolemy II}~\cite{EJLJXLNSY03}.
The platform provides a visual modeling language for specifying the data flow while at the same time abstracting from the concrete topology of the \WSN{}.
In~\cite{SFKS08}, the model-based development environment \textsf{ScatterClipse} is presented, which has been developed to provide an integrated development workflow for sensor applications that run on the \textsf{ScatterWeb}~\cite{SLRWV05} \WSN{} platform.
\textsf{ScatterClipse} adopts the Model-Driven Software Development workflow of the OMG~\cite{OMG14} and allows the user to specify rules that validate the syntactic and semantic correctness of the developed models.
However, no means for specifying consistency constraints of the network topology exists.
In~\cite{LC02}, the virtual-machine environment \textsf{Maté} is presented, which provides a high-level application programming interface for \WSN algorithms as well as distribution and replication mechanisms that allow for an efficient deployment of \WSN algorithms into a hardware testbed.
The focus of~\cite{LC02} lies on the efficient testing of \WSN algorithms in a virtual environment.
This list of approaches for the model-driven development of \TC{} algorithms shows that most frameworks are centered around 
\begin{inparaenum}
\item providing a thorough development process or 
\item separating the concerns of the domain engineer from those of the networking engineer
\end{inparaenum}
In contrast to this paper, none of these approaches considers the provable correctness of the developed algorithms as a core concern.

In~\cite{SW10}, MechatronicUML is presented as a comprehensive framework for specifying and formally analyzing mechatronic systems.
The structure of systems is described by means of component diagrams similar to UML, the behavior of systems is described by means of state charts, and architectural reconfigurations are described by means of \GT rules using the Fujaba real-time tool suite~\cite{BGHST05}, which is a branch of the Fujaba \GT tool~\cite{NNZ00}.
Based on these formal foundations, the authors use the verification techniques model checking~\cite{Ren04} and proof by induction~\cite{BBGKS06} to verify safety properties of the modeled systems.
While the approach proposed in~\cite{SW10} presents a mature and powerful environment for developing and verifying mechatronic systems, safety properties are, again, checked a posteriori, whereas in our approach, safety properties are integrated into the development process.

A series of publications by Scheideler and others use graph theory for constructing provably correct and efficient algorithms for adapting topologies of \WSNs{} and other applications such as distributed search overlays (see, for instance, \cite{JRS03,JRSST09,DPV11,AHLSS13,JRSST14}).
While their approach ensures that consistency properties and optimization goals are fulfilled by the derived algorithm, the structured transformation into an executable implementation is not so obvious.
In this paper, in contrast, we formalize the possible modifications of the network topology as \GT rules and we formalize the consistency properties and optimization goals as graph constraints.
On the one hand, this allows to prove the correctness of the refined \GT{} rules in terms of the specified graph constraints.
On the other hand, the derived \GT{} rules can be directly used to implement and evaluate the \TC{} algorithm.

\section{Conclusion}
\label{sec:conclusion}

This section concludes the paper and highlights possible directions of future research.

\subsection{Summary}
In this paper, we describe a model-based methodology for developing \TC algorithms, which select a subset of the links of a communication network that fulfills specified consistency properties and optimization goals.
We use \GT rules to describe possible \TC operations, which are applied by a \TC algorithm to mark selected links, as well as context events, which specify modifications of the topology by the environment.
We specify consistency and optimization properties of output topologies of \TC algorithms using graph constraints, which are then used to refine the \TC and context event rules in a way that applying the rules preserves the graph constraints.
We apply the constructive approach to the \TC algorithm \ktc{}.
The discussion of a simulation-based evaluation shows that the constructed \TC algorithms may be directly evaluated in a number of realistic simulation scenarios.
We showed that 
\begin{inparaenum}
\item 
the implementation passed all \numprint{21 600} consistency checks, 
\item 
the derived algorithm operates incrementally in practice,
\item 
the performance of the developed integration between the \GT tool \eMoflon and the network simulation platform \Simonstrator is very satisfactory, and
\item
the proposed methodology is applicable to a large class of \TC algorithms.
\end{inparaenum}

\subsection{Outlook}

The results of this paper open up various possible lines of research.
A first line of research addresses the partial automation of our constructive methodology:
The refinement procedure described and applied in \Cref{sec:rule-refinement} is formally founded and the resulting rules are \enquote{correct by construction}.
Still, performing the refinement procedure manually is error-prone and should be automated as far as possible.
We showed that the number of additional application conditions may be reduced considerably if we assume that particular constraints (\eg, the structural constraints) hold prior to each execution of a \TC rule.
It remains to investigate if it is possible to automatically derive repair operations from a set of application conditions.

A second line of research addresses the simulation study.
Simulating a distributed execution of \GT based on a to-be-designed \GT protocol is desirable to be able to investigate concurrency and timing issues already inside the simulator.

The specification of the \TC algorithm in this paper is suitable to operate not only in a centralized way on the whole topology, but also in a distributed way on the local view of each node, which typically consists of all nodes that are two or three hops away from the node.
Finally, we have to overcome the limitation that the \TC algorithm is executed offline, \idest, outside the running simulation, which is currently halted during the execution of the \TC algorithm.
A truly distributed, asynchronous implementation will open up challenging research questions, e.g., about conflicting rule applications and asynchronicity.

A third line of research addresses the evaluation of the developed \TC algorithms in a hardware testbed.
This step is especially challenging because wireless nodes are usually programmed in C/C++ and typically have strictly limited resources in terms of code size and memory.
The \GT{} tool eMoflon, which we applied for the simulation study, is currently only able to generate Java code.
We are currently performing first experiments in generating C code for the \WSN{} platform Contiki~\cite{DGV04}.

\section*{Acknowledgments}
This work has been funded by the German Research Foundation (DFG) as part of projects A01 and C01 within the Collaborative Research Center (CRC) 1053 -- MAKI.
The authors would like to thank Lukas Neumann for his considerable support during the creation of this paper.

\section*{References}
\bibliographystyle{elsarticle-num-names}      
\bibliography{jvlc2015}   

\end{document}